%% file: main.tex
\documentclass[sigconf,nonacm]{acmart} 
\usepackage[utf8]{inputenc}

\pdfoutput=1

\usepackage{amsthm, amsmath}
\usepackage{mathtools}
\usepackage{enumitem}
\usepackage{float}

\usepackage{tikz}
\usetikzlibrary{arrows,shapes,trees}
\usetikzlibrary{arrows.meta}

\AtBeginDocument{}

\DeclareMathOperator{\poly}{poly}

\newcommand\blfootnote[1]{%
  \begingroup
  \renewcommand\thefootnote{}\footnote{#1}%
  \addtocounter{footnote}{-1}%
  \endgroup
}

\begin{document}

\title{Partitioning Hypergraphs is Hard: Models, Inapproximability, and Applications}


\author{P\'al Andr\'as Papp}
\orcid{0009-0005-6667-802X}
\email{pal.andras.papp@huawei.com}
\affiliation{%
  \department{Computing Systems Lab}
  \institution{Huawei Zurich Research Center}
  \city{Zurich}
  \country{Switzerland}
}

\author{Georg Anegg}
\authornote{Author is currently affiliated with ETH Zurich, Switzerland. \vspace{3pt}}
\email{ganegg@ethz.ch}
\orcid{0000-0002-5730-5812}
\affiliation{%
  \department{Computing Systems Lab}
  \institution{Huawei Zurich Research Center}
  \city{Zurich}
  \country{Switzerland}
}

\author{Albert-Jan N. Yzelman}
\email{albertjan.yzelman@huawei.com}
\orcid{0000-0001-8842-3689}
\affiliation{%
  \department{Computing Systems Lab}
  \institution{Huawei Zurich Research Center}
  \city{Zurich}
  \country{Switzerland}
}

\begin{abstract}
	We study the balanced $k$-way hypergraph partitioning problem, with a special focus on its practical applications to manycore scheduling. Given a hypergraph on $n$ nodes, our goal is to partition the node set into $k$ parts of size at most $(1+\epsilon)\cdot \frac{n}{k}$ each, while minimizing the cost of the partitioning, defined as the number of cut hyperedges, possibly also weighted by the number of partitions they intersect. We show that this problem cannot be approximated to within a $n^{1/\poly \log\log n}$ factor of the optimal solution in polynomial time if the Exponential Time Hypothesis holds, even for hypergraphs of maximal degree $2$. We also study the hardness of the partitioning problem from a parameterized complexity perspective, and in the more general case when we have multiple balance constraints.
	
	Furthermore, we consider two extensions of the partitioning problem that are motivated from practical considerations. Firstly, we introduce the concept of hyperDAGs to model precedence-constrained computations as hypergraphs, and we analyze the adaptation of the balanced partitioning problem to this case. Secondly, we study the hierarchical partitioning problem to model hierarchical NUMA (non-uniform memory access) effects in modern computer architectures, and we show that ignoring this hierarchical aspect of the communication cost can yield significantly weaker solutions.
\end{abstract}



\maketitle

\section{Introduction}

One of the most fundamental graph problems is to partition the node set of a graph into $k$ parts of similar size, while minimizing the number of cut edges. Recently, the focus of this balanced $k$-way partitioning problem has shifted from graphs to hypergraphs, where a hyperedge can contain not only two, but an arbitrary number of nodes.

\blfootnote{\copyright P\'al Andr\'as Papp, Georg Anegg, Albert-Jan N. Yzelman, 2023. This is the author's full version of the work, posted here for personal use. Not for redistribution. The definitive version (extended abstract) was published in the 35th ACM Symposium on Parallelism in Algorithms and Architectures (SPAA 2023), https://doi.org/10.1145/3558481.3591087.}

A prominent application of this problem is finding the most efficient way to execute a complex computation in a parallel manner on $k$ processors. That is, we can use each node $v$ of a hypergraph to represent a specific step of a computation (e.g.\ a function call, or in a more fine-grained model, a single operation). A hyperedge $e$, on the other hand, represents a unit of data (e.g.\ an input or output variable) in this computation that is shared by a given subset of nodes; hence it would be desirable to execute the nodes of $e$ on the same processor in order to avoid data movement. In such a setting, the partitioning of the nodes into $k$ parts can be interpreted as an assignment of the computational steps to the $k$ available processors, and the balance constraint on the size of the parts ensures that the workload is indeed sufficiently parallellized. On the other hand, our objective is to cut as few hyperedges as possible, which corresponds to minimizing the total \textit{communication cost} between the processors; this is indeed known to be the bottleneck in many real-world computations.

The main advantage of this hypergraph model is that it allows us to capture the communication cost accurately: if the hyperedge $e$ (representing the value of some variable) intersects $\lambda_e$ out of the $k$ parts, then it takes exactly $(\lambda_e-1)$ data transfer operations to move this value from one of these processors (where it is initially stored) to all others (where it is needed). In contrast to this, if we try to model this connection between the same subset of nodes $e$ as a simple graph, then this will always result in an over- or underestimation of the real cost in some cases \cite{H98, HK00}.

In this paper, we present new hardness results for the balanced hypergraph partitioning problem, extending a $n^{1/ \poly \log\log n}$-factor inapproximability bound which was only known for the similar bisection problem before. More importantly, we show that this hardness result already holds for hypergraphs of very small degree, thus pointing out the crucial role of heuristics in practice.

Besides this general hardness result, we also study two further aspects of the partitioning problem which are both strongly motivated by practical considerations. Firstly, we study the setting where we also have dependencies (precedence relations) between the different computational steps in our hypergraph, e.g.\ when modelling the steps of an entire algorithm. We introduce the notion of \textit{hyperDAGs} to capture the communication cost of a parallel execution in this case: this essentially combines the concept of computational DAGs with the more accurate hyperedge-based modeling of communication costs. We study the partitioning problem on hyperDAGs specifically.

Finally, we also define the hierarchical variant of the hypergraph partitioning problem. This deals with an oversimplification from our original model, namely, that the communication cost is assumed to be uniform between any pair of processors. The majority of modern computing architectures, however, are organized into a hierarchical tree structure: several cores are connected to the same CPU, several CPUs to the same RAM, and then we possibly have multiple such units connected on a network level. Due to this, these architectures exhibit highly non-uniform communication costs between different pairs of processing units: transferring data between two cores on the same processor only induces a small cost, whereas transferring it through multiple levels of the hierarchy is much more time-consuming. As such, for a realistic model of communication cost, it is essential to also incorporate this hierarchical structure into the partitioning problem.

Our main contributions are as follows:
\vspace{2pt}
\begin{itemize}[leftmargin=20pt]
    \setlength{\itemsep}{3pt}
    \setlength{\parskip}{2.5pt}
    \item As our main result, we show that assuming the Exponential Time Hypothesis (ETH), there is no polynomial-time approximation algorithm of factor less than $n^{1/\poly \log\log n}$ to the $\epsilon$-balanced hypergraph partitioning problem for any $k \geq 2$ or $\epsilon \geq 0$. Furthermore, this hardness result already applies for practically relevant cases: it already holds if our inputs are restricted to hyperDAGs of node degree at most $2$.
    \item We define and analyze a special class of hypergraphs (hyperDAGs) that provide a more accurate model of capturing I/O cost in computations modelled by DAGs. We then study two natural techniques to develop more appropriate balance constraints for hyperDAGs. In case of layer-wise constraints, we show that the best solution cannot be approximated to an $n^{1-\delta}$ factor (for any $\delta>0$). In case of schedule-based constraints, our observations show that a precise measurement of parallelization in hyperDAGs is not viable in practice.
    \item We discuss several hardness results for the natural extension of the partitioning problem where we have multiple independent balance constraints.
    \item Finally, we introduce the hierarchical partitioning problem to obtain a significantly more accurate model of I/O cost in today's computing architectures. We show that ignoring this hierarchical aspect of the cost function can result in significantly weaker solutions.
\end{itemize}

\section{Related work}

Both the graph- and hypergraph partitioning problem is known to be NP-hard already for $k=2$, and for any (non-trivial) $\epsilon \geq 0$ \cite{GJ79}. The problems have a wide range of applications, including parallel computing, VLSI design, and scientific computing \cite{applic1}.

There is a long line of work on approximation algorithms for the case of $\epsilon=0$ and $k=2$, also known as the \textit{bisection problem}, culminating in an $O(\log{n})$-approximation by R\"{a}cke \cite{FKN00, KF06, R08}. The variant of the problem without a balance constraint has also been studied \cite{GH94, SV95}, as well as lower bounds for the case when $k$ is a variable part of the input \cite{AR06, FF12}.

Many further works on approximating the partitioning problem have focused on ($\alpha, \beta$)-bi-criteria approximations of the bisection problem \cite{bicrit1, bicrit2, bicrit3, FF12}, where the cost is at most $\alpha$ times that of the optimal bisection, and every partition has at most $\beta \cdot \frac{n}{k}$ nodes (i.e.\ the strict balance constraint can be violated by a factor $\beta$). However, this is a significantly different concept from approximating the $\epsilon$-balanced partitioning problem, because bi-criteria approximations compare each solution only to the optimal bisection cost, and this optimal bisection cost can be a factor $\Theta(n)$ larger than the optimum for $\epsilon$-balanced partitioning. As such, in applications where our goal is to find an $\epsilon$-balanced solution of low cost, the guarantees of the bi-criteria approach might not be meaningful: even if a low-cost $\epsilon$-balanced solution exists, the bi-criteria approximations may return a solution that only approximates the (possibly much higher) optimal bisection cost.

A hierarchical version of the partitioning problem has also been studied  on simple graphs, mainly through similar bi-criteria approximations where $\beta$ also depends on the height of the hierarchy \cite{HJKS14, RS16}.

In recent years, the attention in partitioning problems has shifted to hypergraphs. The case of hypergraph partitioning without a balance constraint has been analyzed \cite{CC20, CL20}. As for the constrained case, the work of R\"{a}cke, Schwartz and Stotz \cite{BisectionApprox} again focuses on the bisection problem: they present an approximation algorithm of factor $\widetilde{O}(\sqrt{n})$, as well as several lower bounds for approximability, and they also show that tree-based methods (which provide some of the best approximations for graph partitioning) are not viable for hypergraphs.

The closest result to our main theorem (also from \cite{BisectionApprox}) is a similar inapproximability bound of $n^{1/\poly \log\log n}$ for the bisection problem. Our result is different from this in two ways. First, we show this bound for the balanced partitioning problem with $\epsilon > 0$, which is an easier problem than bisection: there is a simple reduction from balanced partitioning to bisection via adding isolated nodes (see Appendix \ref{app:basics}), but the other direction is not straightforward. Second, we show that the bound already holds in hypergraphs (or even hyperDAGs) of very small degree.

We note that some of these related works consider the natural extension of the problem with node or edge weights; our hardness results also carry over to these more general settings.

We also point out that a similar (slightly more general) notion to hyperDAGs has already been discussed in the work of \cite{DOH}, noting that it provides a more accurate model of communication cost in computational DAGs; however, this work does not study the topic (either hyperDAGs or the partitioning problem) from a theoretical perspective.

Finally, due to the wide applicability of hypergraph partitioning, finding efficient and scalable solutions in practice has also been an active area of research for a long time. This includes sophisticated heuristics \cite{BMSSS13, hMetis, Parkway, KaHyPar, more_recent_advances22} as well as optimized exact algorithms \cite{KB20, exact1}.

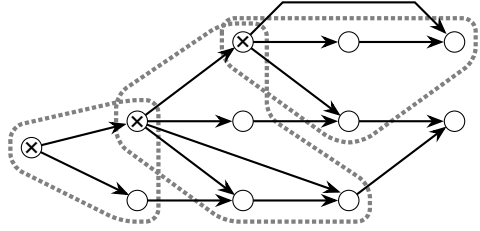
\begin{figure*}[!ht]
    \centering
    \vspace{10pt}
    \input{pics/hyperDAG.tikz}
    \caption{Illustration of converting a computational DAG into a hyperDAG. For simplicity, we only show hyperedges with at least $3$ nodes (generated by the nodes marked with an $\times$).}
    \label{fig:hyperDAG}
\end{figure*}

\section{Preliminaries} \label{sec:prelim}
\subsection{Hypergraphs and partitioning}

A hypergraph $G(V,E)$ consists of a set of nodes $V$ and hyperedges $E \subseteq 2^V$, where $2^V$ denotes the power set of $V$. We denote the number of nodes by $n:=|V|$, the total number of pins by $\rho := \sum_{e \in E} \, |e|$, and the maximal node degree by $\Delta:= \max_{v \in V} \, |\{ e \in E \, | \, v \in e \}|$. We also use $[\ell]$ as shorthand notation for the set of integers $\{ 1, \ldots, \ell \}$.

A $k$-way partitioning $\mathcal{P}$ of $G$ is a disjoint partitioning $P_1, \ldots, P_k$ of the nodes $V$. For a hyperedge $e \in E$, we define $\lambda_e\coloneqq |\{ i \in [k] \mid e \cap P_i \neq \emptyset \}|$ as the number of partitions intersecting $e$, and we say that $e$ is \textit{cut} if $\lambda_e > 1$. There are two popular cost metrics for a $k$-way partitioning $P$: the \textit{cut-net} metric $| \{ e\in E \, | \, \lambda_e>1 \} |$, and the \textit{connectivity} metric $\sum_{e\in E} \, (\lambda_e-1)$. Our hardness results apply to both of these cost metrics (unless one of the metrics is explicitly specified). Note that for the simplest case of $k=2$, the two metrics are identical. In this case, we will also refer to the nodes in $P_1$ and $P_2$ as red and blue nodes for simplicity.

Given a balance constraint parameter $\epsilon>0$, we say that $\mathcal{P}$ is \textit{$\epsilon$-balanced} if for all $i \in [k]$ we have $|P_i| \leq (1+\epsilon) \cdot \frac{n}{k}$. For convenience, this is sometimes relaxed to $|P_i| \leq \lceil (1+\epsilon) \cdot \frac{n}{k} \rceil$ to ensure that a balanced $\mathcal{P}$ always exists. We also implicitly assume $\epsilon<k-1$, i.e.\ the balance constraint ensures that no part $P_i$ can contain the entire set $V$.

\begin{definition}
In the \emph{$\epsilon$-balanced $k$-way hypergraph partitioning problem} (or simply \textit{partitioning problem}), we are given an input hypergraph $G(V,E)$, and our goal is to find an $\epsilon$-balanced partitioning of $V$ with minimal cost (with respect to either the cut-net or the connectivity metric).
In the decision version of the problem, the input also contains an $L \in \mathbb{Z}$, and we need to decide if there is an $\epsilon$-balanced partitioning of cost at most $L$.
\end{definition}

Note that both $k \geq 2$ and $\epsilon \geq 0$ are fixed constants, i.e.\ parameters of the problem and not part of the input. The special case of $k=2,\, \epsilon=0$ is known as the bisection problem.

Some of our hardness results are built on different complexity assumptions (we discuss these in Appendix \ref{app:inapprox} in more detail). Most important among these is the \textit{Exponential Time Hypothesis} (ETH); intuitively, this states that $3$-SAT cannot be solved in subexponential time. We also occasionally use stronger variants of this hypothesis (such as SETH or Gap-ETH).

We also assume some familiarity with the parameterized complexity classes W[1], XP, and para-NP (see Appendix \ref{app:inapprox} for details). Intuitively, a problem with some parameter $L$ is in XP if it can be solved in $n^{f(L)}$ time; it is para-NP-hard if it is already NP-hard for a fixed $L \in O(1)$.

\subsection{HyperDAGs}

General hypergraphs are indeed the appropriate way to model large computations when we can execute the computational steps in arbitrary order; for example, they are often used to model the parallelization of large SpMV (sparse matrix-dense vector) multiplications \cite{KB20}.

However, in other cases, our goal is e.g.\ to model the steps of a complex algorithm, where we also have dependencies between the different computational steps; as such, we clearly cannot execute them in any desired order. Such a setting can be modelled as a directed acyclic graph (DAG), where the nodes again represent specific computational steps (intermediate values to compute), and the directed edges represent precedence relations between these computations: the edge $(u,v)$ implies that the output value of computation $u$ is an input to computation $v$. This computational DAG model has been studied extensively in terms of scheduling, time-memory trade-off and many other perspectives \cite{DAG2proc1, DAGcommcost, RBpebble}.

Note, however, that if we directly apply computational DAGs to capture communication costs, we face the same problem as in simple graphs: the number of cut edges does not directly describe the number of values transferred between processors, and as such, it can provide a very inaccurate metric for the cost. In the extreme case, it can happen that a red node $u$ has an edge to $(n-1)$ distinct blue successors; while this implies $(n-1)$ cut edges between the two parts, in reality, we only need to transfer a single value once: the output of computation $u$ from the red to the blue processor. As such, to obtain an accurate model of communicate costs in computational DAGs, we introduce the notion of \emph{hyperDAGs}.

\begin{definition}
For a given a computational DAG $G(V,E)$, the corresponding \emph{hyperDAG} $G'(V', E')$ is defined by $V' \coloneqq V$ and
\[ E' \coloneqq \left\{ \{u\}\cup S_u \mid u\in V \right\} \, , \]
where $S_u\coloneqq \{v\in V \mid (u,v)\in E\}$ is the set of immediate successors of $u\in V$.
\end{definition}

Given this hypergraph representation of the computational DAG (illustrated in Figure \ref{fig:hyperDAG}), the hypergraph partitioning problem now provides the correct metric for the communication cost: if an intermediate value $u$ (represented by hyperedge $e=\{u\}\cup S_u$) is computed and stored on some processor $p_u \in [k]$, then we need $(\lambda_e-1)$ transfer operations to make this value available for all the other processors that compute a successor of $u$. As such, using hyperDAGs instead of computational DAGs allows us to also capture the communication cost correctly when modelling a computation with precedence constraints.

We point out that many works on computational DAGs assume that the indegrees of nodes are bounded by a small constant \cite{const_indegree1, const_indegree2}. This directly translates to a small $\Delta$ in the resulting hyperDAG: if e.g.\ each node of a DAG has an indegree of at most $2$ (i.e.\ all operations are binary), then our hyperDAG will have $\Delta \leq 3$. As such, hyperDAGs with a small constant degree are of particular interest in practice.

If we have a description of our hyperDAG that specifies for each hyperedge the node from which it was generated, then one can easily verify whether this hyperDAG corresponds to a valid (computational) DAG. On the other hand, if we are only given a general hypergraph $G$ (without the generator nodes specified), it is not trivial to decide whether $G$ is actually a hyperDAG, i.e.\ if it corresponds to the hyperDAG representation of some DAG. For example, the triangle in Figure \ref{fig:non-hyperDAG} is a simple hypergraph that cannot be obtained as a hyperDAG from any original DAG: for instance, it does not have a node of degree $1$ that could correspond to a source of the original DAG.

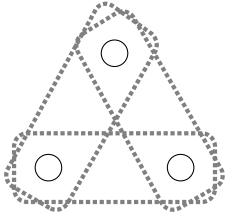
\begin{figure}
	\centering
    \vspace{10pt}
	\input{pics/non-hyperDAG.tikz}
	\caption{Example hypergraph that does not correspond to any computational DAG.}
	\label{fig:non-hyperDAG}
\end{figure}

This shows that hyperDAGs are only a subclass of general hypergraphs, and hence understanding their properties is an important question: hyperDAGs could have some structural properties that make the partitioning problem easier on them. As such, even though this is not closely related to the main focus of the paper, we also provide a brief analysis of the fundamental properties of hyperDAGs in Appendix \ref{app:hyperDAGs}.

\begin{itemize}[topsep=3pt, itemsep=3pt, parsep=2pt]
 \item Firstly, we develop a complete characterization of hyperDAGs: we show that a hypergraph is a hyperDAG if and only if a specific property holds for all of its subgraphs.
 \item Using this characterization, we then also show that it can be decided in linear time whether a given hypergraph is a hyperDAG.
 \item Finally, for the sake of completeness, we prove that the partitioning problem still remains NP-hard if restricted to hyperDAG inputs.
\end{itemize}

\section{Inapproximability result} \label{sec:main}

In this section we discuss our main theorem, which extends the previously known hardness result from the bisection case to the partitioning problem for general $\epsilon\geq 0$. More importantly, we also show that this hardness result already holds in heavily restricted cases (hyperDAGs of degree at most $2$), suggesting that the problem is not even approximable in practically relevant settings.

\begin{theorem} \label{th:main}
Assuming ETH, it is not possible to approximate the optimum of the partitioning problem to an $n^{1/(\log \log n)^{\delta}}$ factor in polynomial time (for some constant $\delta>0$). This holds for any $k \geq 2$ and $\epsilon \geq 0$, even if the input is restricted to hyperDAGs with $\Delta = 2$.
\end{theorem}

\renewcommand*{\proofname}{Proof (sketch)}

\begin{proof}
We use a reduction from the well-studied Smallest $p$-Edge Subgraph (S$p$ES) problem: given a graph $G(V,E)$, we need to find a subset $V_0 \subseteq V$ such that the subgraph induced by $V_0$ has at least $p$ edges, and $|V_0|$ is minimized. It is known that if ETH holds, then there is a $\delta>0$ such that no polynomial-time $n^{1/(\log \log n)^{\delta}}$-factor approximation exists to this problem \cite{ETHhardness}.

Given an instance of the S$p$ES problem, the main idea is to convert it into a hypergraph that mostly consists of \textit{blocks}: groups of nodes which are so densely interconnected by hyperedges that they all need to receive the same color, otherwise we end up with an unreasonably high cost. These blocks are also used as a fundamental ingredient in several other constructions throughout the paper. In our current construction, we begin by creating two very large blocks $A$, $A'$, and enforcing (through the balance constraint) that they must obtain different colors; let us assume w.l.o.g.\ that $A$ is colored blue, and $A'$ is colored red. We also create a smaller block for each $e \in E$, and we carefully select the size of blocks such that at least $p$ of these edge blocks must be colored red in order to satisfy the balance constraint. Finally, for each $v \in V$, we create a hyperedge which contains (i) a node from the block of every $e$ that is incident to $v$, and (ii) a further node that is forced to be blue (due to further hyperedges connecting it to $A$). The construction is illustrated in Figure \ref{fig:mainth}.

\begin{figure*}
	\centering
    \vspace{10pt}
	\input{pics/mainth.tikz}
	\caption{High-level illustration of the construction for Theorem \ref{th:main} in the general case, with the squares denoting block gadgets. To fulfill the balance constraint, at least $p$ of the blocks in the middle (corresponding to edges of $G$) need to be red. The hyperedges corresponding to the nodes of $G$ all contain a blue node, so they are cut if and only if at least one of the incident edge blocks is red. As such, a partitioning of cost at most $L$ corresponds to a subset of $p$ edges in $G$ that are altogether incident to at most $L$ nodes. Note that instead of a separate blue node, the hyperedges could also directly intersect into $A$; however, this would make the adaptation to $\Delta=2$ more technical.}
	\label{fig:mainth}
\end{figure*}
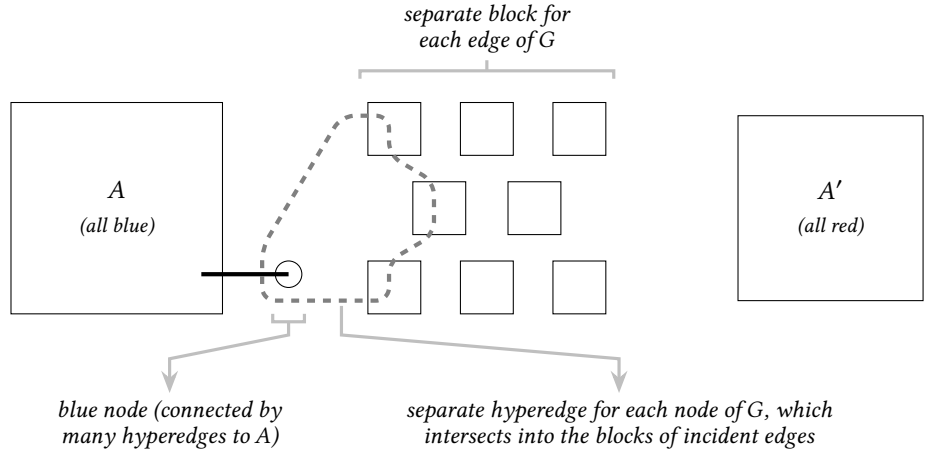

In the resulting hypergraph, we need to select a subset of (at least) $p$ edge gadgets that we color red. However, if a node $v \in V$ is incident to any of these $p$ edges, then the hyperedge corresponding to $v$ will be cut, since it contains both a red and a blue node. Altogether, the cost of a solution will be exactly the number of nodes covered by the $p$ chosen (red) edges of $G$; as such, approximating the minimum cost would also allow us to approximate the S$p$ES problem to the same factor.

The more technical part of the theorem is to extend the reduction first to hypergraphs with $\Delta = 2$, and then to hyperDAGs. For the extension to $\Delta \leq 2$, we essentially replace all the blocks by ``grid gadgets'': these are gadgets which are $2$-regular (each node has degree $2$), but they still ensure that cutting off a significant portion of the nodes from the gadget induces an unacceptably high cost. We then discuss how to connect these grid gadgets to each other in a way such that it essentially exhibits the same properties as the original construction with blocks. Finally, in order to convert the construction into a hyperDAG, we add further auxiliary nodes that cannot affect the optimal partitioning, and then we show that there is an injective assignment from hyperedges to generating nodes, i.e.\ there indeed exists a computational DAG that corresponds to this hypergraph.
\end{proof}

We note that our hardness results also carry over to the special class of $2$-regular hypergraphs recently studied by \cite{KB20} for modelling SpMV problems.

Furthermore, note that Theorem \ref{th:main} was expressed in terms of ETH, which is a rather standard complexity assumption. There are several stronger inapproximability results for S$p$ES based on less standard assumptions; these also provide stronger hardness results for the partition problem.

\begin{corollary}
Our reduction method also shows the inapproximability of the partitioning problem to the following factors based on stronger complexity conjectures:
\vspace{4pt}
\begin{itemize}
    \setlength{\itemsep}{3pt}
    \setlength{\parskip}{2pt}
    \item $n^{f(n)}$ for any function $f(n)=o(1)$, if Gap-ETH holds \cite{ETHhardness},
    \item $n^{\delta}$ for a given $\delta>0$, if specific one-way functions exist \cite{crypt},
    \item $n^{\frac{1}{12}-\delta}$ for any $\delta>0$, if the Hypergraph Dense vs.\ Random Conjecture holds \cite{DenseVsSparse2}.
\end{itemize}
\end{corollary}

Besides approximation algorithms, it is also interesting to study the hardness of the problem from a parameterized complexity perspective, in terms of the allowed cost $L$. It follows easily from the W[1]-hardness of S$p$ES that the partitioning problem is also W[1]-hard. On the other hand, one can show that the problem is in XP, i.e.\ it can be solved in time $n^{f(L)}$ for some function $f$. Intuitively, the main idea is to try all possible combinations of cut hyperedges that can result in a total cost of at most $L$; this means that at most $L$ hyperedges are cut, so there are only $n^{f(L)}$ such cases. Then in each of these cases, we can essentially remove these cut hyperedges (converting them into constraints) to obtain a delicate packing problem that can be solved by a dynamic programming approach.

\begin{lemma} \label{lem:parameterized}
In terms of the allowed cost $L$ as a parameter, the partitioning problem is W[1]-hard (already for hyperDAGs of degree $\leq 2$), but it is in XP.
\end{lemma}

\section{Balance constraints for hyperDAGs} \label{sec:multi}

If our hypergraph models an application where the computational steps have no inter-dependence, i.e., they can be executed in any order, then the balance constraint already ensures that the computational workload on the $k$ processors is evenly distributed. However, if we have a hyperDAG which was obtained from a computational DAG with precedence constraints, then a simple balance constraint may fail to ensure any amount of parallel execution.

Indeed, if for example our DAG $G$ is a serial concatenation of two DAGs $G_1$ and $G_2$ of the same size (as sketched in Figure \ref{fig:1constraint_limits}), then an assignment where $G_1$ and $G_2$ are assigned to the red and blue processors, respectively, is perfectly balanced. Yet, the blue processor will need to wait for all the computations on the red processor to finish before it can begin the computation of any blue nodes at all. Thus even though this partitioning satisfies the balance constraint, we are in fact unable to parallelize the workload between the red and blue processors at all.

This suggests that in order to ensure parallel execution in hyperDAGs, we require a more refined approach for our balance constraint.

\subsection{Layer-wise constraints} \label{sec:layers}
One natural idea is to divide a given algorithm into ``phases'', and ensure that the workload is balanced in each phase separately.
In the simplest case, this corresponds to dividing the nodes of a hyperDAG into \textit{layers}, i.e.\ disjoint sets $V_1, \ldots, V_{\ell}$ such that $\ell$ is the length of the longest path in the DAG, and for each directed edge $(u,v)$ with $u \in V_i, v \in V_j$, we have $i<j$. We illustrate such a layering in Figure \ref{fig:layers}. We can then define a layer-wise version of the partitioning problem where the balance constraint needs to hold in each layer separately (while our goal is still to minimize the cost of the cut as before).

\begin{figure}
	\centering
    \vspace{10pt}
    \resizebox{0.37\textwidth}{!}{\input{pics/1constraint_limits.tikz}}
	\caption{The limits of having a single balance constraint for hyperDAGs: while the above partitioning is perfectly balanced, it does not provide any parallelization opportunities in fact.}
	\label{fig:1constraint_limits}
\end{figure}
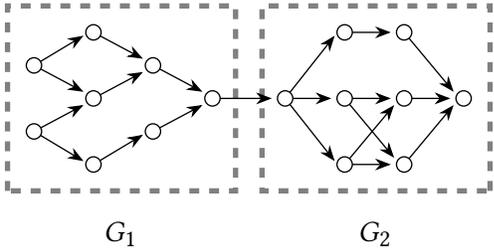

\begin{definition}
In the \emph{layer-wise balanced hyperDAG partitioning problem}, a partitioning is only feasible if each layer is balanced, i.e.\ if for all $j \in [\ell], i \in [k]$, we have $|P_i \cap V_j| \leq (1+\epsilon) \cdot \frac{|V_j|}{k}$.
\end{definition}

In the simplest case, we can create layers by sorting every node into the earliest possible layer: $V_1$ contains the source nodes of the DAG, and then $V_i$ (for $i \geq 2$) is the set of nodes that have all their predecessors contained in $\bigcup_{j=1}^{i-1} \, V_j$; this indeed divides the DAG into layers $V_1, \ldots, V_{\ell}$. However, in general, there are multiple different ways to divide the DAG into layers; for example, in the DAG in Figure \ref{fig:layers}, the lowermost node can be sorted either into layer $V_2$, $V_3$ or $V_4$. Hence we can also define a more general, \textit{flexible layering} version of the partitioning problem, where our goal is twofold: we first need to select a valid layering of the DAG as discussed above, and we then need to find a layer-wise balanced partitioning according to these layers, with the final goal of minimizing the cost.

With these layer-wise constraints, the partitioning problem turns out to be even harder: it is already NP-hard to distinguish between an optimal cost of $0$ and $n^{1-\delta}$ (for any $\delta>0$). 

\begin{theorem} \label{th:layers}
It is NP-hard to approximate the layer-wise balanced partitioning problem to any finite factor, both in the fixed and in the flexible layering case.
\end{theorem}

\begin{proof}
The proof consists of several technical steps; these are discussed in detail in Appendices \ref{app:multi} and \ref{app:sec:layerwise}. The main ideas behind the proof are as follows:
\begin{itemize}[topsep=3pt, itemsep=2pt, parsep=2pt]
\item Our DAG construction consists of several connected components, each having a carefully designed number of nodes in each layer. In order to obtain a partitioning of cost $0$, all of these components need to be monochromatic.
\item Besides the main components, we also add $k$ ``control components'', and use auxiliary layers at the end of the DAG to ensure that these all receive different colors. These control components are then used to add a desired number of nodes of fixed colors to any layer, and hence, intuitively, to ensure that specific layers must contain at least/at most a specific number of nodes of given colors.
\item At the core of our construction, there is a reduction from the well-known graph coloring problem, using the above tools to convert the coloring problem to this multi-constraint partitioning setting.
\end{itemize}
In the resulting DAG, a layer-wise partitioning of cost $0$ exists if and only if the original graph has a valid $3$-coloring. Moreover, our DAG is designed to allow only one possible layering, so this settles the proof for both the fixed and the flexible layering case.
\end{proof}

A slightly different version of this proof also shows that in the flexible layering case, this hardness result already applies separately to the subproblem of finding the best layering of the DAG, i.e.\ the layering where the optimum cost is smallest. In other words, even if we have an oracle that returns the optimal partitioning for a specific fixed layering of the DAG, the optimum is still not approximable to any finite factor.

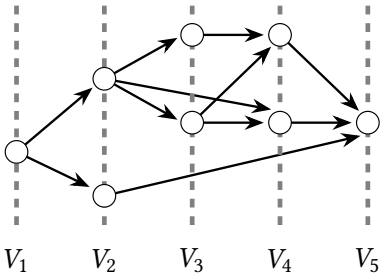
\begin{figure}
	\centering
    \vspace{10pt}
    \resizebox{0.3\textwidth}{!}{\input{pics/DAG_layers.tikz}}
	\caption{An example for dividing a DAG into layers.}
	\label{fig:layers}
\end{figure}

\subsection{Schedule-based constraints}

We have seen that layer-wise constraints (instead of a single constraint) allow us to exclude solutions where computations in a hyperDAG are in fact not parallelized. Unfortunately, the layer-wise approach can run into a different problem: it may impose a condition that is too strict, also excluding some solutions that are in fact perfectly parallelized.

In particular, consider a DAG with two distinct paths of length $3$ from a single source node to a single sink node, and then let us split both the first node in the upper path and the second node in the lower path into a larger set of $b$ nodes, as sketched in Figure \ref{fig:layers_limits}. With layer-wise constraints (and a sufficiently small $\epsilon$), we are forced to partition both of these sets in an (almost) balanced way, since they contain almost all the nodes in the given layer. Hence whichever color we choose for the successor of these sets, we will have a cost of $\Theta(b)$. In contrast to this, if we were to simply color the upper branch red and the lower branch blue (and the source and sink node with an arbitrarily chosen color), then we have near-perfect parallelization, and a cut cost of only $2$ altogether.

In general, the only straightforward way to develop an exact metric of parallelization is to consider a concrete \textit{scheduling} of the DAG, which assigns the nodes not only to processors, but also to time steps. A detailed discussion of scheduling problems is beyond the scope of this paper; in the rest of the section, we briefly show how scheduling can be used to develop a more accurate balance constraint for partitioning, and we also discuss the limits of this approach.

\begin{definition} \label{def:DAG_sched}
Given a DAG and a fixed constant $k$, a \emph{scheduling} is an assignment of the nodes to processors $p:V\rightarrow [k]$ and to time steps $t:V\rightarrow \mathbb{Z}^+$ such that
\vspace{4pt}
\begin{itemize}
    \setlength{\itemsep}{3pt}
    \setlength{\parskip}{2pt}
    \item for all $u,v \in V$ we have either $p(u) \neq p(v)$ or $t(u) \neq t(v)$ (the scheduling is correct),
    \item for all $(u,v) \in E$, we have $t(u) < t(v)$ (the precedence constraints are satisfied).
    \vspace{5pt}
\end{itemize}
The goal is then to minimize the makespan $\max_{v \in V\,} t(v)$ of the scheduling, i.e.\ to execute the computations as fast as possible (without considering communication costs).
\end{definition}

When compared to $n$, this optimal makespan essentially allows us to measure the parallelizability of the DAG. For example, if our DAG is simply a directed path, then the best makespan is $n$ (the DAG is not parallelizable at all); on the other hand, if it consists of $k$ disjoint DAG components of equal size, then the best makespan is $\frac{n}{k}$ (the DAG is perfectly parallelizable). As such, it is a natural idea to use this metric to define a more sophisticated, schedule-based balance constraint on our hyperDAGs, where a given partitioning is feasible if it can be relatively well parallelized compared to the best possible parallelization of the DAG.

More formally, given a DAG, let $\mu$ denote the minimal makespan in general (i.e.\ over all $p',t'$ such that $(p',t')$ is a valid schedule), and let $\mu_p$ denote the minimal makespan for a fixed partitioning $p$ (i.e.\ over all $t'$ such that $(p,t')$ is a valid schedule).

\begin{definition} \label{def:schedconstraint}
In a \emph{schedule-based balance constraint}, we say that a partitioning $p:V\rightarrow [k]$ is feasible if $\mu_p \leq (1+\epsilon) \cdot \mu$.
\end{definition}

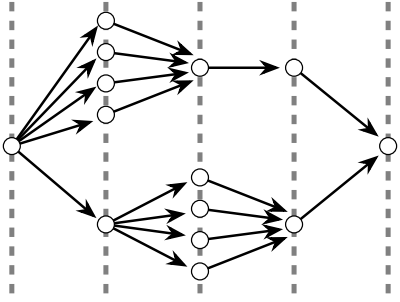
\begin{figure}
	\centering
    \vspace{10pt}
    \resizebox{0.3\textwidth}{!}{\input{pics/layers_limits.tikz}}
	\caption{The limits of layer-based balance constraints. A layer-wise balanced partitioning has to split the large sets both in the upper and the lower branch, resulting in a large cost. On the other hand, coloring the upper branch red and the lower branch blue provides near-perfect parallelization at a cost of only $2$.}
	\label{fig:layers_limits}
\end{figure}

This schedule-based constraint provides a much more sophisticated condition of sufficient parallelization in hyperDAGs. On the other hand, the approach has strong limitations in practice; we discuss these in Appendix \ref{app:scheduling} in detail. In particular, the DAG scheduling problem (computing $\mu$) is not known to be polynomially solvable except for a few special cases, such as for $k=2$, or for special classes of DAGs such as out-trees, level-order DAGs or bounded-height DAGs \cite{DAG2proc1, oppforest, levelorder, boundedheight}. Moreover, maybe more surprisingly, we show that evaluating the quality of a given partitioning (computing $\mu_p$) is an even harder problem which remains NP-hard even in these very special cases.

\begin{theorem} \label{th:scheduling}
Computing $\mu_p$ is already NP-hard for $k=2$, even if the inputs are restricted to out-trees, level-order DAGs or bounded-height DAGs.
\end{theorem}

This provides an unusual situation in these special cases: we can efficiently compute the parallelizability of the DAG in general, but we cannot compute how parallelizable \textit{our own solution} is, and hence we cannot verify if it satisfies a schedule-based balance constraint. This suggests that such a schedule-based constraint is not a viable approach in practice, even for the simplest case of $k=2$.

\section{Multi-constraint partitioning} \label{sec:multi_sub}

The layer-wise balance constraints in our hyperDAGs are in fact a special case of a natural generalization of the partitioning problem, where instead of having only a single balance constraint for the whole set $V$, we have separate balance constraints for smaller subsets of nodes.

This problem might be of independent interest in several applications. For instance, as another (more practical) approach to ensure sufficient parallelization in hyperDAGs, one might decide to heuristically decompose the hyperDAG into relatively independent ``regions'' (preferably larger than layers but smaller than the entire graph, such as e.g.\ the sets $G_1$ and $G_2$ in Figure \ref{fig:1constraint_limits}), and enforce a balance constraint on each region separately. To our knowledge, similar multi-constraint problems have only been studied on simple graphs \cite{KV98} or in particular applications \cite{layerwiseDNN} before.

\begin{definition}
In the \emph{multi-constraint partitioning problem}, our input also contains disjoint subsets $V_1, \ldots, V_c \subseteq V$. We say that a partitioning $\mathcal{P}=\{P_1, \ldots, P_k \}$ of $V$ is feasible if it satisfies the balance constraint for all subsets, i.e.\ for all $j \in [c], i \in [k]$, we have $|P_i \cap V_j| \leq (1+\epsilon) \cdot \frac{|V_j|}{k}$.
\end{definition}

The simplest cases to analyze in terms of hardness are the two extremes of $c$. That is, when we only have $c=O(1)$ constraints, then the problem still remains in XP, and a simple reduction allows us to carry over some of the known approximation algorithms on standard partitioning to this multi-constraint case, although in a significantly weaker form. On the other hand, when we have $c \geq n^{\delta}$ for some constant $\delta>0$, the problem becomes significantly harder both in terms of approximability and parameterized complexity.

\begin{lemma} \label{lem:multi_const}
For $c \in O(1)$ constraints, there exists a reduction from multi-constraint bisection to the standard bisection problem (see Appendix \ref{app:multi:approx} for details), and the partitioning problem is still in XP (with respect to $L$).
\end{lemma}

\begin{lemma} \label{lem:multi_lin}
If we have $c \geq n^{\delta}$ constraints for some constant $\delta >0$, then no polynomial-time approximation exists for the partitioning problem to any finite factor, and the problem is para-NP-hard (with respect to $L$).
\end{lemma}

Between the two cases when $O(1)<c<n^{\delta}$, the question is not so straightforward. However, with a stronger complexity assumption (the Strong Exponential Time Hypothesis, SETH), we can also show a hardness result here for any algorithm running in subquadratic time. This is indeed a relevant observation, since having quadratic running time is often already prohibitive in practice.

\begin{theorem} \label{th:OVP_quadratic}
For multi-constraint partitioning with $c=\omega(\log n)$, no finite factor approximation algorithm is possible in subquadratic time (i.e.\ $n^{2-\delta}$ for some $\delta>0$) if SETH holds.
\end{theorem}

\begin{proof}
The proof uses a reduction from the so-called Orthogonal Vectors Problem (OVP): given a set of $m$ binary vectors $a_1, ..., a_m$, the goal is to decide whether any two of these vectors are orthogonal (i.e.\ their dot product is $0$). It is known that for vectors of dimension $D=\omega(\log{m})$, this cannot be decided in $O(m^{2-\delta})$ time unless SETH is falsified \cite{OVPhard}.

The main idea of the construction is to have a separate node $v_i\,\!^{(j)}$ representing the $j$-th dimension of the $i$-th vector for each $i_{\!} \in_{\!} [m]$, $j_{\!} \in_{\!} [D]$, and a further ``anchor'' node $u_i$ for each vector. For each fixed $i_{0\!} \in_{\!} [m]$, we add a hyperedge containing the node $u_{i_0}$, and all nodes $v_{i_0}\,\!^{(j)}$ such that the $j$-th coordinate of $a_i$ is $1$.

Then through a series of technical steps, we create balance constraints that fulfill the following properties. Firstly, for each fixed $j_{0\!} \in_{\!} [D]$, we add a dimension-wise balance constraint which ensures that at most one of the nodes $v_i\,\!^{(j_0)}$ can be red. Furthermore, we also add a single balance constraint on the anchor nodes, ensuring that at least two of the nodes $u_i$ need to be red.

Assume we want to find a valid multi-constraint partitioning of cost $0$ in the resulting construction. This requires us to color two of the anchor nodes $u_{i_1}$ and $u_{i_2}$ red; however, then for all entries that are $1$ in the chosen vectors, the corresponding nodes $v_i\,\!^{(j)}$ must also be red. However, such a solution can only satisfy the dimension-wise balance constraints if there is no dimension $j$ where both vectors have an entry of $1$, i.e.\ if they are orthogonal.
\end{proof}

\section{Hierarchical cost function} \label{sec:hier}

The simplicity of the partitioning problem makes it a very popular model to analyze the parallel execution of computations. On the other hand, due to this simplicity, the model cannot capture one of the most prominent characteristics of modern computing architectures, namely \emph{non-uniform memory access}: transferring data between different pairs of processing units can have very different costs in practice. This is usually due to the hierarchical structure of these architectures: we often have several cores within the same processor, several processors attached to the same RAM, and then several of these computers connected over a network. Modern architectures even expose such hierarchical structure within single processors. In such an architecture, the communication cost between two cores heavily depends on the highest level of the hierarchy that the data has to cross: sending data between two cores on the same processor is a relatively fast operation, while sending data to another core through the top-level network connection is drastically slower.

As such, it is a natural goal to extend our analysis of partitioning problems to such a hierarchical setting. Formally, we will model these architectures by a rooted tree of depth $d$, with the leaves of the tree corresponding to the compute units. We assume that each level of this tree has a fixed branching factor $b_i$, i.e.\ every node on the $i$-th level (from the top) has exactly $b_i$ children; this implies that we partition our hypergraph into $k=\prod_{i=1}^{d} \, b_i$ sets. Furthermore, assume we have a set of constant cost parameters $g_1, \ldots, g_{d}$ such that if two computing units have their lowest common ancestor in level $i$ of the tree, then transferring a variable between the parts has a cost of $g_i$ (as illustrated in Figure \ref{fig:hierarch}). We assume the $g_i$ are monotonically decreasing, and for simplicity, we normalize them to ensure $g_d=1$.

\begin{figure}
	\centering
    \vspace{10pt}
    \input{pics/hierarchy.tikz}
	\caption{Illustration of communication costs in the hierarchical setting: the cost of transferring a variable depends on the level of the hierarchy that the data has to cross.}
	\label{fig:hierarch}
\end{figure}
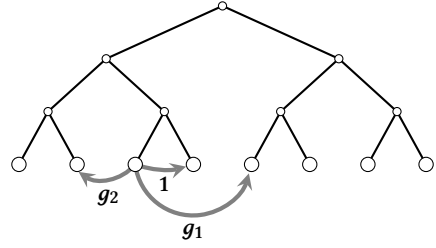

\begin{definition}
In the \emph{hierarchical partitioning problem}, we have constant parameters $b_1, \ldots, b_d$ such that $\prod_{i=1}^d b_i = k$. Our goal is to partition $V$ into $k$ sets $P_1\!^{\,(d)}, \ldots, P_k\!^{\,(d)}$ such that $|P_j\!^{\,(d)}|\leq (1+\epsilon) \cdot \frac{n}{k}$ for all $j \in [k]$. However, we now also need to organize these parts into a hierarchy, i.e.\ for all $i \in \{ 2, \ldots, d \}$, we partition the classes $P_j\!^{\,(i)}$ into $\prod_{\ell=1}^{i
	-1} \, b_{\ell}$ distinct sets of size $b_i$ each to form the $(i-1)$-th level parts $P_j\!^{\,(i-1)}$. For a hyperedge $e$, let $\lambda_e\!^{(i)}$ denote the number of $i$-th level parts that $e$ intersects (and for simplicity, let $\lambda_e\!^{(0)}:=1$). The cost induced by $e$ is the defined as
\[ \sum_{i=1}^{d} \, g_i \cdot (\lambda_e\!^{(i)} - \lambda_e\!^{(i-1)}) \, , \]
and the total cost of a partitioning is again the sum of this cost over all hyperedges $e \in E$.
\end{definition}

For example, if $e$ intersects all the $k=4$ parts in a $2$-level hierarchy with $b_1=b_2=2$, then regardless of which part the variable is stored in, the cost of transferring it to the other three parts is $g_1 + 2 \cdot g_2 = g_1 + 2$: we need to move the variable once over the top level, and twice over the bottom level of the hierarchy. The formula indeed equals to this for $\lambda_e\!^{(1)}=2$ and $\lambda_e\!^{(2)}=4$. Note that the standard partitioning problem is obtained as a special case of this setting when our hierarchy has depth $d=1$.

This hierarchical cost function results in a more complex version of the partitioning problem, where the role of different parts is not symmetric anymore. We briefly discuss some key properties of this more realistic model below, with the technical details deferred to Appendix \ref{app:hier}. Besides this, Appendix \ref{app:gen} also provides a brief discussion of two more questions that arise naturally regarding our results on this hierarchical model: (i) how they carry over to hyperDAGs and/or the multi-constraint setting from previous sections, and (ii) how they can be generalized to cost functions that are inspired not by a tree, but by a different (arbitrary) processor topology.

\subsection{Recursive approach}

\begin{figure*}
	\centering
    \vspace{10pt}
	\input{pics/recursive.tikz}
	\caption{Construction for Lemma \ref{lem:recurse}. Large and small squares correspond to blocks of size $\frac{n}{6}$ and $\frac{n}{12}$, respectively. Recursive bipartitioning (left side) first makes an optimal split of cost $0$ along the vertical axis; however, in the next recursive step, it needs to split one of the blocks to fulfill the balance constraint, resulting in a cost of $\Theta(n)$. On the other hand, direct $k$-way partitioning (right side) can provide a solution of cost $O(1)$ only.}
	\label{fig:recursive}
\end{figure*}
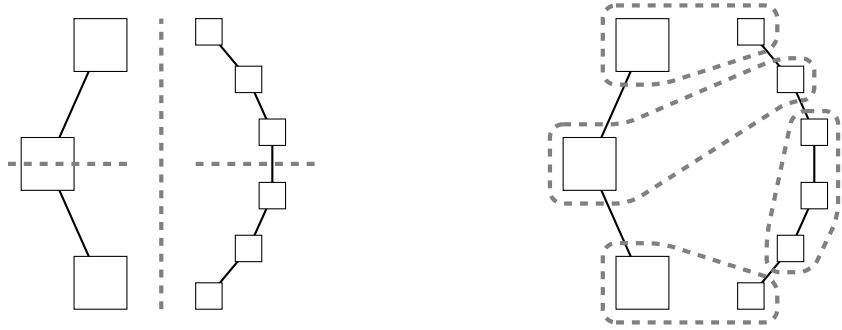

A natural solution idea for partitioning is to recursively split $G$ into smaller and smaller parts. Even in the regular $k$-way partitioning problem, such a recursive approach is very commonly used in heuristics: we can repeatedly split each part into two further parts, until the number of parts reaches $k$.

This recursive method also provides a natural solution approach for our hierarchical partitioning problem: we can first try to split $G$ cleverly into $b_{1}$ parts, then split each of these into $b_{2}$ further parts, and so on, forming the entire hierarchy of $k$ parts in such a recursive way. Moreover, the approach is even more intuitive in this hierarchical case: since the cuts on the highest level induce a much larger cost, it seems reasonable to first minimize the number of these cuts, and only then move on to the lower levels.

It has already been observed before that such a recursive approach is not always optimal \cite{recursive1, recursive2}; however, we can show in a simple example that it can even be a linear factor away from the optimum cost.

\begin{lemma} \label{lem:recurse}
The solution returned by recursive partitioning can be a factor $\Theta(n)$ off the optimum cost, both for regular and hierarchical partitioning, even if each of the recursive steps is optimal separately.
\end{lemma}

\begin{proof}
Let $b_1\!=\!b_2\!=\!2$, and consider the hypergraph sketched in Figure \ref{fig:recursive}, consisting of $9$ densely connected blocks and only a few hyperedges between these blocks. Assume that the larger blocks in the figure each consist of $\frac{n}{6}$ nodes, while the smaller blocks each consist of $\frac{n}{12}$ nodes. 

In an optimal recursive partitioning, the first step will split this hypergraph into two parts of equal size along the vertical axis, without cutting any hyperedges (see the left side of the figure). However, in the next step, the recursive approach needs to split both sides into two further parts; with a small enough $\epsilon$ in the balance constraint, this forces us to split one of the densely connected blocks on the left side, resulting in a cost of $\Theta(n)$.

In contrast to this, there exists a direct $4$-way partitioning of the hypergraph where only $O(1)$ hyperedges are cut (right side of the figure). This solution has a factor $\Theta(n)$ smaller cost than the recursive solution, not only under the regular (cut-net or connectivity) cost metrics, but also according to our hierarchical cost function, since the coefficients $g_i$ are constants.
\end{proof}

\subsection{Hierarchy-agnostic partitioning}

Another natural idea in this setting is to apply a regular partitioning algorithm that does not consider the underlying processor hierarchy at all. More specifically, given an input hypergraph, we can use the following \textit{two-step method} to obtain a hierarchical partitioning:
\vspace{1.5pt}
\begin{enumerate}[label=(\roman*),leftmargin=15pt, itemsep=4pt]
    \setlength{\itemsep}{1pt}
    \setlength{\parskip}{2pt}
    \item first find a good regular $k$-way partitioning of the hypergraph,
    \item then assign these $k$ parts to the $k$ leaf positions in the hierarchy in a clever way.
\end{enumerate}
\vspace{3pt}
For the analysis of this two-step method, we will assume that both steps happen optimally: we first find an optimal (regular) partitioning of the hypergraph, and then we also assign the $k$ parts to hierarchy positions in an optimal way. This allows us to study a fundamental question: what happens if we have a good partitioning algorithm, but we disregard the hierarchical nature of modern computing architectures in the partitioning step?

On the one hand, it is easy to show that the optimal solution with this two-step method is at most a factor $g_{1}$ worse than the true optimum for hierarchical cost. Intuitively,
this is because an optimal algorithm for standard partitioning can only misjudge the real (hierarchical) cost of each hyperedge by a $\frac{g_{1}}{g_d}=g_{1}$ factor.

\begin{figure*}
	\centering
    \vspace{10pt}
	\input{pics/twostep.tikz}
	\caption{Overview of the construction of Theorem \ref{th:two_step_bad} for $k=4$. Block $A$ is connected by a large number of hyperedges to blocks $B_1$, $B_2$ and $B_3$ (we will refer to these as $A \leftrightarrow B_i$ hyperedges). In the regular optimum, the blocks $B_i$ are each placed in a separate part (since the $A \leftrightarrow B_i$ hyperedges are cut anyway, and the $B_i$ have connections to other blocks not shown in the figure). This results in a large hierarchical cost: many of the $A \leftrightarrow B_i$ hyperedges induce a cost of $g_1$. On the other hand, if we place all $B_i$ in the same part, then the $A \leftrightarrow B_i$ hyperedges all induce a cost of $g_2$ only.}
	\label{fig:twostep}
\end{figure*}
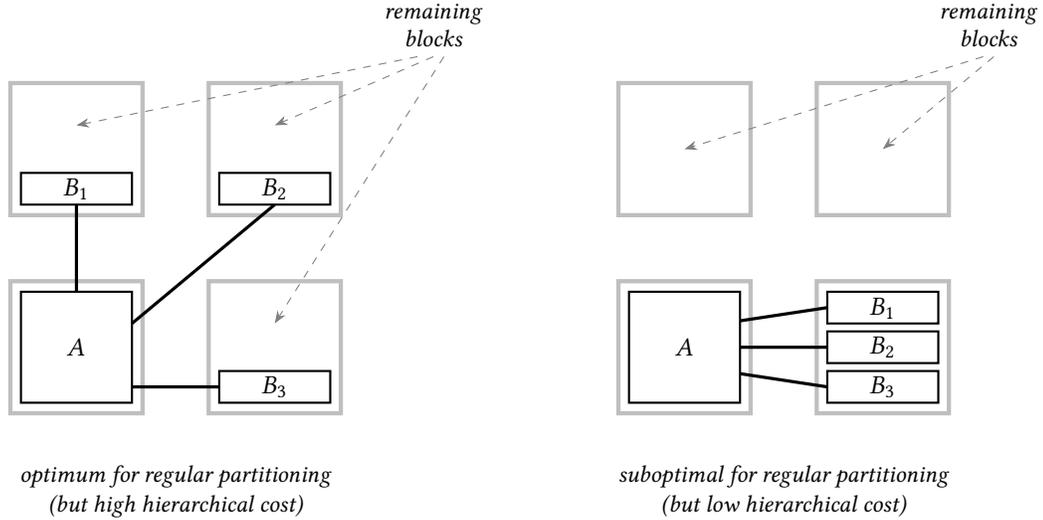

\begin{lemma} \label{lem:twostep_approx}
The two-step method is a $g_{1}$-approximation. 
\end{lemma}

On the other hand, it turns out that in unfortunate cases, the difference can indeed be in the magnitude of a factor $g_{1}$. This carries an important conceptual message: if we ignore the fact that the true nature of the cost function is hierarchical in practice, then even by finding the optimal partitioning, we might still be a large constant factor away from the actual optimum cost. 
\begin{theorem} \label{th:two_step_bad}
The two-step method can be a factor $\frac{b_1-1}{b_{1}} \cdot g_{1}$ worse than the optimum.
\end{theorem}

\begin{proof}
On a high level, the proof requires a star-shaped construction where a large block $A$ is densely connected to $(k-1)$ smaller blocks $B_i$, which can either all fit into a single part, or into separate parts (sketched in Figure \ref{fig:twostep}). The construction is carefully designed such that placing each $B_i$ in a separate part has slightly smaller standard cost, so this solution is preferred by the two-step method. However, the hierarchical cost of this solution is much higher, since the parts containing the $B_i$ are scattered across the hierarchy, so many of the connections between blocks incur a cost of $g_{1}$. On the other hand, if the $B_i$ are all placed into the same part, and this is a sibling of the part containing $A$ on the lowest hierarchy level, then the same connections only incur a cost of $g_{d}$, so there is indeed a solution with a significantly lower hierarchical cost.
\end{proof}

\renewcommand*{\proofname}{Proof}

Note that $\frac{b_1-1}{b_{1}} \geq \frac{1}{2}$ even for the simplest case of $b_1=2$, and as $b_{1}$ grows larger, it approaches $1$, matching the upper bound of Lemma \ref{lem:twostep_approx}. We also note that this theorem only holds for small $\epsilon$ values, i.e.\ when we are not allowed to leave any of the parts empty.

\subsection{Complexity}

Finally, it is also natural to wonder how this hierarchical cost function affects the hardness of the partitioning problem. On the one hand, the problem does become significantly more technical in this hierarchical setting in practice; on the other hand, it does not change much from a complexity-theoretic perspective. In particular, we can make the following simple observations:
\begin{itemize}[topsep=3pt, itemsep=2pt, parsep=2pt]
 \item The hardness results for standard partitioning, e.g.\ Theorem \ref{th:main}, also carry over in a straightforward way to this more complex hierarchical setting.
 \item Due to Lemma \ref{lem:twostep_approx}, any $\alpha$-approximation algorithm for standard partitioning also provides an $O(\alpha)$-approximation for the hierarchical setting (as long as we have $k, g_1 \in O(1)$).
 \item Even with this more complex cost function, the partitioning problem still remains in the parameterized complexity class XP with respect to $L$.
\end{itemize}

The two-step method, on the other hand, raises a more interesting question if we consider its second step as a separate \textit{hierarchy assignment problem}. That is, given an already fixed $k$-way partitioning of the hypergraph, our goal is to assign the $k$ fixed parts to the $k$ available positions in the hierarchy optimally, i.e.\ such that the total hierarchical cost is minimized.

Note that with our assumption so far that $k \in O(1)$, the number of possible solutions to this problem is only a function of $k$, and hence also a constant; as such, the problem is trivial from a complexity-theoretic perspective. However, the problem becomes more interesting if we briefly explore the case when $k$ is a variable part of the input. This setting might be relevant in applications where the partitioning task is severely time-critical, and hence instead of using a fixed architecture, we e.g.\ decide to increase $k$ proportionally to the hypergraph size.

We conclude the paper by briefly analyzing the complexity of this hierarchy assignment problem in the simplest case of only $d=2$ levels. In this case, one can essentially contract each of the $k$ partitions into a single node, and express the hierarchy assignment problem as a specific kind of partitioning task on the resulting contracted hypergraph (which might contain multiple copies of some hyperedges). Our results (discussed in Appendix \ref{app:hierass}) show that this two-level hierarchy assignment problem is polynomially solvable for $b_2=2$, but already NP-hard if $b_2 = 3$.

\begin{theorem} \label{th:step2}
Consider the hierarchy assignment problem with only $d=2$ levels.
\vspace{3pt}
\begin{itemize}
    \setlength{\itemsep}{2pt}
    \setlength{\parskip}{2pt}
    \item for $b_2=2$, the problem is solvable in polynomial time,
    \item for $b_2 = 3$, the problem is already NP-hard.
\end{itemize}
\end{theorem}

\newpage

\bibliographystyle{ACM-Reference-Format}
\bibliography{references}

\newpage

\appendix

\section{Fundamental properties of the partitioning problem} \label{app:basics}

\renewcommand*{\proofname}{Proof.}

We begin with some straightforward observations on the partitioning problem that will often prove useful during our proofs.

\subparagraph*{From $\boldsymbol{\epsilon=0}$ to $\boldsymbol{\epsilon>0}$.} Firstly, we establish a close relation between the partitioning problem for general $\epsilon \geq 0$ and for $\epsilon=0$, showing that the latter is the hardest case of the problem in some sense. For convenience, in the special case of $\epsilon=0$, we will refer to a balanced partitioning as a $k$-section, and to the partitioning problem as the $k$-section problem (similarly to the notion of bisection for $k=2$).

\begin{lemma} \label{lem:bisection_reduction}
Given a polynomial algorithm for the $k$-section problem, this also provides a polynomial algorithm for the balanced partitioning problem with any $\epsilon>0$.
\end{lemma}

\begin{proof}
Consider an instance of the partitioning problem for some $\epsilon>0$, and let us add $\epsilon \cdot n$ new isolated nodes to the graph. With this the new number of nodes is $n'=(1+\epsilon) \cdot n$, so it fulfills $\frac{n'}{k}=(1+\epsilon) \cdot \frac{n}{k}$. We claim that there is a $k$-section of cost $L$ in this new hypergraph if and only if there is a $\epsilon$-balanced partitioning of cost $L$ in the original hypergraph. Indeed, any $k$-section has parts of size $\frac{n'}{k}$ in the new graph, so restricting it to the original graph implies that the balance constraints are satisfied. On the other hand, any $\epsilon$-balanced partitioning in the original graph can be extended to a $k$-section if we color the isolated nodes appropriately, i.e.\ such that each color occurs $\frac{n'}{k}$ times altogether.
\end{proof}

This reduction shows that any $\alpha$-approximation for $k$-section also provides an $O(\alpha)$-approximation for the problem with any $\epsilon>0$. This e.g.\ implies that the $\widetilde{O}(\sqrt{n})$-approximation of R\"{a}cke, Schwartz and Stotz \cite{BisectionApprox} for the bisection problem can also be extended to any $\epsilon>0$ value.

\begin{corollary}
For $k=2$ and any $\epsilon>0$, there is a polynomial-time $\widetilde{O}(\sqrt{n})$-approximation algorithm to the partitioning problem.
\end{corollary}

\subparagraph*{Choice of $\boldsymbol{\epsilon}$.} Furthermore, recall that the only strict restriction for $\epsilon$ we introduced was $\epsilon < (k-1)$, which ensures $(1+\epsilon) \cdot \frac{n}{k} < n$; otherwise, the optimum is always the trivial partitioning that places the entire graph into the same part. However, note that if we have a high number of processors $k$, then large $\epsilon$ values are still often unrealistic, even when $\epsilon<(k-1)$; in particular, it means that many of the processors can remain unused in an optimal solution.

\begin{lemma}
There exists an optimal partitioning where less than $\frac{2k}{1+\epsilon}$ of the parts are non-empty.
\end{lemma}

\begin{proof}
Assume we have at least $p_0 \geq \frac{2k}{1+\epsilon}$ non-empty parts, and consider the two smallest ones $P_1$, $P_2$; then we must have $|P_1|+|P_2| \leq 2 \cdot \frac{n}{p_0} \leq (1+\epsilon) \cdot \frac{n}{k}$. This means that we can merge $P_1$ and $P_2$ into the same part and still satisfy the balance constraint. Merging can only reduce the cost or leave it unchanged, so the solution remains optimal.
\end{proof}

This means that already for $\epsilon=1$, one of our processors will remain idle, and the usage of the computational resources further decreases for even larger $\epsilon$. On the other hand, with a small enough $\epsilon$, the balance constraint is already strict enough to ensure that every part is non-empty.

\begin{lemma}
Having $\epsilon< \frac{1}{k-1}$ ensures that every part is non-empty.
\end{lemma}

\begin{proof}
Such a small $\epsilon$ implies that $(k-1) \cdot (1+\epsilon) \cdot \frac{n}{k} < n$, i.e.\ $(k-1)$ parts of maximal size are still not enough to cover the entire graph.
\end{proof}

\subparagraph*{Blocks.} One of the fundamental tools in our construction is a \textit{block} of a specific size $b \geq 2$. A block $B$ consists of $b$ nodes $v_1, \ldots, v_b$, and $b$ distinct hyperedges of size $(b-1)$ each, such that the $i$-th hyperedge contains every node from $B$ except for $v_i$. The main role of blocks in our construction is to behave as basic building blocks that are essentially unsplittable.

\begin{lemma}
If a block $B$ of size $b$ intersects with more than one partition, then the given partitioning has cost at least $(b-1)$.
\end{lemma}

\begin{proof}
Consider a color (say, red) that appears at least two times in $B$, and a node $v \in B$ that has a different color from this (say, blue). All the $(b-1)$ hyperedges containing $v$ have both a blue and a red node in them, so they each induce a cost of $1$ at least. On the other hand, if no color appears twice in $B$ at all, then each hyperedge induces a cost of at least $1$.
\end{proof}

In our constructions, we usually select $b$ such that $(b-1)$ is larger than the cost of some trivial partitionings that are straightforward to find in our construction. This will imply that any reasonable algorithm must color each of the blocks in the hypergraph monochromatically (if this is possible at all): otherwise, we can trivially improve the solution, by simply replacing it with an arbitrary partitioning that does not split any of the blocks.

With a slight abuse of terminology, we will sometimes also use the word block to refer to the slightly different gadgets that serve the same purpose (they are suboptimal to split) in special classes of hypergraphs, e.g.\ hyperDAGs.

\subparagraph*{Number of hyperedges.} Note that when we express our positive results as a function of $n$, this is a slight abuse of notation. That is, we implicitly assume that the hypergraph has reasonable size, i.e.\ the number of hyperedges is polynomial in $n$; otherwise, it might not even be possible to read the input in time that is polynomial in $n$. For larger hypergraphs, we can reinterpret these results such that the corresponding running time is a function of the input size (in bits, in a classical complexity-theoretic sense) instead of the number of nodes to extend our claims to this case.

\subparagraph*{Non-integer thresholds.} When discussing our constructions, it is often convenient to assume that the threshold $(1+\epsilon) \cdot \frac{n}{k}$ of the balance constraint is an integer value. As such, we sometimes omit the corresponding floor function when providing a high-level overview in our proofs.

We mentioned in Section \ref{sec:prelim} that the constraints are sometimes relaxed to a looser threshold of $\lceil (1+\epsilon) \cdot \frac{n}{k} \rceil$ to ensure that a feasible partitioning always exists. Note that such a relaxation is rarely necessary: if we have $\epsilon >0$, then $\lfloor (1+\epsilon) \cdot \frac{n}{k} \rfloor \geq \frac{n}{k}$ is already ensured for $n \geq \frac{k}{\epsilon}=O(1)$, so this mostly plays a role in the $k$-section problem when $\epsilon=0$. Most of our proofs are straightforward to adapt to this slightly different problem definition with the relaxed constraint. This question is most relevant in the context of the layer-wise constraints for hyperDAGs: in computational DAGs, we can easily have some very small layers where this relaxation is indeed necessary. Alternatively, we could also modify the problem formulation in the layer-wise case, and decide to ignore these degenerate layers, only imposing a balance constraint on layers above a specific size.

\section{Detailed discussion of hyperDAGs} \label{app:hyperDAGs}
This section discusses the hyperDAG representation of computational DAGs in more detail. Recall that given a DAG $G$, our hyperDAG has a hyperedge for every node $u$, containing both $u$ and all the immediate successors of $u$. For simplicity, we disregard hyperedges of size $1$, i.e.\ we do not add a hyperedge for $u$ if $u$ does not have any outgoing edges; such degenerate hyperedges do not have any effect on the partitioning problem anyway. This means that our resulting hyperDAG has exactly $n-|V_{sink}|$ hyperedges, where $V_{sink} \subseteq V$ denotes the sink nodes of the computational DAG.

We point out that a similar approach to convert computational DAGs to hypergraphs has already been suggested in the work of Hendrickson and Kolda \cite{HK00}; in this approach, the hyperedge corresponding to a node $u$ contains both the immediate predecessors and the immediate successors of $u$. However, the hyperedges in this model do not corresponds to specific units of data, and as such, the cut size can significantly overestimate the actual communication costs, similarly to models with simple graphs. In particular, consider a DAG with $(k-1)$ source nodes and $m$ sink nodes, where every source has a directed edge to every sink. Assume the sinks are all red, and each source takes a different one of the other $(k-1)$ colors. In this case, each hyperedge corresponding to a sink node will induce a connectivity cost of $(k-1)$, and hence the total cost is at least $m \cdot (k-1)$. In contrast to this, the actual cost of data transfer (accurately modeled by our hyperDAGs) is only $(k-1)$: the value of every source node needs to be transferred to the red processor. 

\subsection{Characterization of hyperDAGs}

Recall that if the generating node of each hyperedge is specifically marked (i.e.\ we know which hyperedge corresponds to which node), then we can easily check if a hyperDAG corresponds to a valid computational DAG: we simply convert each hyperedge back into directed edges, and verify whether the resulting directed graph is a DAG.

On the other hand, if the generating nodes are unknown, then it is not so straightforward to decide whether a hypergraph can be obtained as a hyperDAG. We have seen a simple example in Figure \ref{fig:non-hyperDAG} for a hypergraph that cannot be obtained as a hyperDAG from any computational DAG. In fact, one can observe that Figure \ref{fig:non-hyperDAG} does not satisfy several necessary conditions for hyperDAGs. For instance, any hyperDAG must have at least one node of degree $1$, since any source node of the original DAG becomes a degree-$1$ node after the transformation. Also, any hyperDAG must satisfy $|E| \leq n-1$: each hyperedge must have a distinct generator node, and the DAG must also have a sink node that does not generate a hyperedge.

Furthermore, even if a hypergraph can be obtained as a hyperDAG, it might be obtained in multiple different ways, i.e.\ the same hyperDAG might correspond to several non-isomorphic computational DAGs. E.g.\ if we consider a hypergraph on $3$ nodes with two distinct hyperedges of size $2$, then this can be obtained from two different computational DAGs: a directed path of length $2$, or a DAG with two source nodes and a single sink.

For a complete characterization of hyperDAGs (i.e.\ a necessary and sufficient condition), the key concept is that of potential source nodes, i.e.\ nodes of degree $1$: we need to have such a node in every induced subgraph. By induced subgraph in a hypergraph, we mean a subset of nodes $V_0 \subseteq V$ and the hyperedges $e \in E$ such that $e \subseteq V_0$.

\begin{lemma} \label{lem:hyperDAG1}
A hypergraph $G$ is a hyperDAG if and only if every induced subgraph of $G$ has a node of degree at most $1$.
\end{lemma}

\renewcommand*{\proofname}{Proof.}
\begin{proof}
One direction of this lemma is straightforward: every subgraph of a hyperDAG corresponds to a subgraph in the original computational DAG. However, if a node $u$ has degree $\geq 2$ in a hyperDAG, then one of the two adjacent hyperedges provides an incoming edge to $u$ in the original DAG, and hence $u$ cannot be a source node of the subDAG. As such, if we have a subgraph in our hypergraph where all nodes have degree $\geq 2$, then the corresponding subDAG can have no source nodes, which is a contradiction.

On the other hand, if the property is satisfied, then we can always create a corresponding computational DAG iteratively. In each step, let us select a node of degree $1$ in our hypergraph, and make it the generating node of its (only) incident hyperedge. We can then remove both the node and the hyperedge from our hypergraph (and additionally any other nodes that have degree $0$ after this step). The property in the lemma ensures that we can always find a node of degree $1$ in this process, since otherwise the remaining set of nodes would form an induced subgraph with all degrees $\geq 2$. The directed graph obtained from this process is indeed a DAG, since the order of removing the nodes is a valid topological ordering.
\end{proof}

In fact, this also provides an efficient algorithm to recognize whether $G$ is a hyperDAG; intuitively, we can greedily remove degree-$1$ nodes and their incident hyperedges, and $G$ is a hyperDAG if and only if we can do this (i.e.\ always find a degree-$1$ node) until all hyperedges are removed.

\begin{lemma} \label{lem:hyperDAG2}
It can be decided in linear time whether a hypergraph $G$ is a hyperDAG.
\end{lemma}

\begin{proof}
This iterative node removal also gives a simple polynomial time algorithm to recognize if a hypergraph is a hyperDAG (or find a violating subset). If the iterative process fails to finish, then the hypergraph induced by the remaining nodes has degrees $\geq 2$, and hence our original hypergraph is not a hyperDAG.

The above algorithm can indeed be executed in time that is linear in the input size (the number of pins $\rho$) with the appropriate data structure. Consider a linked list representation of the incident edges for each node $v \in V$. For each hyperedge $e \in E$, we maintain a list of pointers to the list items that represents $e$ in the incidence list of each $v \in e$; this allows us to delete a hyperedge $e$ in $O(|e|)$ time. Furthermore, consider a vector with indices from $0$ to $\Delta$ where the $i$-th entry maintains all the nodes of degree $i$ as a linked list; this allows us to find a node of degree $1$ in $O(1)$ time in each step. Finally, for each $v \in V$, we have a distinct pointer to the list item in the corresponding degree list that represents $v$; this allows us to update the degrees in $O(1)$ time.

Using this structure, we can find a potential source in $O(1)$ time in each step, identify the incident hyperedge $e$, and remove $e$ from the incidence list of each other node (also updating the degrees). Whenever no nodes of degree $1$ remain, we can check whether any hyperedges have remained in the hypergraph to conclude whether we have a valid hyperDAG. The process removes each hyperedge at most once, so the running time is linear in the number of pins.
\end{proof}

Note that the iterative process also shows that the $i$-th smallest degree in a hyperDAG is at most $i$. As such, for the densest possible hyperDAG on $n$ nodes $v_1, \ldots, v_n$, we need to add a hyperedge for all $i \in \{ 1, \ldots, n-1\}$ which contains the nodes $\{ v_i, \ldots, v_n\}$. This results in a degree sequence of $(1, 2, \ldots, n-2, n-1, n-1)$.

\subsection{HyperDAG partitioning is NP-hard}

Finally, we also prove for completeness that the partitioning problem remains NP-hard if the inputs are restricted to hyperDAGs. Note that this claim seems like a weaker version of Theorem \ref{th:main}. However, strictly speaking, it is still an independent statement that does not follow directly from Theorem \ref{th:main}, since in contrast to the theorem, this claim also holds without assuming ETH.

\begin{lemma} \label{lem:hyperDAG_NPhard}
The partitioning problem is still NP-complete if we restrict the input to hyperDAGs.
\end{lemma}

\renewcommand*{\proofname}{Proof.}
\begin{proof}
Consider an instance $G(V,E)$ of the partitioning problem for general hypergraphs, and assume we want to decide if there exists a cut of cost $L$ in this hypergraph. In case of hyperDAGs, the role of our block gadgets will be fulfilled by the densest possible hyperDAGs (i.e.\ a hyperDAG of degree sequence $(1, 2, \ldots, m-1, m-1)$) discussed before. We replace each original node $v \in V$ by such a ``hyperDAG block'' on $m$ nodes for some large parameter $m$. In case of each hyperedge $e \in E$, for all nodes $v \in e$, we only include the last node of the hyperDAG block corresponding to $v$. Finally, in each hyperedge $e$, we insert an extra node (let us call it a \textit{light node}). Note that this is indeed a hyperDAG: the light nodes can be chosen as the generator nodes for each original hyperedge $e \in E$, and the hyperDAG blocks have designated generator nodes for each hyperedge. The new number of nodes is $n'=m \cdot |V| + |E|$.

Assume that our original balance constraint is $\epsilon>0$. Then we will define a new balance constraint parameter such that $(1+\epsilon') \cdot \frac{n'}{k} = m \cdot \lfloor (1+\epsilon) \cdot \frac{|V|}{k} \rfloor +|E|$, i.e.\ we select
\[ \epsilon'=\frac{(1+\epsilon) \cdot m \cdot |V|+k \cdot |E|}{m \cdot |V|+|E|}-1 \, ; \]
this indeed provides an $\epsilon'>0$ if we have $\epsilon>0$, assuming that $m$ is chosen large enough such that $m > (k-1) \cdot \frac{|E|}{\epsilon \cdot |V|}$. Intuitively, this ensures that (i) we can only put $(1+\epsilon) \cdot \frac{|V|}{k}$ hyperDAG blocks into any partition, and (ii) we can put the light nodes in any of the parts.

If we have a solution to the original partitioning problem, then a solution of the same cost also exists in our derived problem: we can place place the entire hyperDAG block of each node $v$ into the original part containing $v$, and we can then place the light node of a hyperedge $e$ into any part that intersects $e$. This satisfies the balance constraint, and each hyperedge induces the same cost as in our original partitioning.

On the other hand, assume there is a partitioning of size at most $L$ in our derived hyperDAG. Note that $L \leq (k-1) \cdot |E|$, otherwise the problem is trivial. Let us define an $m_0$ large enough such that $m_0 > L \cdot |V|+|E|$. We then select $m=m_0+L$ for our parameter $m$. We claim that the last $m_0$ nodes of every block must be in the same partition: these $m_0$ nodes induce $(m_0-1)$ edges already, so splitting them to multiple parts would induce a cost at least $(m_0-1)>L$.

Now consider the partitioning in the original hypergraph where each $v \in V$ is placed in the same part as the last $m_0$ nodes of the corresponding block in our hyperDAG; the cost of this solution is at most as much as the cost of the original partitioning in our hyperDAG. Furthermore, the solution also satisfies the balance constraint in our original hypergraph. Assume for contradiction that the balance constraint is violated, i.e.\ there are at least $\lfloor (1+\epsilon) \cdot \frac{|V|}{k} \rfloor +1$ nodes in a partition. This implies that our hyperDAG also had at least $(\lfloor (1+\epsilon) \cdot \frac{|V|}{k} \rfloor+1)\cdot m_0$ nodes in a partition. However, recall that $m_0=m-L$ and $m > L \cdot (|V|+1)+|E|$, so we have
\begin{gather*}
(\lfloor (1+\epsilon) \cdot \frac{|V|}{k} \rfloor + 1) \cdot m_0 \, > \, \lfloor (1+\epsilon) \cdot \frac{|V|}{k} \rfloor \cdot m + m -(|V|+1) \cdot L \, > \\
> \, \lfloor (1+\epsilon) \cdot \frac{|V|}{k} \rfloor \cdot m + |E| \, = \, (1+\epsilon') \cdot \frac{n'}{k} \, ,
\end{gather*}
which contradicts the fact that our hyperDAG partitioning was balanced.

Since hypergraph partitioning is NP-hard for any $\epsilon'>0$, this completes the reduction for any $\epsilon>0$. Note that $m=O(|E| \cdot |V|)$, so the number of nodes in our hyperDAG constructions is only $n'=m \cdot |V|+|E| =O(|E| \cdot |V|^2)$.

To extend the proof to $\epsilon=0$, we can apply the approach discussed in Lemma \ref{lem:bisection_reduction}.
\end{proof}

\section{Proof of the main theorem} \label{app:inapprox}

We present the proof of Theorem \ref{th:main} in three separate parts. We first discuss the proof for general hypergraphs. We then show how to convert the construction first into a hypergraph with $\Delta=2$, and then into a hyperDAG. Also, we first assume for convenience that $k=2$, and we discuss the generalization to $k \geq 3$ after the proof.


\subsection{Theorem \ref{th:main}: general case}

We prove Theorem \ref{th:main} through the following reduction.

\begin{lemma} \label{lem:spes_reduct}
If there exist a polynomial-time approximation algorithm for the $\epsilon$-balanced hypergraph partitioning problem to a factor $\alpha(n)$ (some function of $n$), then there also exists a polynomial-time approximation algorithm for the S$p$ES problem to a factor $\alpha(c_0 \cdot n^3)$, for some constant $c_0$.
\end{lemma}

\renewcommand*{\proofname}{Proof}
\begin{proof}
The main idea of the reduction has been outlined in Section \ref{sec:main}. Assume we have an instance of S$p$ES, i.e.\ a graph $G(V,E)$ (with $n=|V|$) and an integer $p$. We then create a block $B_e$ for each $e \in E$, and a node $b_v$ for each $v \in V$. We set the size each block $B_e$ to a parameter $m \geq n+1$; as such, splitting any of the blocks has a cost of at least $n$.

To model the structure of the original graph, for each $v \in V$, we add a hyperedge that contains (i) the node $b_v$, and (ii) for all $e \in E$ such that $e$ is incident to $v$, it contains an arbitrary node from the group $B_e$. We will refer to these kind of hyperedges as the \textit{main hyperedges}.

Finally, we create two more even larger blocks $A$, $A'$ to establish the role of the two colors. For each original node $v \in V$, we add $m$ distinct hyperedges that contain (i) an arbitrary node from $A$, and (ii) the node $b_v$. As such, in order to have a cut cost below $m$, we will need to ensure that all the nodes $b_v$ have the same color as $A$. 

We choose the sizes of $A$ and $A'$ carefully, such that the following two conditions hold. Firstly, we will ensure that $|A|+|A'|>(1+\epsilon) \cdot \frac{n'}{2}$ for our final number of nodes $n'=|A|+|A'|+|E| \cdot m + n$. This implies that $A$ and $A'$ will receive different colors in any reasonable solution; let us define the color of $A$ as blue, and the color of $A'$ as red. Furthermore, we also ensure that $|A'|+p \cdot m = (1-\epsilon) \cdot \frac{n'}{2}$; intuitively speaking, this means that we need to color at least $p$ of the groups $B_e$ red in order to satisfy the balance constraint.

Altogether, the hypergraph ensures that in any reasonable solution (with cost below $m$), the blocks are all uncut, the nodes $b_v$ are all blue, and at least $p$ of the groups $B_e$ are red; hence only the main hyperedges can be cut. Since $b_v$ is blue, each main hyperedge is cut exactly if at least one of the groups $B_e$ of the edges incident to $v$ is colored red. As such, the problem corresponds to choosing at least $p$ edges of the original graph (to color red) such that the number of nodes covered by these edges is minimal, which completes the reduction. Note that $n' \leq O(1) \cdot n^2 \cdot m = O(n^3)$, so an $\alpha(n')$-approximation for partitioning would also provide an $\alpha(c_0 \cdot n^3)$-approximation for S$p$ES.

For a more formal description of the reduction, let $\textsc{OPT}_{SpES}$ and $\textsc{OPT}_{part}$ denote the optimum of the original S$p$ES and the derived hypergraph partitioning problem, respectively. One can observe that $\textsc{OPT}_{SpES} = \textsc{OPT}_{part}$. Indeed, if we take any S$p$ES solution and color the corresponding $B_e$ red, we get a balanced partitioning with the same cost. On the other hand, in any partitioning that colors $A'$ and $p$ edge gadgets red (and everything else blue), we have a cost of at most $n < m$. Hence the optimal partitioning must be reasonable: it cannot split any of the blocks, it must color all the $v_b$ blue, and at least $p$ of the $B_e$ red. Given such an optimum, taking $p$ of the red edges (corresponding to the red $B_e$) provides an S$p$ES solution with at most the same cost.

Assume we have an algorithm that returns a partitioning of cost at most $\alpha(n') \cdot \textsc{OPT}_{part}$. If this solution has cost larger than $n$, we can replace it by any reasonable solution (selecting $p$ arbitrary edges); if the number of red $B_e$ is larger than $p$, we can recolor some of them to blue. These steps do not increase the cost or violate the balance constraint. We can then convert the partitioning into an S$p$ES solution of the same cost, i.e.\ cost at most $\alpha(n') \cdot \textsc{OPT}_{part} \leq \alpha(c_0 \cdot n^3) \cdot \textsc{OPT}_{SpES}$. 

It remains to discuss the size of $A$ and $A'$ in more detail. Let $s=|E| \cdot m + n$, and let us select $n'$ large enough such that $s < (1-\epsilon) \cdot \frac{n'}{2}$. This already implies $|A|+|A'|> (1+\epsilon) \cdot \frac{n'}{2}$, and also, note that $n'=O(s)=O(n^3)$. Furthermore, we need to select $|A'|$ such that $p \cdot m \leq (1-\epsilon) \cdot \frac{n'}{2} - |A'| < (p+1) \cdot m$ holds, e.g.\ by setting $|A'|=\lfloor (1-\epsilon) \cdot \frac{n'}{2} \rfloor - p \cdot m$. Finally, let $|A|=n'-s-|A'|$.
\end{proof}

\noindent Theorem \ref{th:main} follows easily from this reduction.

\renewcommand*{\proofname}{Proof of Theorem \ref{th:main}}
\begin{proof}
Given a $n^{1/(\log\log{n})^{\delta'}}$-approximation algorithm for partitioning for some constant $\delta'$, Lemma \ref{lem:spes_reduct} gives a $c_0 \cdot n^{3/(\log\log{c_0 \cdot n^3})^{\delta'}}$-approximation for S$p$ES. However, we know that there is no $n^{1/(\log\log{n})^{\delta}}$-approximation for S$p$ES for a specific $\delta>0$ if ETH holds \cite{ETHhardness}. If we select $\delta'>\delta$, then we have
\[ c_0 \cdot n^{3/(\log\log{c_0 \cdot n^3})^{\delta'}} \, < \, n^{4/(\log\log{n})^{\delta'}} < n^{1/(\log\log{n})^{\delta}} \, \]
for $n$ large enough. This contradicts to the inapproximability result on S$p$ES.
\end{proof}

\subsection{Conversion to $\Delta=2$}

When adjusting the construction to $\Delta=2$, the key challenge is to achieve the functionality of the blocks in this case. That is, we need to develop densely interconnected subgraphs that have degree $\leq 2$, but still induce a large cost in case they are split.

\begin{definition}
A \emph{grid gadget} is an $\ell \times \ell$ grid of nodes, where nodes in each row and column form a hyperedge (of size $\ell$).
\end{definition}

\begin{lemma} \label{lem:gridcost}
Consider a $2$-way (not necessarily balanced) partitioning of a grid gadget. If there are at least $t_0$ occurrences of the less frequent color in the grid, then the cut cost is at least $\sqrt{t_0}$.
\end{lemma}

\renewcommand*{\proofname}{Proof.}
\begin{proof}
Assume the minority color is red, and let $h_c$ and $h_r$ be the number of columns and rows that are entirely red, respectively. Note that $h_c+h_r<\ell$, since red is the minority color.

If $h_c \geq 1$ and $h_r \geq 1$, then all other (i.e.\ not entirely red) rows and columns contain a blue node, so they are cut; the number of cut hyperedges is then $2 \ell - (h_c+h_r) \geq \ell \geq \sqrt{t_0}$. If $h_c \geq 1$ and $h_r =0$, then all rows have blue node, so they are cut; the number of cut hyperedges is again at least $\ell \geq \sqrt{t_0}$. The same holds if $h_c = 0$ and $h_r \geq 1$.

Finally, if $h_c = 0$ and $h_r = 0$, then each row or column containing a red node is cut. To minimize the number of rows/columns with a red node, we can place them in a $\sqrt{t_0} \times \sqrt{t_0}$ square shape; this way the number of cut hyperedges is $2 \cdot \sqrt{t_0}$.
\end{proof}

In our construction, we will replace all the blocks by a grid gadget. Consider the number of nodes of the minority color in each of the grid gadgets, and let $t$ denote the sum of these (i.e.\ the total number of minority nodes) over all the grid gadgets. It follows easily that the cut cost in our construction is at least $\sqrt{t}$.

\begin{lemma} \label{lem:totalminority}
Given several grid gadgets with a total of $t$ minority nodes, the cut cost is at least $\sqrt{t}$.
\end{lemma}

\begin{proof}
Let the number of minority nodes be $t_1, t_2, \ldots$ in the different grids. According to Lemma \ref{lem:gridcost}, the total cut cost is at least $\sqrt{t_1}+\sqrt{t_2}+\ldots$. With $t_1+t_2+\ldots=t$, this is minimized if we have all the minority nodes in a single grid ($t_1=t$, $t_2=0$, \ldots) since the square root function is concave. This induces a cost of at least $\sqrt{t_1}=\sqrt{t}$.
\end{proof}

This implies that similarly to before, we can assume that any reasonable partitioning of our construction has $t \leq n^2$; otherwise it induces a cost of more than $n$, and hence we can replace it with any reasonable solution that colors each grid monochromatically.

As a next step, we define \textit{extended grids}. Given an $\ell \times \ell$ grid gadget, we add $\ell_0$ so-called \textit{outsider} nodes to the grid (for some $\ell_0 \leq \ell$), and include the $i$-th outsider node in the hyperedge corresponding to the $i$-th row. Note that each node still has degree at most $2$ after this, and our observations about minority-colored nodes in the original $\ell \times \ell$ part of the grid also remain true.

\begin{lemma} \label{lem:recolor}
Consider an extended grid where red is the minority color in the original $\ell \times \ell$ part of the grid, and every outsider node has degree at most $2$. If we recolor the entire extended grid to blue, then the total cost does not increase.
\end{lemma}

\begin{proof}
After recoloring, all the hyperedges within the grid will be uncut. The only (originally uncut) hyperedges that can become cut in the process are the hyperedges incident to outsider nodes that were originally red. Since the outsider nodes also have degree at most $2$, they are incident to at most one hyperedge besides the grid row that they are contained in.

Let $h_c$ and $h_r$ again be the number of columns and rows in the original part of the grid that are entirely red. Similarly to the proof of Lemma \ref{lem:gridcost}, if we have either $h_c \geq 1$ or $h_r \geq 1$, then the number of conflicts is at least $\ell$. This implies that recoloring the entire grid to blue decreases the total cost by at least $\ell$, while it introduces at most $\ell_0 \leq \ell$ new conflicts.

On the other hand, let $h_c = h_r = 0$, and consider a red outsider node $v$. Since no row of the grid is entirely red, the row containing $v$ is cut. Hence recoloring each such row to blue decreases the cost by $1$, and increases the cost (through the other hyperedge incident to $v$) by at most $1$.
\end{proof}

\begin{lemma} \label{lem:spes_degree2}
If stated with $\alpha(c_0 \cdot n^4)$ instead of $\alpha(c_0 \cdot n^3)$, Lemma \ref{lem:spes_reduct} also holds when restricted to hypergraphs of degree at most $2$.
\end{lemma}

\begin{proof}
We replace each block in our construction with an extended grid. For each $B_e$, we use an extended grid of size $\ell \times \ell$, where we set $\ell=2 \cdot n$, and we add $2$ outsider nodes. Both outsider nodes will represent one endpoint of the original edge $e$.

We turn $A$ into an extended grid of size $\ell_A$ with $n$ outsider nodes, each representing one of the nodes $v \in V$, and taking the role of the original $b_v$. For each $v \in V$, we add a hyperedge containing $b_v$ and all the outsider nodes in the incident $B_e$ that represent $v$ as an endpoint. Note that each node has degree exactly $2$ this way.

Finally, we turn $A'$ into a grid gadget of size $\ell_{A'}$. Note that altogether we have $n'=\ell_A\,\!^2+\ell_{A'}\,\!^2+n+|E| \cdot (\ell^2 + 2)$. We select $\ell_A$ and $\ell_{A'}$ such that $\ell_A\,\!^2+\ell_{A'}\,\!^2-t > (1+\epsilon) \cdot \frac{n'}{2}$. This implies that $A$ and $A'$ will again have a different majority color, otherwise either the balance constraint is violated, or we have at least $t$ minority-colored nodes altogether, so the solution has unreasonably high cost. Furthermore, we need to ensure that $\ell_{A'}\,\!^2+p \cdot (\ell^2 + 2) = \lceil (1-\epsilon) \cdot \frac{n'}{2} \rceil$. For this, let again $s=|E| \cdot (\ell^2+2) + n$ and select $n'$ large enough to ensure $s < (1-\epsilon) \cdot \frac{n'}{2}$; this already provides the appropriate values for $\ell_A\,\!^2$ and $\ell_{A'}\,\!^2$. Assume for convenience that the resulting $\ell_A\,\!^2$ and $\ell_{A'}\,\!^2$ are square numbers (we discuss this technicality after the proof). Note that $s\leq n^2 \cdot O(n^2) + n = O(n^4)$ and $|A|, |A'| \in O(s)$, so the total number of nodes is $n'=O(n^4)$ in this case.

Given a solution to S$p$ES, we can easily convert it into a partitioning of the same cost: we color $A'$ and the extended grids of the $p$ chosen edges red, and everything else blue. This only cuts the main hyperedges that contain a node $b_v$ with at least one incident red edge.

On the other hand, given a valid partitioning, we can also turn it into an S$p$ES solution with at most the same cost. Recall that we have $t \leq n^2$, otherwise we can switch to an arbitrary reasonable solution (with $p$ arbitrarily chosen red edges). One can observe that there must be at least $p$ grids $B_e$ where the majority color is red; otherwise, we have at most $\ell_{A'}\,\!^2$ red nodes in $A'$, at most $\ell^2+2$ red nodes in the red-majority $B_e$, at most $n+2\cdot |E|$ red outsider nodes, and at most $t$ red nodes as minority-colored nodes in the remaining grids. As such, the total number of red nodes is at most
\[ \ell_{A'}\,\!^2 + (p-1) \cdot (\ell^2+2) + n+2 \cdot |E| + t \leq \ell_{A'}\,\!^2 + (p-1) \cdot (\ell^2+2) + 4 \cdot n^2 \, . \]
This is less than the required number of red nodes $(1-\epsilon) \cdot \frac{n'}{2}=\ell_{A'}\,\!^2+p \cdot (\ell^2 + 2)$, since we have $\ell^2+2 > 4n^2$ due to $\ell = 2n$.

Hence there are at least $p$ grids $B_e$ with red majority. This means that we can recolor each extended grid to the majority color, since having $A'$ and at least $p$ of the $B_e$ red already ensures that we satisfy the balance constraint. Furthermore, Lemma \ref{lem:recolor} shows that recoloring all extended grids to the majority color does not increase the cost. Hence the resulting solution corresponds to an S$p$ES solution of at most the original cost, completing the reduction.
\end{proof}

Theorem \ref{th:main} then naturally extends to this case of degree at most $2$.

We point out that this method also shows the bound for the special class of SpMV hypergraphs studied in \cite{KB20}. Besides having degree exactly $2$, these hypergraphs also have a ``bipartite property'' on their hyperedges: the hyperedges can be partitioned into two classes $E_1, E_2$ such that any two hyperedges in the same class are disjoint. Our construction also satisfies this property: the hyperedges describing grid rows are class $1$, and the remaining hyperedges (describing grid columns or main hyperedges) are class $2$. The grids themselves are then bipartite by design, and the main hyperedges always intersect with row hyperedges of the grids only.

Note that we assumed for convenience that the resulting values $\ell_A\,\!^2$ and $\ell_{A'}\,\!^2$ are squares of integers, which is rarely the case. To overcome this, one can first select $n'$ explicitly such that $\ell_A\,\!^2 = \lfloor (1+\epsilon) \cdot \frac{n'}{k} \rfloor - (|E|-p) \cdot (\ell^2+2)-n$ is a square number (we increment our original candidate for $n'$ until this is fulfilled; note that since $\frac{1+\epsilon}{k}<1$, the corresponding $\ell_A\,\!^2$ is also incremented by at most $1$ in each step, so we can find such a square number without changing the magnitude of $n'$). We then set $|A'|=n'-\ell_A\,\!^2-s$ even if this is not a square number by modifying the grid gadget of $A'$: we add an outsider node to a sufficient number of rows \textit{and} columns to increase the size of the $\ell_{A'} \times \ell_{A'}$ square grid to the required $|A'|$ (this is always possible since we can add up to $2 \cdot \ell_{A'}$ outsider nodes). These extra outsider nodes are colored red in any reasonable solution, otherwise it can be easily improved.

\subsection{Conversion to a hyperDAG}

It only remains to further adjust the construction in order to ensure that it is a hyperDAG. This is rather straightforward. We only need to add one more outsider node to the extended grids $A$ and $A'$; this does not affect the proof discussed above.

In order to show that the resulting construction is a hyperDAG, we only need to select a distinct generator node for every hyperedge without creating cycles; we discuss such an assignment. For the main hyperedges, we select the corresponding node $b_v$ as the generating node. In each of the extended grids, we first select an outsider node: the newly added node for $A$ and $A'$, and an arbitrary one of the two outsider nodes in case of all $B_e$. We make this chosen outsider node the generator of the row hyperedge that it is contained in. Then we select every other node in the same row as the generator node of the column hyperedge it is contained in. Finally, we consider the first column hyperedge of the grid, and in every row (except for the already processed row where the generator is the outsider node), we select the first node of the row as the generator node of this row hyperedge. One can check that this selection of generator nodes indeed produces a hyperDAG, i.e.\ there are no directed cycles in the resulting DAG structure.

\subsection{Generalization to $k \geq 3$}

The same reduction approach also generalizes to $k \geq 3$ colors. Following the same idea, we ensure that $\lfloor (1+\epsilon) \cdot \frac{n'}{k} \rfloor = |A|+(|E|-p) \cdot m + n$, i.e.\ besides $A$ and its outsider nodes we can only color $(E-p)$ edge blocks blue, so the remaining blocks need to attain a different color. Note that there is no motivation to color the remaining edge blocks with multiple colors: in any such case, we can recolor them all to red without increasing the cost.

If we have $2 \cdot \frac{1+\epsilon}{k} > 1$, then it is not even required to add further nodes to the construction: the entire construction can be colored by two colors, so we can recolor any solution such that all the remaining (non-blue) nodes are red. Otherwise, let $k_0$ be the smallest integer such that $k_0$ parts can cover the whole hypergraph, i.e.\ $k_0=\lceil \frac{k}{1+\epsilon} \rceil$, and let us develop $(k_0-1)$ components of equal size in the remaining graph. That is, let $T_0=\frac{1}{k_0-1} \cdot (1-\frac{1+\epsilon}{k}) \cdot n'$, and first select $|A'|$ such that $|A'|+p \cdot m =T_0$, and then add $(k_0-2)$ further blocks of size $T_0$ each. The same reduction proof still holds for this construction: for any partitioning solution, we can (i) recolor edge blocks blue until $|E|-p$ of them are blue, and then (ii) recolor the remaining edge gadgets and $A'$ red, and the $i$-th extra block to color $i$, without increasing the cost in either step. Note that these changes only increase the size of the construction by a factor $k \in O(1)$ at most.

The adjustment to $\Delta=2$ and hyperDAGs also carries over to this case. Lemma \ref{lem:gridcost} and \ref{lem:totalminority} also apply if we define minority color as any color that occurs at most $\frac{1}{2} \cdot \ell^2$ times in the grid; we always have such a color unless the grid is monochromatic. Lemma \ref{lem:recolor} holds for recoloring the grid to the most frequent color (with $h_r$ and $h_c$ denoting the number of rows/columns that are monochromatic with any other color). This allows us to use the same proof approach as in Lemma \ref{lem:spes_degree2}. In particular, in a partitioning of reasonable cost, blue can only be the most frequent color in at most $(|E|-p)$ edge grids, otherwise the number of blue nodes is at least $|A|+(|E|-p+1) \cdot (\ell^2+2)-t$, which is larger than the threshold since $(\ell^2+2) \geq n+2 \cdot |E| + t$. We can then again recolor $A$ and the blue-majority $B_e$ to blue, and $A'$ and all other $B_e$ to red without increasing cost, thus obtaining an S$p$ES solution of at most the original cost.

\subsection{Different complexity-theoretic assumptions}

Note that different inapproximability results are known for S$p$ES based on different complexity-theoretic assumptions. We now quickly outline the main idea of these assumptions, and discuss the slightly stronger versions of Theorem \ref{th:main} that they imply. The discussed form of Theorem \ref{th:main} is based on ETH, which intuitively states that if $s_3$ denotes the infimum of values $\delta$ such that $3$-SAT can be solved in time $O(2^{\delta\cdot n})$, then we have $s_3>0$.

Gap-ETH essentially states that it is not even possible in subexponential time to decide whether a $3$-SAT formula is satisfiable or whether only a $(1-\delta)$ fraction of its clauses can be satisfied for some $\delta>0$. This assumption implies that S$p$ES cannot be approximated to any $n^{f(n)}$ factor, where $f$ denotes a function such that $f(n)=o(1)$ \cite{ETHhardness}. This provides an inapproximability result of the same factor to our partitioning problem.

There are also several complexity assumptions for the hypergraph version of S$p$ES, the so-called Minimum $p$-Union (M$p$U) problem \cite{DkSH}: given a hypergraph, our goal is to select $p$ hyperedges such that their union is as small as possible. Note that while this is a more general problem, our reduction in Lemma \ref{lem:spes_reduct} is in fact straightforward to extend to this case: now each block $B_e$ will have not two, but up to $n$ incident main hyperedges. Furthermore, in the extension to $\Delta=2$, this implies that each grid $B_e$ will have up to $n$ outsider nodes; however, this does not affect the proof of Lemma \ref{lem:spes_degree2} significantly. In particular, if we add exactly $n$ outsider nodes to each $B_e$ for simplicity, then we need to replace $(\ell^2+2)$ by $(\ell^2+n)$ when adjusting the size of blocks in our construction. Furthermore, if we note that having at least $t_0$ minority-colored outsider nodes already induces a cost of at least $t_0$ (i.e.\ either the given row is cut, or it is monochromatic, and then similarly to Lemma \ref{lem:gridcost}, the cost is at least $\ell$ anyway), then the number of minority-colored outsider nodes over all grid gadgets can also be upper bounded by $t$.

The work of \cite{crypt} introduces a more complex cryptographic assumption that specific kinds of one-way functions (or alternatively, pseudo-random generators) exist. This implies the inapproximability of M$p$U to a factor $n^{\delta}$ for a given $\delta>0$. Due to our reduction from M$p$U, this again provides the same $n^{\delta}$-factor bound for the partitioning problem.

Finally, the Hypergraph Dense vs.\ Random Conjecture states that we are unable to distinguish a random and an adversarially created hypergraph (of specific densities) in polynomial time; this shows the inapproximability of M$p$U to a $n^{1/4-\delta}$ factor for any $\delta>0$ \cite{DenseVsSparse2}. In this case, it becomes significant that our reduction in Lemma \ref{lem:spes_reduct} increases the number of nodes to $O(n^3)$ (or $O(n^4)$ in case of $\Delta=2$): we only get a contradiction to this result if we have a partitioning algorithm of factor $n^{1/12-\delta}$ for general hypergraphs, or of factor $n^{1/16-\delta}$ when $\Delta=2$. As such, the inapproximability results in this case are somewhat weaker than the baseline M$p$U result.

\subsection{Parameterized complexity}

We now briefly discuss the partitioning problem from a parameterized complexity perspective, with respect to the allowed cost $L$ as a parameter. We briefly summarize the intuitive properties of the parameterized complexity classes that we will mention; for more details, we refer the reader to \cite{parameterized}.
\vspace{2pt}
\begin{itemize}
    \setlength{\itemsep}{2pt}
    \setlength{\parskip}{2pt}
    \item The class W[1] is essentially the set of problems that can be represented by a combinatorial circuit of weft at most $1$. The clique problem (deciding whether a clique of size $L$ exists in a graph) is a well-known example of a W[1]-complete problem.
    \item The class XP is essentially the set of problems that can be solved in time $n^{f(L)}$ for some computable function $f(L)$. This implies that a polynomial algorithm exists for each fixed $L$, but the exponent depends on $L$.
    \item The class para-NP is defined through non-deterministic algorithms. For our purposes, it suffices to know that a problem is para-NP-hard if there exists a fixed $L \in O(1)$ such that the problem is already NP-hard for $L$.
\end{itemize}

\renewcommand*{\proofname}{Proof of Lemma \ref{lem:parameterized} (W[1]-hardness)}
\begin{proof}
The W[1]-hardness of the problem follows directly from the reduction from S$p$ES, which is known to be W[1]-hard (it is a generalization of the clique problem). Note that the parameter remains unchanged in our reduction.
\end{proof}

\renewcommand*{\proofname}{Proof of Lemma \ref{lem:parameterized} (containment in XP)}
\begin{proof}
To show containment in XP, note that if the cost is at most $L$, then the number of cut hyperedges is at most $L$. For a simple approach, we can consider all the possible subsets $E_0$ of at most $L$ hyperedges that are cut in a solution; the number of such subsets is $\binom{|E|}{0} + \binom{|E|}{1} + \ldots + \binom{|E|}{L} \leq L \cdot |E|^L$. For each such subset of hyperedges $E_0$, and then each hyperedge $e \in E_0$, we can consider all the $2^k-1$ possible subsets of the $k$ colors, and assume that only (some of) these colors appear in $e$. We will call the combination of these assumptions a \textit{configuration}: a set $E_0$ of at most $L$ hyperedges, and for each $e \in E_0$ a subset of the $k$ colors that can appear in $e$. Note that the number of cases for each possible $E_0$ is still only $2^{k^L}$, so the total number of configurations we have to consider is at most $L \cdot |E|^L \cdot 2^{k \cdot L} \leq n^{f(L)}$. Hence if we can check in polynomial time for each configuration whether a solution exists, then the problem is indeed in XP.

Firstly, observe that the cost corresponding to each configuration can be easily computed: in case of cut-net cost, each hyperedge $e \in E_0$ induces a cost of $1$, and for connectivity cost, each $e$ induces a cost of $(\lambda_e-1)$ if $\lambda_e$ colors can appear in $e$. Note that his cost is somewhat pessimistic if not all the permitted colors appear in $e$; however, we find the solutions for such cases in the simpler configuration where the colors are not even allowed. As such, we can consider this cost for each configuration, and immediately exclude the configurations where this cost is larger than $L$. For the remaining configurations, we have to decide whether a corresponding partitioning exists.

To analyze a given configuration, we can simply remove the hyperedges $E_0$ and consider the remaining hypergraph. Note that according to our assumptions, none of the hyperedges in this remaining hypergraph is cut; as such, we can essentially contract each connected component into a single node, since these nodes need to be placed into the same part anyway. Each such contracted component has two crucial properties for our problem. Firstly, the number of nodes $a_i$ in the component, since the components will need to be separated into parts such that the corresponding sum of sizes satisfies the balance constraint. Furthermore, each component has a list of possible colors that are allowed in this component, obtained as the intersection of the allowed color set for each $e \in E_0$ that intersects this component. For example, if some nodes of the component are originally contained in two removed hyperedges $e, e' \in E_0$, and (in the current configuration) $e$ is allowed to contain the colors red and blue, while $e'$ is allowed to contain the colors red and green, then since the component has to be monochromatic, the only possible color for this component is red. If any of the components has no valid color at all, then we can immediately conclude that the current configuration is not feasible.

These steps reduce the analysis of the current configuration to the following problem. We are given a set of integers $a_1, \ldots, a_h$ ($h$ being the number of components), such that $\sum_{i=1}^h \, a_i = n$, and a nonempty subset of feasible colors for each integer. Can we partition the numbers into the $k$ (numbered) parts such that (i) each $a_i$ is in a part that is feasible for $a_i$, and (ii) the sum of numbers in each subset is at most $(1+\epsilon) \cdot \frac{n}{k}$?

Fortunately, this problem can be solved by a dynamic programming approach (somewhat similarly to the $k$-way number partitioning problem \cite{multiway}). We create a table $\tau$ with $k+1$ dimensions, with the first $k$ dimensions (with indices from $0$ to $n$) describing the number of nodes we have in the specific color, and the last dimension (indices from $0$ to $h$) used to iterate through the connected components. We initialize each cell of the table to \texttt{false}, except $\tau(0,0,\ldots,0)=$\texttt{true}. Cell $\tau(s_1, \ldots, s_k, i)$ of the table indicates whether it is possible to place the first $i$ items into the parts such that the number of nodes of color $j$ is exactly $s_j$ for all $j \in [k]$. We can then fill out this table with a simple recurrence relation: if $\tau(s_1, \ldots, s_k, i)=$\texttt{true}, then for each color $j \in [k]$ such that number $a_{i+1}$ is feasible for color $j$, we can also set cell $\tau(s_1, \ldots, s_{j-1}, s_j+a_{i+1}, s_{j+1}, \ldots,  s_k, i+1)$ to \texttt{true}. Note that the number of cells in the table is polynomial in $n$, and we only execute constantly many operations for each cell; as such, the table can indeed be filled in polynomial time. In the end, we can scan the filled table, and if we have $\tau(s_1, \ldots, s_k, h)$=\texttt{true} for any cell where $s_j \leq (1+\epsilon) \cdot \frac{n}{k}$ for all $j \in [k]$, then the desired partitioning indeed exists.
\end{proof}

\section{Multi-constraint partitioning} \label{app:multi}

We now present the proofs for our claims in Section \ref{sec:multi_sub}. Note that in contrast to the ordering in the main part of the paper, we discuss these proofs before the layer-wise claims in Section \ref{sec:layers}, since many of the techniques in this section can then be conveniently adapted to the layer-wise setting.

\subsection{With $O(1)$ constraints: approximation} \label{app:multi:approx}

We first consider the multi-constraint partitioning problem in the case when the number of constraints is only $c \in O(1)$. For this setting, we show that there is a natural reduction from the multi-constraint $k$-section problem to the original (single-constraint) version of $k$-section. Unfortunately, this reduction increases the size of the graph to $n' \approx n^{(c+1)}$; as such, it does not allow us to extend the strongest known upper bounds (e.g.\ the $\widetilde{O}(\sqrt{n})$-approximation algorithm of \cite{BisectionApprox}) to the multi-constraint case. However, we will show that we can still use this method to establish some non-trivial upper bounds for specific kinds of hypergraphs in the multi-constraint case.

Note that in contrast to most of our results, this reduction only considers the $k$-section problem: our proof approach is not straightforward to generalize to $\epsilon>0$ values, and the corresponding upper bounds from \cite{BisectionApprox} are only stated for the bisection problem anyway.

Our reduction from the multi-constraint case to the standard partitioning problem is as follows.

\begin{lemma} \label{lem:O1reduction}
If there is a polynomial-time approximation algorithm for the $k$-section problem to a factor $\alpha(n)$ (some function of $n$), then for any $c \in O(1)$, there also exists a polynomial-time approximation algorithm for the $c$-constraint case of $k$-section to a factor $O(\alpha(n^{(c+1)}))$.
\end{lemma}

\renewcommand*{\proofname}{Proof.}
\begin{proof}
Assume we have an $\alpha(n)$-approximation for the standard $k$-section problem; we show how to solve $k$-section with $c$ constraints. Given an instance of the $c$-constraint problem on a hypergraph $G(V,E)$ with $n$ nodes and balance constraint classes $V_1, \ldots, V_c$, we first add $(k-1) \cdot |V \setminus \cup_{i=1}^c V_i|$ isolated nodes to the graph; these will ensure that the nodes that are not included in any balance constraint can be sorted into any part of our choice in the single-constraint case. Also, we assume for convenience that $|V_i|$ is divisible by $k$ for each $i \in [c]$ (otherwise either no solution exists, or in the relaxed version of the problem where parts of size $\lceil \frac{|V_i|}{k} \rceil$ are allowed, we can add at most $(k-1)$ more isolated nodes to each $V_i$). Let us denote the new number of nodes in our hypergraph after the addition of all these isolated nodes by $n_0$; note that $n_0 \leq k \cdot n = O(n)$. 

We consider each constraint class $V_i$ with $i \in [c]$, and we replace each node in $V_i$ by a block of size $m_i$. As in our previous proofs, we will ensure that splitting any of the blocks results in an unreasonably high cost (the design of our blocks is somewhat more technical in this case; we defer the discussion of this to the end of the proof). Apart from replacing our nodes with large blocks, our original hyperedges remain unchanged; as such, any solution that does not cut a block has the same cost as the same coloring of original nodes (instead of blocks) in the original hypergraph.

Let us select $m_i = n_0\,\!^i$ as the size of our blocks; note that this ensures $m_i=n_0 \cdot m_{i-1}$. Since $c \in O(1)$, this means that the size of our new hypergraph is altogether
\[ n' \leq n_0 \cdot m_c = n_0\,\!^{(c+1)} \leq k^{(c+1)} \cdot n^{(c+1)} =O(n^{(c+1)}) \, . \]

Now let us consider the balance constraints. First, note that any valid partitioning in the initial hypergraph also provides a valid partitioning after this transformation, since the nodes in a specific balance constraint are increased to the same size, and the extra isolated nodes can always be colored in the appropriate way. For the opposite direction, assume that a partitioning (with monochromatic blocks) fulfills the single balance constraint in the transformed hypergraph; we claim that the corresponding partitioning fulfills each balance constraint $V_i$ separately. Consider an induction in $i$, going from $c$ to $1$, and assume that balance constraints $V_c, \ldots V_{i+1}$ are already fulfilled; this means that each of these sets contain an identical number of nodes from each part, so by removing them all, the node set of the transformed hypergraph remains balanced. Recall that each block generated from $V_i$ has size $m_i$, whereas the total number of remaining nodes besides $V_i$ is at most $(n_0-1) \cdot m_{i-1} < m_i$ (assuming $|V_i| \geq 1$). Assume that the blocks of $V_i$ are not colored in a balanced way, and hence there is a color with at most $\frac{|V_i|}{k}-1$ blocks from $V_i$. However, then the total number of remaining nodes of this color is at most
\[ \left(\frac{|V_i|}{k}-1\right) \cdot m_i + (n_0-1) \cdot m_{i-1} < \frac{|V_i|}{k} \cdot m_i \, . \]
As we still have at least $|V_i| \cdot m_i$ nodes remaining in the graph, this contradicts the fact that the single balance constraint is satisfied in our transformed graph.

This shows that as long as no blocks are cut, the solutions in the original and transformed hypergraphs are in a $1$-to-$1$ correspondence with identical cost. As such, any $\alpha(n)$-approximation for the standard $k$-section problem also provides an approximation for the $c$-constraint $k$-section problem to an $\alpha(n')=\alpha(k^{(c+1)} \cdot n^{(c+1)})$ factor.

It remains to discuss a technical detail about our blocks. Note that our smallest blocks have size $n_0 \leq k \cdot n$; since the number of hyperedges can be larger than this, this does not necessarily ensure that all blocks are unsplit in an optimal solution. As such, we require a more complex block gadget for the case when we have $|E|=\omega(n)$ in our input hypergraph. Recall our assumption that the number of hyperedges is polynomial in $n$, i.e.\ $|E| \leq n^{h}$ for some small constant $h$; this implies that the cost of any solution (that does not split blocks) is also at most $(k-1) \cdot |E| \leq k \cdot n^h$. As such, in this case, we will define a block on $b$ nodes to contain every possible subset of at least $(b-h-2)$ nodes as a hyperedge. This implies that every node $v$ in any of our blocks has a degree of at least ${b-1 \choose h+1}$, since this is the number of subsets of size $(b-h-2)$ containing $v$. Since we have $m_i \geq n_0$, this means that regardless of how we split a block, we incur a cost of at least ${n_0-1 \choose h+1}$. As ${n_0-1 \choose h+1}=\Theta(n\,\!^{(h+1)}) \geq k \cdot n^h$ for $n$ large enough, this implies that any solution splitting a block is suboptimal. Note that even with these more complex blocks, the number of hyperedges in our constructions is at most $O(n'\cdot m_c\,\!^{(h+2)})$, i.e.\ still polynomial in $n$.
\end{proof}

If we combine this reduction with the $\widetilde{O}(\sqrt{n})$-approximation from \cite{BisectionApprox}, then we only get an approximation with a $\widetilde{O}(n^{(\frac{c+1}{2})})$ factor. This is larger than $n$ already for $c=2$, so this upper bound is only meaningful in hypergraphs where the number of hyperedges is significantly higher than $n$. Alternatively, the work of \cite{BisectionApprox} also provides a $\widetilde{O}(n^{\gamma})$-factor approximation algorithm for hypergraphs where either (i) every hyperedge has size at most $O(n^{\gamma})$, or (ii) every hyperedge has size at least $\Omega(n^{(1-\gamma)})$. In this case, our reduction provides an approximation of a $\widetilde{O}(n^{\gamma \cdot (c+1)})$ factor for the multi-constraint problem; if $\gamma$ is small, then this can be significantly lower than $n$ for several $c$ values. As such, the reduction indeed allows us to establish some upper bounds on restricted classes of hypergraphs.

\subsection{With $O(1)$ constraints: parameterized complexity}

We now show that with $c=O(1)$ constraints, the partitioning problem remains within XP in terms of the allowed cost $L$.

\renewcommand*{\proofname}{Proof of Lemma \ref{lem:multi_const}, second part}
\begin{proof}
For this, we require a more sophisticated version of the dynamic programming approach discussed in Lemma \ref{lem:parameterized}. We again check every configuration where at most $L$ hyperedges are cut, and contract the connected components of the remaining hypergraph into integers $a_i$.

For each of configuration, we need to check if there is a feasible partitioning of the numbers $a_i$ that satisfies all the $c$ balance constraints. For this, we modify our original dynamic programming approach to a table that has $c \cdot k + 1$ dimensions: the last dimension again iterates over the numbers (indices from $0$ to $h$), whereas the first $c \cdot k$ dimensions each correspond to a combination of a balance constraint and a color (with indices from $0$ to the size of the balance constraint). For $j \in [k]$, $j' \in [c]$, there is a dimension that describes the number of nodes of color $j$ in constraint $V_{j'}$; that is, we set $\tau(s_1\,\!^{(1)}, \ldots, s_1\,\!^{(c)}, s_2\,\!^{(1)}, \ldots, s_2\,\!^{(c)}, \ldots, s_k\,\!^{(1)}, \ldots, s_k\,\!^{(c)}, i)$ to \texttt{true} if there exists a feasible partitioning of the first $i$ integers such that for all $j \in [k]$, $j' \in [c]$, the number of nodes of color $j$ in balance constraint $V_{j'}$ is exactly $s_{j}\,\!^{(j')}$. We can fill out this table in the same way as before: if a given cell $\tau(s_1\,\!^{(1)}, \ldots, s_k\,\!^{(c)}, i)$=\texttt{true}, then we consider the $(i+1)$-th component and all the feasible colors $j \in [k]$ for this component, and we also set $\tau(s_1\,\!^{(1)}, \ldots, s_1\,\!^{(c)}, \ldots, s_j\,\!^{(1)}+I_1, \ldots, s_j\,\!^{(c)}+I_c, \ldots, s_k\,\!^{(1)}, \ldots, s_k\,\!^{(c)}, i+1)$ to \texttt{true}, where $I_{j'}$ denotes the intersection size of this $(i+1)$-th component with balance constraint $V_{j'}$. A given configuration is valid if there is a \texttt{true} cell $\tau(s_1\,\!^{(1)}, \ldots, s_k\,\!^{(c)}, h)$ in the end such that $s_{j}\,\!^{(j')} \leq (1+\epsilon) \cdot \frac{|V_{j'}|}{k}$ for all $j \in [k]$, $j' \in [c]$. Note that the size of the table is $O(n^{c \cdot k + 1})$, so each configuration can be checked in polynomial time.
\end{proof}

\subsection{General tools for our negative results}

We continue by discussing a lemma which will be useful in several proofs in the rest of the section. For simplicity, we first state the lemma for the simplest case of $k=2$. Assume that we have a specific number of \textit{fixed} red and blue nodes in our construction, i.e.\ nodes for which the rest of the construction ensures that they always take the colors red and blue, respectively. Our lemma states that one can easily fill up a balance constraint set (for any $\epsilon \geq 0$) with the appropriate number of fixed red and blue nodes in order to achieve a desired behavior on a specific set $S$.

\begin{lemma} \label{lem:enforce}
Assume we have a set of nodes $S$ and an integer $0 \leq h \leq |S|$. We can form a set $V_0$ consisting of $S$ and a specific amount of fixed red and blue nodes such that $|V_0|=O(|S|)$, and the balance constraint is satisfied in $V_0$ if and only if
\vspace{2pt}
\begin{itemize}
    \setlength{\itemsep}{2pt}
    \setlength{\parskip}{2pt}
    \item at most $h$ nodes in $S$ are colored red (in case of $\epsilon>0$),
    \item exactly $h$ nodes in $S$ are colored red (in case of $\epsilon=0$).
\end{itemize}
\end{lemma}

\renewcommand*{\proofname}{Proof.}
\begin{proof}
Let $m=|V_0|$ be the size of the balance constraint we create. First consider $\epsilon>0$. In order to fit the set $S$ into $V_0$ appropriately, we need to ensure that $\frac{\epsilon}{2} \cdot m > h$, and also that $\left(1- \frac{1+\epsilon}{2} \right) \cdot m > |S| -h$; since $\epsilon$ is a constant, this is indeed possible with a choice of $m=O(|S|)$. We then add a set $R_0$ of $\lfloor \frac{1+\epsilon}{2} \cdot m \rfloor - h$ fixed red nodes to $V_0$; note that since $\frac{\epsilon}{2} \cdot m > h$, we have $|R_0| \geq \frac{m}{2}$, so the the balance constraint can never be violated by having too many blue nodes. Finally, we add a set $B_0$ of $m-|S|-|R_0|$ fixed blue nodes to $V_0$. This ensures that the balance constraint is fulfilled exactly if at most $h$ of the nodes in $S$ are red.

For the special case of $\epsilon=0$, we ensure that $\frac{m}{2} > h$ and $\frac{m}{2} > |S| - h$, and select an even $m$ value. We then set $|R_0|=\frac{m}{2} - h$ and $|B_0|=\frac{m}{2} - (|S| - h)$. The constraint is fulfilled exactly if $h$ nodes in $S$ are red.
\end{proof}

Naturally, the same argument holds for the color blue, and hence also the variant where the constraint holds if \textit{at least} $h$ nodes are colored red/blue.

Note that the lemma is slightly different for the special case of bisection, where we must require that exactly $h$ nodes are red. When discussing our constructions, we will often focus on the general case of $\epsilon>0$, but our proof techniques also carry over naturally for this special case. In particular, by adding $h$ further isolated nodes to $S$ in the beginning, we can also turn the ``exactly'' constraint of $\epsilon=0$ to an ``at most'' constraint: whenever the number of red nodes in the original part of $S$ is less than $h$, we can color some of the isolated nodes red instead to ensure that the number of red nodes in $S$ is exactly $h$.

In many of our constructions, we will explicitly add \textit{fixed blocks} to have the sufficient number of fixed nodes. That is, assume we want to decide if there is a multi-constraint partitioning with cost $0$ in a given construction. We can add two large blocks of $m_0$ nodes each, both contained in a single hyperedge of size $m_0$, and we combine the two blocks together in a separate balance constraint. This implies (for any $\epsilon$) that the union of the two blocks must have both a red and a blue node. Hence the only way to partition the blocks without incurring any cost is to color one of the blocks red, the other one blue; w.l.o.g.\ let us define the color red as the color of the first block. We can use the nodes in the two blocks as fixed red and blue nodes in order to apply Lemma \ref{lem:enforce} on specific subsets of nodes. Note that the total size of our input sets $S$ is at most $n$, so the number of required fixed nodes is at most $O(n)$; as such, adding the fixed blocks does not change the magnitude of the number of nodes in the graph.

In the proof of Theorem \ref{th:layers}, we will also require a slightly different variant of Lemma \ref{lem:enforce} where $V_0$ already contains initially a predetermined number $m_0$ of occurrences of both colors besides $S$, and we again add further fixed red and blue nodes to the set $V_0$ besides these $2 \cdot m_0 + |S|$ original nodes. The lemma easily extends to this case: we simply need to select $m$ large enough such that $|R_0| \geq m_0$ and $|B_0| \geq m_0$, and consider these original nodes as part of the fixed node set of the given color.

Finally, we point out that similarly to Lemma \ref{lem:enforce}, we present our constructions in the rest of this section specifically for the case of $k=2$ colors for simplicity. However, the same proof techniques can be generalized to an arbitrary $k \in O(1)$ number of colors; we discuss this briefly after the proofs in Appendix \ref{app:sec:morecolors}.

\subsection{With $n^{\delta}$ constraints}

Now let us consider multi-constraint partitioning when $c \geq n^{\delta}$ for some $\delta>0$.

\renewcommand*{\proofname}{Proof of Lemma \ref{lem:multi_lin}}
\begin{proof}

We provide a reduction from 3-coloring in this case; assume we have a graph $G(V,E)$ that we need to color with colors $i \in \{ 1, 2, 3 \}$. For each node $v \in V$ with degree $deg_v$, we create $3 \cdot deg_v$ nodes altogether: for each edge $e \in E$ incident to $v$ and each color $i \in [3]$, we create a separate node labeled $w_{v,e,i}$. Furthermore, for each node $v \in V$ and each color $i \in [3]$, we create two further nodes $\hat{w}_{v,i,1}$ and $\hat{w}_{v,i,2}$. For every node $v_0 \in V$ and fixed color $i_0$, we add a hyperedge of size $(deg_v+2)$ which contains all the $w_{v_0,e,i_0}$ for all $e \in E$ incident to $v$, plus the nodes $\hat{w}_{v_0,i_0,1}$ and $\hat{w}_{v_0,i_0,2}$.

For each node $v_0 \in V$, we introduce two balance constraints: one constraint to ensure that at most one of the $3$ nodes $\hat{w}_{v_0, i, 1}$ (for $i \in [3]$) are red, and another constraint to ensure that at least one of the $3$ nodes $\hat{w}_{v_0, i, 2}$ (for $i \in [3]$) are red. We can indeed do this according to Lemma \ref{lem:enforce}. Finally, for each edge $e_0=(u,v) \in E$ and each color $i_0$, we create a separate balance constraint for the nodes of $w_{u,e_0,i_0}$ and $w_{v,e_0,i_0}$, ensuring in a similar fashion that at most one of these two nodes can be red.

To obtain a valid partitioning of the hypergraph with cost $0$, each hyperedge needs to be monochromatic. This means that we need to select exactly one of the three colors for each node $v$, i.e. color the nodes $w_{v,e,i}$ and $\hat{w}_{v,i}$ red for exactly one $i \in [3]$ in case of every $v$. Furthermore, the balance constraints on the edges ensure that if $(u,v) \in E$, then $u$ and $v$ must have different colors. Hence a valid partitioning exists if and only if the graph is $3$-colorable. The number of balance constraints is $2 \cdot n+3 \cdot |E|$, which we can upper bound by $n^3$ for $n$ large enough; let us add further isolated nodes to the graph until the number of nodes is $\hat{n}=n^{3/\delta}$. This ensures that the number of constraints is indeed less than $\hat{n}^{\delta}$ for any desired constant $\delta \in (0,1)$, but the size $\hat{n}$ of the construction is still polynomial in $n$.

This shows that it is already NP-hard to decide whether the optimal partitioning has cost $0$ or larger. Hence we cannot have a finite-factor approximation to the problem in polynomial time (unless P=NP), and the problem is para-NP-hard in terms of the allowed cost. 
\end{proof}

\subsection{With $\omega(\log{n})$ constraints}

Finally, we consider the case when the number of balance constraints is at least slightly larger than logarithmic, i.e.\ in $\omega(\log{n})$. For this case, we show that in subquadratic time (i.e. $n^{2-\delta}$) no finite approximation ratio is achievable unless we falsify SETH. We apply SETH through the Orthogonal Vectors Problem (OVP): given a set of $m$ binary vectors $A$ of dimension $D=\omega(\log{m})$, we need to decide if there are $a_1, a_2 \in A$ such that $a_1 \circ a_2 = 0$, where $\circ$ denotes a vector dot product. It is known that this cannot be decided in $O(m^{2-\delta})$ time for any $\delta>0$, assuming that SETH holds \cite{OVPhard}.

\renewcommand*{\proofname}{Proof of Theorem \ref{th:OVP_quadratic}}
\begin{proof}
Given an instance of OVP, we create a gadget on $D+1$ nodes for each vector $a_i \in A$: $D$ distinct nodes $v_i\,\!^{(j)}$ that correspond to the dimensions $j \in [D]$, and an anchor node $u_i$. Furthermore, we add a hyperedge that contains $u_i$ and every $v_i\,\!^{(j)}$ that corresponds to a coordinate of value $1$, i.e.\ all $j \in [D]$ such that $a_i\,\!^{(j)}\!=\!1$.

As discussed before, we add two fixed blocks of $m_0=O(m)$ nodes each that are guaranteed to take the colors red and blue, respectively. We also add a balance constraint that contains every anchor node $u_i$, and ensures that at least two of the anchor nodes $u_i$ need to be red to fulfill the constraint; this is possible according to Lemma \ref{lem:enforce} with the appropriate number of fixed nodes.

Finally, for each of the dimensions $j \in [D]$, we add a balance constraint which contains the nodes $v_i\,\!^{(j)}$ for every $i \in [m]$. We again use Lemma \ref{lem:enforce} to ensure that each of these dimension-wise constraints is satisfied exactly if we have at most $1$ red node among the $v_i\,\!^{(j)}$.

In the resulting graph, we need to color at least two of the anchor nodes red (and in fact, coloring more of them red is pointless: we can then simply recolor arbitrary ones to blue without violating any constraints). If we want to avoid a cut hyperedge, a red anchor node $u_i$ also means that $v_i\,\!^{(j)}$ has to be red for every $j$ where $a_i\,\!^{(j)}\!=\!1$. This is only possible if the corresponding two vectors $a_{i_1}$ and $a_{i_2}$ are orthogonal, i.e.\ if for all $j \in [D]$ we have either $a_{i_1}\,\!^{(j)}\!=\!0$ or $a_{i_2}\,\!^{(j)}\!=\!0$; otherwise the corresponding dimension-wise constraint is violated.

As such, achieving a cut cost of $0$ is only possible if there exist two orthogonal vectors in $A$. On the other hand, if two orthogonal vectors $a_{i_1}, \, a_{i_2}$ exist, then the optimum is indeed $0$: we can color the nodes connected to $u_{i_1}$ and $u_{i_2}$ red, and all other nodes in the vector gadgets blue, satisfying all constraints. Hence if an algorithm can approximate the optimal cut to any finite factor, then it can also decide whether the optimum is $0$ or not, and thus solve OVP.

It remains to show that a runtime of $O(n^{2-\delta})$ is also subquadratic in $m$, i.e.\ it translates to a runtime of $O(m^{2-\delta'})$ for some $\delta'>0$. Note that the number of nodes in the graph is $n\! =\! \Theta(m \cdot D)$, and the number of balance constraints is $c=D+2$. Hence our runtime of $O(n^{2-\delta})$ translates to $O(m^{2-\delta} \cdot D^{2-\delta})$; to ensure that this is below $O(m^{2-\delta'})$ for some $\delta'<\delta$, we only need $D^{2-\delta} \leq m^{\delta-\delta'}$. This indeed holds for $m$ large enough if we have e.g.\ $D \leq \poly \log {m}$, which is already sufficient for the OVP hardness result. As such, a finite-factor approximation in $O(n^{2-\delta})$ time for the partitioning problem would provide an algorithm for OVP in $O(m^{2-\delta'})$ time, which contradicts SETH.

Note that for any dimension function $D=g(m) \in \omega({\log{m}})$ covered by the OVP hardness result, our construction only requires $c = D + 2 \approx g(m) \leq g(n)$ constraints, and hence the reduction works for any number of constraints $c \in \omega(\log{n})$.
\end{proof}

We note for a larger additive difference between the two cases, we can connect the same nodes in the vector gadgets by multiple different hyperedges (each containing a separate auxiliary node).

\subsection{Generalization for $k \geq 3$} \label{app:sec:morecolors}

We now outline how to generalize the negative-result constructions in this section to the case of arbitrary $k \in O(1)$.

In order to apply Lemma \ref{lem:enforce} in a setting with $k \geq 3$, we also need to add a specific fixed block of each extra color to our construction. In case of $\epsilon \geq 0$, we can apply the same idea as in Lemma \ref{lem:enforce} for an extension of this approach. We now ensure $|S|<\epsilon \cdot \frac{m}{k}$ and $|S| < (1-\frac{1+\epsilon}{k}) \cdot m$. We add $\lfloor (1+\epsilon) \cdot \frac{m}{k} \rfloor - h$ fixed red nodes to $V_0$ first. We then distribute the remaining fixed nodes equally among the other $(k-1)$ colors, i.e.\ if $m_0=\lceil (1-\frac{1+\epsilon}{k}) \cdot m \rceil - (|S|-h)$, then we add $\lfloor \frac{m_0}{k-1} \rfloor$ or $\lceil \frac{m_0}{k-1} \rceil$ fixed nodes for each of the remaining colors. This ensures that the constraint is only satisfied if we have at most $h$ red nodes in $S$, and we can also never have too many nodes of the remaining colors, since $\frac{m_0}{k-1} + |S| \leq \frac{m}{k} + |S| \leq (1+\epsilon) \cdot \frac{m}{k}$.

For $\epsilon=0$, we can again ensure that $m$ is divisible by $k$, and we add $\frac{m}{k} - h$ fixed red nodes, $\frac{m}{k} - (|S|-h)$ fixed blue nodes, and $\frac{m}{k}$ fixed nodes of the remaining colors. This ensures that the constraint is only satisfied if it has exactly $h$ red and $|S|-h$ blue nodes; this version of the lemma suffices for our purposes when $\epsilon=0$.

Generalizing the fixed-color block gadget to $k \geq 3$ is also not straightforward: e.g.\ for $3$ large blocks of $m_0$ nodes each with $k=3$ and $1 \leq \epsilon < 2$, it could happen that two of the blocks are red, one is blue, and the third color is not used at all. In order to force the presence of all colors, we use a slightly different gadget. We once again create $k$ large blocks of $m_0$ nodes each, each connected by a hyperedge, but now we apply more sophisticated balance constraints. That is, for each block $i \in [k]$, we create a balance constraint on $m$ nodes that contains $\lfloor \frac{1+\epsilon}{k} \cdot m \rfloor$ nodes of block $i$ and a single node of all other blocks (i.e.\ we select $m$ such that $\lceil m \cdot (1-\frac{1+\epsilon}{k}) \rceil = k-1$). Since all nodes of block $i$ must receive the same color to avoid a cut, none of the other blocks included in this balance constraint can have the same color as the nodes in block $i$. As such, these constraints indeed ensure that all the $k$ blocks attain a different color. For $\epsilon=0$, this modification is not required at all.

Now assume we have a construction designed for $k=2$ where we need to decide if the optimum cost is $0$. Given fixed-colored nodes of each color, we can easily modify the construction to ensure that all nodes in the original part of the construction are indeed red or blue, and hence it has the same behavior even for $k \geq 3$. We simply insert $(k-2)$ extra nodes $u_1, \ldots, u_{k-2}$ into each hyperedge, corresponding to the remaining $(k-2)$ colors. Then for each $i \in [k-2]$, we combine the nodes $u_i$ over all hyperedges into a single balance constraint, and ensure with Lemma \ref{lem:enforce} that there are at most $0$ nodes of extra color $i$ in this set of nodes. If we have a cut cost of $0$, then all the nodes in any hyperedge must have the same color; due to these balance constraints on the extra nodes, this color must be either red or blue for each hyperedge of our original construction. Note that if a node $v$ of the original construction is not contained in any hyperedge, we can create a similar hyperedge containing only $v$ and the extra nodes. The method only increases the size of the hypergraph by a factor of at most $k \in O(1)$, and it only adds $(k-2) \in O(1)$ new balance constraints.

Note that applying the same method for $\epsilon=0$ is again a bit of a special case; here we only create a single extra node $u$ for all the $|E|$ hyperedges, and combine this in a balance constraint with $|E|$ fixed nodes of each of the $(k-2)$ new colors, plus $|E|$ isolated nodes. This ensures that the extra nodes cannot take any of the new colors, but each of them can be red or blue independently (if the balance is then fixed through an appropriate coloring of the isolated nodes).

Finally, note that if we ensure that $S$ only contains red and blue nodes, then we can also again apply the ``at least'' variant of Lemma \ref{lem:enforce} by requiring that $S$ contains at most $(|S|-h)$ blue nodes.

\section{Layer-wise balance constraints} \label{app:sec:layerwise}

\subsection{Proof of Theorem \ref{th:layers}}

As for Theorem \ref{th:layers}, note that we have already shown a similar result for multi-constraint partitioning in Lemma \ref{lem:multi_lin}. We now show how to extend this construction to the case of hyperDAGs, ensuring that every balance constraint in this construction corresponds to a separate layer, and hence they translate to layer-wise constraints.

\renewcommand*{\proofname}{Proof of Theorem \ref{th:layers}}
\begin{proof}
Consider our hypergraph representation $G(V,E)$ of the graph coloring problem from the proof of Lemma \ref{lem:multi_lin}, and let $n_0=|V|+|E|$. We convert this into a computational DAG as follows. We begin with $n_0$ distinct directed paths of length $c+O(1)$; each of their first $c$ layers will correspond to a balance constraint, while the last $O(1)$ layers will be used to ensure that the paths fulfill some basic properties. The $n_0$ paths each correspond to either a node or a hyperedge in the original hypergraph.

Besides this, we also add $(k-1) \cdot n_0$ extra directed paths (we will call them \textit{filler paths}) of the same length $c+O(1)$. These extra paths will play no active role in our DAG; we will use them to ensure that the number of connected components colored with each color is identical. Finally, we add $k$ more directed paths of the same length (called \textit{control paths}), which will be used as a resource of fixed nodes for the $k$ distinct colors.

As the main idea of the construction, we will replace specific layers of the directed paths by adding blocks of a specific size $m$: that is, we replace the original node in layer $i$ of the path by $m$ distinct nodes which all have an incoming edge from the $(i-1)$-th node of the same directed path (if it exists), and all have an outgoing edge to the $(i+1)$-th node of the same directed path. As a technical detail, note that whenever we replace both the $i$-th and $(i+1)$-th layer node in the path by a block, then we add a separate edge from every node in the $i$-th layer to every node in the $(i+1)$-th layer. We will then use these blocks in the specific layers of the DAG to enforce the desired balance constraints over the different connected components.

We begin by ensuring that each of the $k$ control paths has a different color, following the ideas discussed for fixed nodes in Appendix \ref{app:sec:morecolors}. In particular, we dedicate a specific layer at the end of the paths to each color $j \in [k]$, and we create blocks in this layer as follows. We add a block of size $m_1$ on the control path of color $j$, and blocks of size $m_2$ on the remaining $(k-1)$ control paths, such that the total size of the layer is (the $k$ blocks and the original nodes in the $k \cdot n_0$ paths) is $m_0=m_1+(k-1) \cdot m_2 + k \cdot n_0$. For the threshold $T=\frac{1+\epsilon}{k} \cdot m_0$, we ensure that $m_1+k \cdot n_0 \leq T$, but $m_1+m_2 > T$. This means that none of the other control paths can have the same color as control path $j$, while placing no restriction on the remaining $k \cdot n_0$ components. For a concrete choice of values that fulfill these properties in case of $\epsilon>0$, let us select $\epsilon' \in (0, \epsilon)$ such that $\epsilon' < \frac{1}{k-1} \cdot ( 1- \frac{1+\epsilon}{k} )$, and select $m_0$ such that $\epsilon' \cdot m_0 = k \cdot n_0$. We can then choose $m_1=T-\epsilon' \cdot m_0$ and $m_2=\frac{1}{k-1} \cdot (m_0-T)$. Note that we have $m_2 \leq m_1$, so the constraints in these layers is indeed satisfied if all the control paths have a different color.

For the special case of $\epsilon=0$, it suffices to have a single layer with a block of size $m_1 = \frac{n_0}{k} + 1$ on each control path; then the size of the layer is $m_0=k \cdot m_1 + n_0 = 2 \cdot n_0+k$. In this case, two control paths of the same color would already add up to $2 \cdot m_1 > \frac{2 \cdot n_0}{k} +1 = \frac{m_0}{k}$ nodes.

Having these constraints ensures that all control paths have a different color. We can then use them to provide fixed nodes for the rest of our balance constraints by placing blocks of specific size in the control-path components in specific layers.

As a next step, we use further balance constraints to ensure that among the remaining $k \cdot n_0$ connected components, we have $n_0$ occurrences of each color. Using fixed nodes generated with the control paths, we can easily ensure this according to Lemma \ref{lem:enforce}. For $\epsilon >0$, we use $2k$ layers to ensure that the number of occurrences of a specific color $j \in [k]$ among these paths is both at least $n_0$ and at most $n_0$. For $\epsilon=0$, the situation is even simpler: a single layer without blocks already ensures that each color occurs exactly $n_0$ times. Note that this still allows us to color the path corresponding to each original node and hyperedge as desired, since we can always color the filler paths appropriately to obtain $n_0$ instances of each color.

Finally, to model the actual multi-constraint construction $G(V,E)$, we assign each of the first $c$ layers to one of our original balance constraints $V_i$. In each layer $i \in [c]$, we consider all the components that correspond to a node $v \in V_i$, and we add a block of size $m=2$ in layer $i$ of the component of $v$ (i.e.\ a single node besides the original node of the path in this layer). We then use our control paths to enforce the original balance constraints in these $|V_i|$ extra nodes in each layer: we add the appropriate block sizes to each control path to ensure (according to Lemma \ref{lem:enforce}) that the extra set $V_i$ has at least (or at most) the given number of nodes from the desired color. Note that besides the $|V_i|$ extra nodes, this layer already contains $k \cdot n_0 + k$ nodes by default in the different connected components, with $n_0+1$ occurrences of each color. We design the fixed sets enforcing the balance constraints such that they also contain these nodes (as per the discussion after Lemma \ref{lem:enforce}).

We then consider each hyperedge $e \in E$, and for every node $v \in e$, we draw an edge from an arbitrary node of the component of $v$ to an arbitrary node (in a subsequent layer) of the component of $e$. Hence in order to achieve a cost of $0$, we not only have to color every component in a monochromatic fashion, but we now also have to assign the same color to the components that are connected by a hyperedge, i.e.\ every connected component of the original hypergraph $G$.

The resulting construction enforces the original balance constraints on the different connected components: the constraint in a given layer of the DAG is satisfied if and only if the extra nodes in the corresponding layer satisfy the original balance constraint $V_i$ in $G$. As such, there exists a layer-wise partitioning with cost $0$ if and only if the input graph of the construction from Lemma \ref{lem:multi_lin} is $3$-colorable. This completes our reduction, showing that it is also NP-hard to decide whether the optimum is $0$ in this layer-wise setting.
\end{proof}

Note that this construction also ensures that every node of the DAG is contained in a directed path of maximal length; as such, there is only one valid layering of the DAG. This means that the hardness result holds even in the flexible layering case, for any fixed layer-selection strategy.

\subsection{Finding the best layering}

Now consider the problem of choosing the best layering separately. We discuss a different construction to show that in the flexible case, it is even hard to find the best assignment to layers; in other words, if the cost of each possible layering is defined as the optimal partitioning cost in the multi-constraint setting with this layering, then finding the best layering (even without the partitioning part) is already a hard problem in itself.

Note that even in this flexible layering setting, we only consider layerings where the number of layers is equal to the length of the longest path in the DAG. As such, the nodes that can be sorted into multiple layers are exactly those that are not contained in any path of maximal length. 

\begin{theorem} \label{th:layering}
The best layering of a DAG cannot be approximated to any finite factor in polynomial time.
\end{theorem}

To prove this claim, we first begin with a new variant of Lemma \ref{lem:enforce}, adapted for the case when we have a variable number of nodes in the set $S$, but they are guaranteed to all be red.

\begin{lemma} \label{lem:enforce2}
Consider a set $S$ (of variable size) that only contains red nodes, and an integer $h \geq 1$. Let $F_j$ be sets of fixed nodes of color $j \in [k]$, and let $m_j=|F_j|$. Then there exist positive integers $m_j$ such that the set $V_0:= \bigcup_{j=1}^k F_j \cup S$ satisfies the balance constraint
\vspace{2pt}
\begin{itemize}
    \setlength{\itemsep}{2pt}
    \setlength{\parskip}{2pt}
    \item if and only if $|S| \leq h$ (in case of $\epsilon>0$),
    \item alternatively, if and only if $|S| \geq h$ (in case of $\epsilon>0$, assuming that we have some upper bound $h_0$ on $|S|$),
    \item if and only if $|S| = h$ (in case of $\epsilon=0$).
\end{itemize}
\end{lemma}

\renewcommand*{\proofname}{Proof.}
\begin{proof}
Assume w.l.o.g.\ that $j=1$ is the index of the color red, and consider the case of $\epsilon>0$ first with the $|S| \leq h$ condition. We first select a large enough integer $m$ such that $\frac{m+h}{k \cdot m + h} \leq \frac{1+\epsilon}{k}$ holds; this is indeed possible with $\epsilon>0$. We set $m_j:=m$ for all $j \in [k]$. We then gradually increment $m_1$ until the inequality $\frac{m_1+h}{(k-1) \cdot m + m_1 + h} \leq \frac{1+\epsilon}{k}$ still holds; since we have $\frac{1+\epsilon}{k} < 1$ and $\lim_{m_1 \rightarrow \infty} \, \frac{m_1+h}{(k-1) \cdot m + m_1 + h} = 1$, there exists a maximal integer $m_1$ for which this inequality still holds. This specific $m_1$ (and $m_j=m$ for all other $j$) indeed satisfies our requirements: the balance constraint is still satisfied if $|S|=h$, but for $|S|=h+1$, the constraint is violated since $\frac{m_1+h+1}{(k-1) \cdot m + m_1 + h +1} > \frac{1+\epsilon}{k}$.

For $\epsilon>0$, we can also have the alternative variant of this lemma where the balance constraint is satisfied if and only if $|S| \geq h$; however, in this case, we also need an external upper bound $h_0$ on the size of $S$, since a large enough $S$ will always mean that the number of red nodes is too high. This upper bound $h_0$ comes naturally when we apply the lemma in Theorem \ref{th:layering}: we can simply use the total number of nodes that allow a flexible layering at all.

Assuming that we always have $|S| \leq h_0$, we can also create a constraint that is satisfied only if $|S| \geq h$. For this, let us select a parameter $m$ that ensures $m>h$ and $m>h_0$. We then select an integer $m_2$ such that
\[ \frac{1+\epsilon}{k} \cdot \left( (k-2) \cdot m + h \right) < \left(1-\frac{1+\epsilon}{k} \right) \cdot m_2 \leq \frac{1+\epsilon}{k} \cdot (k-1) \cdot m \, ; \]
such an integer indeed exists. We also ensure with the choice of $m_2$ that we have $m_2 \geq m$; this is always possible, since $m_2=m$ is still allowed by the upper bound above. The above relations then imply that
\[ \frac{m_2}{(k-1) \cdot m + m_2} \leq \frac{1+\epsilon}{k} < \frac{m_2}{(k-2) \cdot m +m_2+h} \, . \]

We then select $m_j=m$ for all $j \in \{ 3, \ldots, k \}$. It remains to choose an appropriate value for $m_1$; we claim that there exists an $m_1 \in \{ 0, 1, \ldots, m-h \}$ that fulfills the desired properties. Indeed, there must exist a smallest integer $m'$ that satisfies $\frac{m_2}{(k-2) \cdot m + m_2 + m'} \leq \frac{1+\epsilon}{k}$, and due to our inequalities above, we have $h < m' \leq m$. We can then set $m_1=m'-h$ to satisfy the lemma: then having $m_1+h$ red nodes in the constraint is still acceptable, but having only $m_1+(h-1)$ red nodes violates it, since the number of nodes of color $2$ grows too large. Note that our choices ensure that $m_2 \geq m_j$ for all $j \neq 2$, so no other color can violate the constraint. 

Finally, for the case of $\epsilon=0$, we can simply select $m_1=1$ and $m_j=h+1$ for all other $j \in \{ 2, \ldots, k \}$.

Note that the sizes of the fixed sets is again in the magnitude of $O(h)$ (or $O(h_0)$).
\end{proof}

\noindent With this we can already move on to the DAG layering problem.

\renewcommand*{\proofname}{Proof of Theorem \ref{th:layering}.}
\begin{proof}
We show a reduction from the $3$-partition problem, which is known to be NP-hard: given a set of positive integers $a_1, \ldots, a_{3t}$ such that $b:=\frac{1}{t} \cdot (\sum_{i=1}^{3t}\, a_i )$ is an integer and we have $\frac{b}{4}<a_i<\frac{b}{2}$ for all $i \in [3t]$, can we partition the numbers into $t$ distinct triplets such that the sum of each triplet is $b$? The problem is known to be strongly NP-hard, i.e.\ NP-hard even if the input is in a unary encoding.

Given an instance of $3$-partition, we create a hyperDAG that can only be partitioned with cost $0$ if the 3-partition problem is solvable. Our DAG consists of $k$ independent connected components, with each of them being a single directed path of the same length $2t+O(1)$ initially. We use the same technique as in the proof of Theorem \ref{th:layers} to ensure (using the layers at the end of the path) that each of these components has to receive a different one of the $k$ colors.

The general idea is again to use each of the components to generate a desired number of fixed nodes in each layer; that is, as in Theorem \ref{th:layers}, the components will have a specific number of nodes in each layer $i$, and every node in the $i$-th layer will always have an edge to every node in the $(i+1)$-th layer. As such, there is only one possible layering for these parts of the components.

However, in one of the components (the ``red component''), we also add further nodes called group gadgets. For each integer $a_i$ of the $3$-partition problem, we create (i) a set of $a_i$ nodes called the \textit{first-level group} of this integer, and (ii) a set of $a_i \cdot m$ nodes (for a large parameter $m > t \cdot b$) called the \textit{second-level group} of this integer. The nodes in the first-level group of $a_i$ have no incoming edges, but each of them has an outgoing edge to every node in the second-level group of $a_i$. The nodes in the second-level group of $a_i$ all have an outgoing edge to one of the nodes in the $(2t+1)$-st layer of the red component, i.e.\ the $(2t+1)$-st node in the original directed path.

We select the sizes of the layers in each component to provide the appropriate number of fixed nodes (according to Lemma \ref{lem:enforce2}) to ensure the following: in every odd-numbered layer, the number of extra red nodes (sorted into this layer from the group gadgets) is \textit{at most} $b$, whereas in every even-numbered layer, the number of extra red nodes is \textit{at least} $b \cdot m$. E.g.\ for the simplest case of $\epsilon=0$, every other component has $b$ and $m \cdot b$ nodes in every odd and even layer, respectively (besides the single original node of the main path in each component).

To fulfill the balance constraints, we need to ensure that each odd layer of the red component has at most $b$ nodes, and every even layer has at least $b \cdot m$ nodes. We show that this is only possible by grouping the integers into groups of sum $b$ (which are then unavoidably triplets, since $\frac{b}{4}<a_i<\frac{b}{2}$). In particular, consider levels $j$ and $(j+1)$ for some odd $j$, and assume in an inductive fashion that for every $a_i$, either all the nodes of the first- and second-level groups of $a_i$ have already been sorted into previous layers, or none of them. Note that the total size of all first level groups is only $t \cdot b < m$, so we need second-level groups in order to place $m \cdot b$ nodes into layer $(j+1)$. However, in order to place nodes of the second-level group of $a_i$ into layer $(j+1)$, we need to ensure that all nodes in the first-level group of $a_i$ are already placed into layer $j$. In particular, in order to have second level groups of total size $b \cdot m$ in layer $(j+1)$, we need to place a set of first-level groups in layer $j$ that have total size $b$ at least. On the other hand, the balance constraint ensures that the maximal number of nodes we can place into layer $j$ is $b$. Hence the only way to place the appropriate number of nodes into layers $j$ and $(j+1)$ is to select a triplet of integers $a_i$ that sum up to exactly $b$, place all the nodes from the first-level groups of these integers into layer $j$, and place all the nodes from the second-level groups of these integers into layer $(j+1)$.

This shows that a layer-wise balanced partitioning with optimal cost $0$ only exists in this hyperDAG if the corresponding $3$-partition problem is solvable.
\end{proof}

\section{Brief discussion of DAG scheduling} \label{app:scheduling}

This section discusses the DAG scheduling problem and its application for developing a more sophisticated balance constraint in hyperDAGs.

The basic DAG scheduling problem (see Definition \ref{def:DAG_sched}) has been studied extensively for decades; however, for any constant $k \geq 3$, it is still a longstanding open problem whether it is solvable in polynomial time or NP-hard. On the other hand, it is known that the best makespan can be found in polynomial time for $k=2$ \cite{DAG2proc1, DAG2proc2, DAG2proc3}. Furthermore, the problem can also be solved in polynomial time for any constant $k \in O(1)$ in some special classes of DAGs; some examples are the following:
\vspace{2pt}
\begin{itemize}
    \setlength{\itemsep}{3pt}
    \setlength{\parskip}{2pt}
    \item \textit{Chain graphs:} DAGs conisisting of several dijoint directed paths, i.e.\ DAGs where every node has indegree at most $1$ and outdegree at most $1$ \cite{oppforest}.
    \item \textit{Out-trees:} DAGs where every node has indegree at most $1$, but arbitrary outdegree \cite{oppforest}. Note that chain graphs are a special class of out-trees.
    \item \textit{Level-order DAGs:} DAGs where in every connected component, the nodes are organized into layers in a way such that every node in layer $j$ has an edge to every node in layer $(j+1)$ \cite{levelorder}.
    \item \textit{Bounded-height DAGs:} DAGs where the longest directed path has length $O(1)$ \cite{boundedheight}.
\end{itemize}

As discussed in Section \ref{sec:layers}, the optimal makespan $\mu$ of the scheduling problem essentially measures how parallelizable a given DAG is. As such, this concept allows for the more sophisticated, schedule-based balance constraint of definition \ref{def:schedconstraint}.

However, if we apply this kind of balance constraint, then in order to decide whether a candidate solution (a specific partitioning of the hyperDAG) is feasible, we also need to be able to compute $\mu_p$, i.e.\ solve a different version of the scheduling problem where a fixed partitioning $p:V\rightarrow [k]$ is already part of the input, and we only want to find an assignment to time steps $t:V\rightarrow \mathbb{Z}^+$ for this partitioning that minimizes the makespan.

Unfortunately, it turns it that the fixed-partitioning version of the scheduling problem is even harder; in particular, it remains NP-hard even in the special cases where the basic DAG scheduling problem is polynomially solvable, e.g.\ for $k=2$, or the very restricted DAG classes listed above. This shows that the schedule-based constraint is not a viable approach in practice, even for $k=2$ or this special classes of DAGs, since we cannot even decide the feasibility of a solution in polynomial time (unless we change the partitioning problem significantly, requiring the output to also contain an explicit schedule).

\renewcommand*{\proofname}{Proof of Theorem \ref{th:scheduling} for out-trees and level-order DAGs}
\begin{proof}
We again show a reduction from the $3$-partition problem discussed in Theorem \ref{th:layering}.

Consider a main path of $2 \cdot t \cdot b$ nodes of alternating color blocks: it begins with $b$ blue nodes, then $b$ red nodes, then $b$ blue nodes again, and so on. Furthermore, for each $a_i$, we create a smaller path of length $2 a_i$, beginning with $a_i$ red nodes, and then having $a_i$ blue nodes. Consider an upper bound of $L=2 \cdot t \cdot b = \frac{n}{2}$ on the allowed makespan, i.e.\ we need to decide if the DAG can be parallelized flawlessly.

In order to achieve this, each time step needs to execute one node from the main path (since it has length $\frac{n}{2}$) and one node from one of the smaller paths, and these two nodes need to have different colors. This implies that we need to process the smaller paths by always computing $b$ red and then $b$ blue nodes alternatingly. In our DAG, it is only possible to process $b$ blue nodes between steps $b+1$ and $2b$ if in the first $b$ steps we have processed the small paths corresponding to a set $A_1$ of numbers that sum up to $b$. Due to $\frac{b}{4}<a_i<\frac{b}{2}$, it is guaranteed that $A_1$ consists of exactly $3$ numbers. Since the three small paths are entirely processed by step $2b$, we can continue the same argument inductively, and show that we need to have another triplet of such items between steps $2b+1$ and $4b$, $4b+1$ and $6b$, and so on. Conversely, given a partitioning into such triplets, this strategy indeed gives a scheduling with makespan $\frac{n}{2}$, so solving this scheduling problem would also provide a solution to 3-partition.

The construction outlined above is already a level-order DAG: in every connected component, each node of layer $j$ is predecessor to every node of layer $(j+1)$. It is also an example for a chain graph. If we add a common source node to this construction (of any color, and increase the upper limit on $L$ by $1$), we also obtain an out-tree. This is also an example for some more general classes of graphs (such as e.g.\ opposing forests).
\end{proof}

\noindent We provide a slightly different construction for the case of bounded-height DAGs.

\renewcommand*{\proofname}{Proof of Theorem \ref{th:scheduling} for bounded-height DAGs}
\begin{proof}
We provide a reduction from the clique problem. Given a graph $G(V,E)$ and a desired clique size $L$, let us create a blue node for each node $v \in V$, and also create a red node representing every edge $e \in E$. We draw an edge from the node representation of $v$ to the node representation of $e$ if and only if $e$ is incident to $v$. Finally, we add another four-layer component $C$ where each node is always a predecessor of every node in the next layer. The four layers have (i) $L$ red nodes, (ii) $\binom{L}{2}$ blue nodes, (iii) $|V|-L$ red nodes, and (iv) $|E|-\binom{L}{2}$ blue nodes, respectively. The resulting DAG has height $4=O(1)$, and altogether $|V|+|E|$ nodes of both colors.

Note that we cannot compute more than one node of $C$ in any step, since there are no two different-colored nodes in $C$ that are not connected by a directed path. Hence for flawless parallelization (a makespan of $|V|+|E|$), we need to compute exactly one node of $C$ in each step, since $C$ contains $|V|+|E|$ nodes. This implies that with the processor not working on $C$, we need to compute a blue node outside of $C$ in the first $L$ steps, and then a red node outside of $C$ in the next $\binom{L}{2}$ steps. This is possible if and only if there exists a clique of size $L$ in our original graph. After processing the clique, we can compute first the $|V|-L$ blue nodes representing the remaining nodes and then the $|E|-\binom{L}{2}$ red nodes representing the remaining edges in any order.
\end{proof}

\section{Hierarchical partitioning} \label{app:hier}

\subsection{Recursive partitioning}

The recursive approach has already been outlined in Section \ref{sec:hier}. In case of standard partitioning, this allows us to follow a recursive bipartitioning method instead of finding a $k$-way partitioning directly, whereas in the hierarchical case, it allows for a $d$-step approach, splitting each current part into $b_i$ further parts in balanced way in each step $i \in [d]$.

The counterexample of Lemma \ref{lem:recurse} shows, however, that this recursive approach can be a $\Theta(n)$ factor off the optimum in unfortunate cases. This even holds in the hierarchical setting, where finding a good cut in the upper levels of the tree is even more imperative: since the parameters $g_i$ are constants, the construction still shows a difference of a $\Theta(n)$ factor in this case. The construction itself has mostly been discussed in Section \ref{sec:hier} already, and its technical details are identical to our previous proofs: the squares represent blocks on $h$ nodes and $h$ hyperedges, so splitting any of them induces a cost of $(h-1)$ at least. In our current setting, we have $h=\Theta(n)$.

One can generalize this example to any hierarchy parameters $b_1$, \ldots, $b_d$; let $b'=b_2 \cdot \ldots \cdot b_{d}$. We create $(b'+1)$ distinct larger blocks of size $\frac{n}{b_1 \cdot (b'+1)}$, connected by a chain of single edges as in Figure \ref{fig:recursive}, and $(b_1-1)$ further chains that each consist of $b' \cdot (b'+1)$ smaller blocks of size $\frac{n}{b_1 \cdot b' \cdot (b'+1)}$. A recursive approach will split the construction into the $b_1$ different chains on the highest level; however, since the first chain of $(b'+1)$ large blocks will need to be split to $b'$ parts ultimately, one of the blocks will be cut in case of $\epsilon \approx 0$, resulting in a cost of $\Theta(n)$. On the other hand, we can combine each large block with a smaller block into a part of size $\frac{n}{b_1 \cdot (b'+1)}+\frac{n}{b_1 \cdot b' \cdot (b'+1)}=\frac{n (b'+1)}{b_1 \cdot b' \cdot (b'+1)}=\frac{n}{k}$, and combine $(b'+1)$-tuples of the remaining smaller blocks into groups of size $(b'+1) \cdot \frac{n}{b_1 \cdot b' \cdot (b'+1)}=\frac{n}{k}$; this is a balanced partitioning with cost $O(1)$.

\subsection{Two-step method}

Let us now discuss our claims regarding the two-step method. We first show that if both steps are executed in an optimal way, then the final cost is at most a factor $\frac{g_1}{g_d}=g_1$ away from the optimal hierarchical cost.

\renewcommand*{\proofname}{Proof of Lemma \ref{lem:twostep_approx}}
\begin{proof}
Let $OPT_{\textsc{cut}}$ denote the optimum cost of the standard partitioning problem (with the connectivity metric), and $OPT_{\textsc{hier}}$ denote the optimum for the hierarchical cost function. Since every intersection with a new part within a hyperedge comes at a cost of at most $g_1$, the solution of our two-step method has cost at most $g_1 \cdot OPT_{\textsc{cut}}$. Assume for contradiction that the cost of the solution returned by the two-step method is larger than $g_1 \cdot OPT_{\textsc{hier}}$; this implies $OPT_{\textsc{hier}} < OPT_{\textsc{cut}}$. However, this is not possible, since the solution $OPT_{\textsc{hier}}$ directly provides a solution of the same (or smaller) cost if we evaluate it according to the standard cost function. 
\end{proof}

\noindent However, there exists a construction where this difference is indeed in the magnitude of $g_1$. We show this for $\epsilon=0$ below, and then provide a discussion of extending the approach to other $\epsilon$ values.

\renewcommand*{\proofname}{Proof of Theorem \ref{th:two_step_bad}}
\begin{proof}
Let $\epsilon=0$, and let us use $T$ to denote the number of nodes that fit into a single part of our construction, i.e.\ $T=\frac{n}{k}$ for the final size $n$ of our hypergraph. We once again use blocks of $b$ nodes as a building block; for $b$ large enough, splitting any of these blocks results in a very high cost. Our construction will consist of the following blocks:
\vspace{2pt}
\begin{itemize}
    \setlength{\itemsep}{2pt}
    \setlength{\parskip}{2pt}
    \item a single block $A$ on $T$ nodes,
    \item blocks $B_1, \ldots, B_{k-1}$ on $\frac{T}{k-1}$ nodes each,
    \item blocks $C_1, \ldots, C_{k-2}$ on $T-\frac{T}{k-1}=\frac{k-2}{k-1} \cdot T$ nodes each,
    \item a single block $D$ on $\frac{T}{k-1}$ nodes,
    \item blocks $E_1, \ldots, E_{k-3}$ on $\frac{T}{k-1}$ nodes each.
\end{itemize}
\vspace{4pt}
We then add single edges between these blocks (i.e.\ hyperedges of size $2$ connecting two arbitrary nodes of two distinct blocks) in the following way:
\vspace{2pt}
\begin{itemize}
    \setlength{\itemsep}{2pt}
    \setlength{\parskip}{2pt}
    \item for $i \in [k-1]$, we draw $m$ distinct edges from $A$ to $B_i$ (for some large parameter $m$),
    \item for $i \in [k-2]$, we draw a single edge from $B_i$ to $C_i$,
    \item we draw a single edge from $B_{k-1}$ to $D$.
\end{itemize}

With this, we essentially obtain a ``star-like'' construction where almost all the edges go between the block $A$ and the blocks $B_i$. However, $A$ is already large enough to fill an entire part itself, so each of these edges will be cut in any solution (that does not split blocks).

Intuitively, the optimal solution in the hierarchical setting is to place all the $B_i$ in the same partition; this would mean that the two parts can be assigned to bottom-level siblings in the hierarchy tree, so all of these cut edges only induce a cost of $g_d=1$. After this, the remaining edges are cut anyway, so we simply need to place the remaining blocks into parts such that each of them has size $T$.

This is indeed the optimum for hierarchical cost: we have (i) one part with $A$, (ii) one part with $B_i$ for all $i \in [k-1]$, (iii) $k-3$ distinct parts, containing $C_i$ and $E_i$ for $i \in [k-3]$, and (iv) a single last part containing $C_{k-1}$ and $D$. All of these parts have size $T$ exactly. The $(k-1) \cdot m$ edges between the first two parts have cost $g_d=1$, and the remaining $O(k)$ edges have cost at most $g_{1}$, so in total the hierarchical cost is $(k-1) \cdot m + O(k)$. However, the number of cut edges altogether is $(k-1) \cdot m+(k-1)$.

On the other hand, in a standard $k$-way partitioning, the edges between $A$ and $B_i$ are cut anyway, so in order to minimize the cost, an algorithm will ensure that the remaining $(k-1)$ edges (between the $B_i$ and $C_i$) are uncut. That is, in the standard optimum, we have (i) one part with $A$, (ii) $k-2$ distinct parts, containing $B_i$ and $C_i$ for $i \in [k-2]$, and (iii) a single part containing $B_{k-1}$, $D$, and all the $E_i$ for $i \in [k-3]$. All these parts have size $T$ exactly, so the solution satisfies the balance constraint.

The standard cost in this case is only $(k-1) \cdot m$, so this is the optimal solution for the standard $k$-way partitioning problem. However, now $A$ and all the $B_i$ are in different parts; hence regardless of the way we develop the hierarchy in the second step, we will have $m$ edges going from the part of $A$ to all the other $(k-1)$ parts in the hierarchy. This means that at least $\frac{b_{1}-1}{b_1} \cdot k \cdot m$ of our edges will induce a cost of $g_{1}$. In fact, if the branching factors in the hierarchy are $b_1$, \ldots, $b_{d}$, then the exact hierarchical cost will be
\begin{gather*}
 \left( \frac{b_{1}-1}{b_{1}} \cdot k \cdot g_{1} + \frac{b_{2}-1}{b_{2}} \cdot \frac{k}{b_{1}} \cdot g_{2} + \ldots \right) \cdot m \, = \\
 = \, \sum_{i=1}^d \left( g_i \cdot \frac{b_i-1}{b_i} \cdot \prod_{j=i}^{d} b_j \right) \cdot m \, .
\end{gather*}
For the simplest case of $b_1=\ldots=b_d=2$, this evaluates to $m \cdot (\frac{k}{2} \cdot g_1 + \frac{k}{4} \cdot g_{2}+ \ldots + 1 \cdot g_d)$. Recall that the optimal hierarchical cost is $(k-1) \cdot m + O(k) \leq k \cdot m$ for $m$ large enough; as such, the solution returned by the two-step method is at least a factor $\frac{b_{1}-1}{b_{1}} \cdot g_1$ larger.

As for the technical choice of our parameters: firstly, we need to ensure that $m \geq g_1 \cdot k = O(1)$. Furthermore, we need to ensure that none of the blocks are split in a reasonable solution, i.e.\ $b-1 > g_1 \cdot (m+1) \cdot (k-1)$; with this, splitting a single block already has larger cost than the maximum cost induced by all of our $(m+1) \cdot (k-1)$ edges altogether.
\end{proof}

Recall that as $b_1$ grows larger, this ratio approaches the upper bound of Lemma \ref{lem:twostep_approx}. We also note that the construction of Theorem \ref{th:two_step_bad} can also be extended to hyperDAGs by replacing each block with a larger hyperDAG of maximal density, following similar ideas as in the proof of Lemma \ref{lem:hyperDAG_NPhard}.

We note that it is only straightforward to generalize Theorem \ref{th:two_step_bad} to a setting where the balance constraint ensures that all of the parts are non-empty, i.e.\ $\epsilon<\frac{1}{k-1}$. For any such $\epsilon$, we can indeed generalize the idea above by setting $|A|=(1+\epsilon) \cdot \frac{n}{k}$, introducing $T':=\frac{n-|A|}{k-1}$, and setting $|B_i|=|D|=|E_i|=\frac{T'}{k-1}$ and $|C_i|=\frac{k-2}{k-1} \cdot T'$. If $\epsilon$ is larger, then the same technique only guarantees a smaller factor of difference. In particular, for general $\epsilon$, we need at least $\lceil \frac{k}{1+\epsilon} \rceil$ parts to cover the whole graph, and hence we can leave $h_{\emptyset} = k-\lceil \frac{k}{1+\epsilon} \rceil$ parts empty; then the same construction approach provides an example for a factor $(\frac{b_1-1}{b_1} \cdot k - h_{\emptyset}) \cdot g_1 \cdot \frac{1}{k-h_{\emptyset}}$ of difference.

\subsection{Observations on complexity}

Note that the hardness result in our main theorem also carries over easily to the case of this more complex cost function. In particular, recall that in the generalization of the proof to $k \geq 3$ colors, the relevant part of the construction can also be colored with at most two colors. As such, if we consider this construction in the hierarchical setting, then any partitioning can be modified (without increasing cost) such that we only use red and blue in the relevant part of the construction, and these two colors are siblings in the bottom level of the tree. The cut cost in this construction is identical in the standard and the hierarchical setting, so our reduction also applies to the hierarchical case.

Also, any $\alpha$-approximation of the regular partitioning problem provides an $O(\alpha)$ approximation in the hierarchical setting due to Lemma \ref{lem:twostep_approx}: we can use the approximation algorithm in the first step of our two-step method, and then find the optimum hierarchy assignment for this partitioning (this second step only needs $O(1)$ time as long as $k \in O(1)$). This provides a solution with cost that is at most $\alpha \cdot g_1 = O(\alpha)$ times the optimum.

The parameterized complexity of the problem also remains unchanged.

\renewcommand*{\proofname}{Proof}
\begin{lemma}
Even with the hierarchical cost function, the partitioning problem is in XP (with respect to the allowed cost as a parameter).
\end{lemma}

\begin{proof}
We can follow the same proof approach as in Lemma \ref{lem:parameterized}: it still holds that if the cost is at most $L$, then the number of cut hyperedges is also at most $L$. Hence we again consider every configuration, i.e.\ every subset $E_0$ of at most $L$ hyperedges and every possible subset of the $k$ colors appearing in each hyperedge $e \in E_0$. Note that we can still compute the cost induced by each such configuration according to our hierarchical cost function, and exclude the configurations where this cost is above $L$. We can iterate through the remaining configurations, and check their feasibility exactly as in Lemma \ref{lem:parameterized}.
\end{proof}

\section{The hierarchy assignment problem} \label{app:hierass}

Next we briefly study the second step of the two-step method as a separate problem: given a fixed partitioning of the hypergraph, we want to assign the $k$ parts to the $k$ positions in the hierarchy optimally, i.e.\ such that the hierarchical cost is minimized.

\subsection{Motivation and general discussion}

The first natural question in this hierarchy assignment problem is the number of possible solutions, i.e.\ non-equivalent ways to assign the parts to slots in the hierarchy. A simple combinatorial argument shows that the number of non-equivalent hierarchy assignments is
\[ f(k) \, = \, \frac{k!}{b_{1} ! \, \cdot \, (b_{2}!) ^ {b_{1}} \, \cdot \, (b_{3}!) ^ {b_{1} \cdot b_{2}} \, \cdot \, \ldots} \: = \: \frac{k!}{\prod_{i=1}^{d} \: (b_i !) ^ {\prod_{j=1}^{i-1} \, b_j} } \: , \]
since there are $k!$ possible permutations in general, but on every level $i$ there are $b_1 \cdot \ldots \cdot b_{i-1}$ internal nodes of the tree, each of them with $b_i$ children, and any permutation of these children provides an equivalent solution. One can check that for several choices of the branching factors $b_1, \ldots, b_d$, this function $f(k)$ grows exponentially in $k$; a more detailed discussion of this formula is beyond the scope of our paper.

This function $f(k)$ can be interpreted as the difference between the size of the search space (up to symmetries) in the two versions of the partitioning problem: any regular $k$-way partitioning corresponds to $f(k)$ different solutions in the hierarchical case. This rapid growth of $f(k)$ can have a significant impact in practice; for example, it means that there are much fewer opportunities for symmetry breaking in a branch-and-bound based exploration of the search space.

Recall that with our previous assumption that $k \in O(1)$, we also have have $f(k) \in O(1)$, and hence the problem is not interesting from a complexity-theoretic perspective, since the solution space has constant size. Hence to study this problem, we briefly explore the case when $k$ is a variable part of the input. This setting might be realistic in applications where partitioning is time-critical, so instead of a fixed architecture, we apply more machines for larger hypergraphs, e.g.\ by increasing $k$ proportionally to $n$. These machines could then be e.g.\ connected over a network in the highest hierarchy level, resulting in $b_{1} \in \Theta(k)$ and $b_2, \ldots, b_d \in O(1)$ for our problem.

In general, given a hypergraph with its set of nodes already sorted into $k$ parts, we can contract each part into a single node to obtain a simplified hypergraph $G'$ on $k$ nodes as the new input of the hierarchy assignment problem. If a hyperedge was uncut in our partitioning, then it will only contain a single node after this contraction step; we can remove all such hyperedges for simplicity, since they have no effect on the cost of our hierarchy assignment. On the other hand, note that even if two hyperedges are different in the original hypergraph, it can happen that they contain they same subset of nodes after the contraction step. As such, the simplified input to our problem is in fact a multi-hypergraph, i.e.\ it can contain multiple copies of the same hyperedge. Alternatively, we can also represent $G'$ as a simple hypergraph with hyperedge weights $w_e \in \mathbb{Z}^+$ assigned to each $e \in E$.

\subsection{Complexity for two levels}

Finally, we briefly study the hardness of the hierarchy assignment problem for the simplest case of only $d=2$ levels. In this case, it turns out that the complexity depends on the choice of the branching factors $b_1$ and $b_2$ (recall that $b_1 \cdot b_2 = k$).

Note that with $d=2$, we only have two levels of the hierarchy: cuts on the bottom level induce a cost of $1$, while cuts on the top level induce a cost of $g_1 > 1$. In fact, if a hyperedge $e$ contains $\lambda_e\!^{(2)}$ different nodes altogether, then this will in any case induce a cost of $w_e \cdot (\lambda_e\!^{(2)}-1)$, plus a further cost of $w_e \cdot (\lambda_e\!^{(1)}-1) \cdot (g_1-1)$ for every data transfer over the higher level. As such, we can subtract this fixed amount of $\sum_{e \in E} \, w_e \cdot (\lambda_e\!^{(2)}-1)$ from the cost, and divide the remaining cost by $g_1'=(g_1-1)$ to obtain a simplified cost function for the problem. With this, the problem formulation is as follows: given a multi-hypergraph on $k$ nodes, can we partition the nodes into $b_1$ sets of $b_2$ nodes each such that the sum of the connectivity cost $w_e \cdot (\lambda_e-1)$ over all hyperedges is minimized?

We separate the two parts of Theorem \ref{th:step2} for clarity.

\begin{lemma}
Hierarchy assignment is polynomially solvable for $b_2=2$.
\end{lemma}

\begin{proof}
For $b_2=2$, our problem can be solved through a maximal matching. That is, we can create an edge-weighted complete graph on $k$ nodes where the weight $w_{(u,v)}$ of each edge $(u,v)$ is the number of hyperedges containing both $u$ and $v$. In the worst case (if all nodes of all hyperedges are in different parts), the total cost of a partitioning is $\sum_{e \in E} \, w_e \cdot (|e|-1)$. However, if we pair up $u$ and $v$ (as one of our parts of size $b_2=2$), then this implies that we have saved a cost of $w_{(u,v)}$ compared to this total cost. In general, if we find a weighted matching of total weight $w_0$, then the corresponding partitioning will have a cost of $\sum_{e \in E} \, w_e \cdot (|e|-1) - w_0$. As such, finding the best partitioning is equivalent to finding a maximal-weight matching, which can be done in polynomial time using Edmond's algorithm.
\end{proof}

\begin{lemma}
Hierarchy assignment is NP-hard for $b_2=3$.
\end{lemma}

\begin{proof}
For $b_2=3$, we present a reduction from the $3$-dimensional matching (3DM) problem: given a tripartite hypergraph $X$,$Y$,$Z$, and a set of hyperedges $(x,y,z) \in X \times Y \times Z$ of size $3$, our task is to find the largest 3-dimensional matching, i.e.\ the largest possible subset of hyperedges that are disjoint. The problem remains NP-hard even in the special case when $|X|=|Y|=|Z|$ and the hypergraph is $3$-regular, i.e.\ each node has degree exactly $3$ \cite{GJ79}.

Given a 3DM problem, we turn this into an instance of hierarchy assignment with $b_2=3$ on the same set of nodes. If we consider 3DM as a problem where we need to divide the nodes into triplets, and each triplet has value $1$ if it induces a hyperedge and $0$ otherwise, then this is a rather intuitive reduction: both problems require us to divide the nodes into triplets. As such, we begin by considering the same set of nodes as in our 3DM problem, and for each hyperedge $(x,y,z)$, we add a hyperedge over the same set of nodes (with weight $1$) in our hierarchy assignment problem.

There are two challenges to address in order to complete the reduction. Firstly, we need to restrict the selected triplets to the tripartite structure, i.e.\ ensure that the triplets chosen in hierarchy assignment are always selected from $X \times Y \times Z$. Moreover, we also need to model the slightly different objective function, i.e.\ a different gain value when a triplet $V_t$ intersects a hyperedge $e$. More specifically, in 3DM, if $|V_t \cap e|=2$, then this does not result in any gain, and we only have a gain of $1$ when $|V_t \cap e|=3$. In contrast to this, in hierarchy assignment with $b_2=3$, the saved cost (compared to the worst-case) is $(\lambda_e-1)$: there is already a gain of $1$ for $e$ if $|V_t \cap e|=2$, and a gain of $2$ for $e$ if $|V_t \cap e|=3$.

For the first problem, we can simply consider each triple in $X \times Y \times Z$, and add a hyperedge with a constant weight $w_0$ that contains exactly these three nodes. Since the number of triples is in $O(k^3)$, the new number of hyperedges is still polynomial in $k$. This already ensures that the optimal triplet assignment satisfies the tripartite condition: since the input hypergraph of the 3DM problem is $3$-regular, any triplet can only result in a constant gain. Hence by choosing $w_0$ large enough, we can ensure that any solution that contains a non-tripartite triplet can be improved by reorganizing the non-tripartite triplets into arbitrary tripartite triplets. Note that these auxiliary hyperedges contribute the same total gain to any feasible (i.e.\ tripartite) triplet, so they do not influence the optimality of a solution among the feasible triplets.

It remains to model the cost function of 3DM appropriately. Recall that in hierarchy assignment, each triplet $V_t$ results in a gain that we will denote by $(1,2)$: a gain of $1$ for each hyperedge that contains two nodes of the triplet, and a gain of $2$ for each hyperedge that contains all three nodes of the triplet. To correctly model the cost of 3D-matching, we would need a gain of $(0,1)$ instead: $0$ for each hyperedge with two nodes in $V_t$, and $1$ for each hyperedge with three nodes in $V_t$.

In the 3DM problem, each of the hyperedges has size $3$. Note that if we replace such a size-$3$ hyperedge by $3$ distinct hyperedges of size $2$ (connecting each pair of nodes), then we change this gain of $(1,2)$ into a gain of $(1,3)$: having two nodes of $V_t$ in an original hyperedge now provides a gain of $1$, and having all three nodes of $V_t$ in it now provides a gain of $3$.

The key observation is that we can obtain our desired gain of $(0,1)$ as $(1,3)-(1,2)$. That is, for each size-$3$ hyperedge $e$ of the 3DM problem, we (i) replace the hyperedge by three size-$2$ hyperedges connecting each pair of nodes, and (ii) we add a gain of ``$(-1,-2)$'' by adding a hyperedge for \textit{every} triplet of nodes in the hypergraph except for $e$. This indeed provides the desired gain. Intuitively, if we combine only two nodes of $e$ in a triplet, then we get a gain of $1$ (due to the size-$2$ edge) and lose an otherwise guaranteed gain of $1$ due to the fact that these two nodes have one less size-$3$ hyperedge in common. If we combine all three nodes of $e$ in a triplet, then we get a gain of $3$ from the three size-$2$ edges, and lose an otherwise guaranteed gain of $2$ due to the fact that these three nodes do not form an additionally added size-$3$ hyperedge.

Note that this is just the intuitive explanation of the construction: in practice, we can drastically simplify it by decreasing the weight of each size-$3$ hyperedge by the same amount. In other words, our construction is equivalent to the following: (i) we replace each hyperedge $e$ by three size-$2$ hyperedges connecting each pair of nodes, and (ii) we add a single size-$3$ hyperedge for every triplet of nodes that \textit{do not} form a hyperedge originally. It is clear that in this construction, the number of hyperedges is polynomial in $k$.

To compute the total gain of a solution, consider one of the $\frac{k}{3}$ triplets we choose. For any two nodes $v_1$, $v_2$  in the triplet and a node $u$ outside of it, we always have a gain of $1$: either if $(v_1, v_2, u)$ was an hyperedge originally (then we added a size-$2$ edge $(v_1,v_2)$), or if it was not (then we added $(v_1, v_2, u)$ as a new hyperedge, which has two nodes in this triplet). For any three nodes $v_1$, $v_2$, $v_3$ that were chosen as a triplet, we have a gain of $3$ if this was an original hyperedge (due to the three size-$2$ edges added), and a gain of $2$ if it was not (since we added $(v_1, v_2, v_3)$ as a new hyperedge). As such, the triplet results in a gain of $3 \cdot (k-3) + 3$ if it is an original hyperedge, and $3 \cdot (k-3) + 2$ if not. Finally, we need to consider the extra hyperedges of weight $w_0$ added for each tripartite $(x, y, z)$; if our triplet is tripartite, then these result in a total gain of $2 \cdot w_0 + 3 \cdot (\frac{k}{3}-1) \cdot w_0 = (k-1) \cdot w_0$. Altogether, this means that there is a perfect 3D-matching for $X,Y,Z$ if and only if there exists a solution with gain of at least $\frac{k}{3} \cdot \left( 3 \cdot (k-3) + 3 + (k-1) \cdot w_0 \right)= (1+\frac{w_0}{3}) \cdot k^2 - (2+\frac{w_0}{3}) \cdot k$, i.e.\ a hierarchy assignment of cost at most $\sum_{e \in E} \, w_e \cdot (|e|-1) - (1+\frac{w_0}{3}) \cdot k^2 + (2+\frac{w_0}{3}) \cdot k$.
\end{proof}

Note that for $b_2 \geq 3$ values, we can extend this approach to show NP-hardness by introducing auxiliary classes of nodes $Z_1$, $Z_2$, \ldots, $Z_{b_2-3}$, and ensuring with extra hyperedges that the $i$-th nodes of $Z$, $Z_1$, $Z_2$, \ldots, $Z_{b_2-3}$ always end up in the same $b_2$-tuple. This works for $b_2$ values of smaller magnitudes; from some point, it introduces too many auxiliary nodes into the graph, and hence a different approach is needed. In the other extreme, note that for $b_2 \in \Theta(k)$, we essentially have a regular $b_1$-way partitioning problem (with $\epsilon=0$), so NP-hardness follows easily.

\section{Further discussion of the hierarchical setting} \label{app:gen}

Finally, we briefly discuss two further questions that arise very naturally regarding our results on the hierarchical cost model.

\subsection{HyperDAGs and multi-constraint partitioning}

Note that we have analyzed the hierarchical cost function in Section \ref{sec:hier} for general hypergraphs. It is a natural question how the results in this section carry over to hyperDAGs or to a multi-constraint partitioning setting. We provide a brief discussion of this below, without going into much technical detail.

Firstly, most of the results in Section \ref{sec:hier} carry over to hyperDAGs with relative ease. The positive results (Lemma \ref{lem:twostep_approx}, or the first point of Theorem \ref{th:step2}) require no modification at all. For Lemma \ref{lem:recurse} and Theorem \ref{th:two_step_bad}, we need to show that the corresponding example constructions can also be converted to hyperDAGs; we discuss this below. 
The only claim that is not straightforward to adapt to hyperDAGs is the second point of Theorem \ref{th:step2}.

For the conversion in Lemma \ref{lem:recurse} and Theorem \ref{th:two_step_bad}, we can essentially apply similar techniques to the NP-hardness proof for hyperDAGs (Lemma \ref{lem:hyperDAG_NPhard}). In particular, both constructions consist of blocks; we replace each of these blocks by two-level blocks, i.e.\ two groups of size $b_0$ and $b_1$, respectively (for some $b_0$, $b_1$), and $b_0$ hyperedges that each contain a distinct node of the first group and all the $b_1$ nodes of the second group. The idea is then that with $b_1$ being much larger than $b_0$, the second group replaces the original blocks in the constructions, while the nodes in the first group serve as a generating nodes in order to ensure that our construction is a hyperDAG. If the further hyperedges connecting the blocks are disjoint (we can ensure this in both of the constructions), then we can select an arbitrary generating node for these hyperedges from the second groups.

In particular, in the construction of Lemma \ref{lem:recurse}, we can select $b_0=\frac{n}{36}$ and $b_1=\frac{5 n}{36}$ for the larger blocks, and $b_0=\frac{n}{72}$ and $b_1=\frac{5 n}{72}$ for the smaller blocks. This implies $(1+\epsilon) \cdot \frac{n}{4} < 2 \cdot \frac{5 n}{36}$ for $\epsilon<\frac{1}{9}$, hence the second recursive step still needs to split one of the blocks of size $\frac{5 n}{36}$ on the left side, resulting in a cost of at least $\frac{n}{36}=O(n)$. In the construction of Theorem \ref{th:two_step_bad}, we convert each original block of size $b$ into a first group of size $b_0=\frac{1}{k^2} \cdot b$ and a second group of size $b_1=\frac{k^2-1}{k^2} \cdot b$; this ensures that even if all the first groups are combined, they have a total size that is still smaller than any single second group, i.e.\ $\frac{n}{k^2}<\frac{k^2-1}{k^2} \cdot \frac{T}{k-1}$ (recall from the proof of Theorem \ref{th:two_step_bad} that $T=\frac{n}{k}$). This ensures that the second groups can still only be combined as before in order to form parts of size $T$, unless one of the second groups is split (which results in a cost of at least $b_0$, which is a significantly larger magnitude than $k \cdot m$).

Adapting the results to a multi-constraint setting is significantly simpler. In particular, the negative results in Lemma \ref{lem:recurse} and in Theorem \ref{th:two_step_bad} carry over without any changes, since single-constraint partitioning is a special case of the multi-constraint setting. The proof of Lemma \ref{lem:twostep_approx} also requires no change. Theorem \ref{th:step2} is also unaffected, since in the hierarchy assignment problem, we already begin with a fixed partitioning that satisfies the constraint(s).

Finally, the results can also be adapted to the combination of the two settings, i.e.\ layer-wise hyperDAG partitioning. When transforming the constructions of Lemma \ref{lem:recurse} and Theorem \ref{th:two_step_bad} to hyperDAGs, we simply need to ensure that the first and second groups of the blocks end up in the first and second layer of the DAG, respectively; if this holds, then it is sufficient to consider the second layer for both of the proofs. To achieve this, we only need to consider the few extra hyperedges that were previously assigned generator nodes from the second groups for simplicity, add a new artificial generator node for each of these, and besides, also add a few isolated nodes to the first layer to ensure that these new artificial nodes can be sorted into any of the parts as desired.

\subsection{Different processor topologies}

While the tree-like hierarchy of processors accurately captures the majority of today's computing architectures, it is another natural question to study how an analogue of the hierarchical partitioning problem would behave if the underlying topology for the connections between processors was different. As a straightforward generalization of our setting, we can consider an arbitrary processor topology represented by a weighted complete graph on $k$ nodes (each node representing a processor), where the weight of the edge between each pair of nodes describes the cost of communication between the given two processors. We can assume for simplicity that the weights in this graph satisfy the triangle inequality. As a natural generalization of our hierarchical cost function to this setting, we can interpret the cost incurred by a given hyperedge in our partitioning as the cost of the smallest subtree (Steiner tree) in this graph that contains the set of processors which appear in the hyperedge.

Note that some of our claims, e.g. Lemma \ref{lem:recurse}, only use the fact that the communication cost between any pair of processors is a constant; as such, these easily carry over to any topology. Lemma \ref{lem:twostep_approx} also carries over if we replace $g_1$ with the the maximum transfer cost between any pair of processors in our topology, and we again assume that the minimum transfer cost is normalized to $1$.

Theorem \ref{th:two_step_bad} is more complicated to adjust to this case. If we denote the cost of communication between processors $p_1$ and $p_2$ by $g_{p_1, p_2}$, and the set of processors by $\Upsilon$, then in our construction, the cost of the regular partitioning optimum will be proportional to $\min_{p \in \Upsilon} \, \sum_{p' \neq p} g_{p, p'}$ (i.e.\ the smallest possible total distance from one processor to all others), while the cost of the real hierarchical optimum will be proportional to $(k-1) \cdot \min_{p' \neq p} \, g_{p, p'}$ (i.e.\ $(k-1)$ times the distance between the closest processor pair), plus an additional constant cost. Theorem \ref{th:two_step_bad} will then hold with the ratio of the two expressions above, minus an arbitrarily small $\delta>0$.

The analogue of Theorem \ref{th:step2} is not straightforward to define in such a general network topology where we have no concept of `levels'. However, it is easy to see that in general topologies, the hierarchy assignment problem is NP-hard, since the second point of our original Theorem \ref{th:step2} is a special case of this.

\end{document}

%% file: pics/hyperDAG.tikz
\begin{tikzpicture}

	\begin{scope}[rounded corners = 5pt]
	\draw[gray, ultra thick, densely dotted] (-8pt,38pt) -- (-8pt,23pt) -- (34pt,2pt) -- (48pt,2pt) -- (48pt,48pt) -- (36pt,48pt) -- cycle;
	
	\draw[gray, ultra thick, densely dotted] (32pt,32pt) -- (32pt,50pt) -- (72pt,77pt) -- (89pt,77pt) -- (89pt,42pt) -- (128pt,16pt) -- (128pt,2pt) -- (73pt,2pt) -- cycle;
	
	\draw[gray, ultra thick, densely dotted] (72pt,79pt) -- (72pt,63pt) -- (115pt,32pt) -- (125pt,32pt) -- (168pt,63pt) -- (168pt,79pt) -- cycle;
	\end{scope}

	\begin{scope}[thick, arrows=-Stealth]
	\draw (0pt,30pt) -- (37pt,12pt);
	\draw (0pt,30pt) -- (36pt,39pt);
	\draw (40pt,10pt) -- (76pt,10pt);
	\draw (40pt,40pt) -- (76pt,40pt);
	\draw[thick, arrows=-Stealth] (40pt,40pt) -- (77pt,68pt);
	\draw[thick, arrows=-Stealth] (40pt,40pt) -- (77pt,12pt);
	\draw[thick, arrows=-Stealth] (40pt,40pt) -- (117pt,13pt);
	\draw[thick, arrows=-Stealth] (80pt,10pt) -- (116pt,10pt);
	\draw[thick, arrows=-Stealth] (80pt,40pt) -- (116pt,40pt);
	\draw[thick, arrows=-Stealth] (80pt,70pt) -- (116pt,70pt);
	\draw[thick, arrows=-Stealth] (80pt,70pt) -- (117pt,42pt);
	\draw[thick, arrows=-Stealth] (80pt,70pt) -- (95pt,85pt) -- (145pt,85pt) -- (157pt,73pt);
	\draw[thick, arrows=-Stealth] (120pt,40pt) -- (156pt,40pt);
	\draw[thick, arrows=-Stealth] (120pt,10pt) -- (157pt,38pt);
	\draw[thick, arrows=-Stealth] (120pt,70pt) -- (156pt,70pt);
	\end{scope}

	\draw[black, fill=white] (0pt,30pt) circle (1.0ex);
	\draw[black, fill=white] (40pt,10pt) circle (1.0ex);
	\draw[black, fill=white] (40pt,40pt) circle (1.0ex);
	\draw[black, fill=white] (80pt,10pt) circle (1.0ex);
	\draw[black, fill=white] (80pt,40pt) circle (1.0ex);
	\draw[black, fill=white] (80pt,70pt) circle (1.0ex);
	\draw[black, fill=white] (120pt,10pt) circle (1.0ex);
	\draw[black, fill=white] (120pt,40pt) circle (1.0ex);
	\draw[black, fill=white] (120pt,70pt) circle (1.0ex);
	\draw[black, fill=white] (160pt,40pt) circle (1.0ex);
	\draw[black, fill=white] (160pt,70pt) circle (1.0ex);

\node at (0pt, 30pt) {\small$\boldsymbol\times$};
\node at (40pt, 40pt) {\small$\boldsymbol\times$};
\node at (80pt, 70pt) {\small$\boldsymbol\times$};

\end{tikzpicture}

%% file: pics/non-hyperDAG.tikz
\begin{tikzpicture}
    
    \begin{scope}[rounded corners = 5pt]
    \draw[gray, ultra thick, densely dotted] (-13pt,-13pt) -- (-13pt,13pt) -- (63pt,13pt) -- (63pt,-13pt) -- cycle;
    
    \draw[gray, ultra thick, densely dotted] (4.75pt,-17.76pt) -- (-17.75pt,-4.76pt) -- (20.24pt,64.06pt) -- (42.76pt,48.06pt) -- cycle;
    
    \draw[gray, ultra thick, densely dotted] (8.76pt,48.06pt) -- (29.76pt,61.06pt) -- (67.76pt,-4.76pt) -- (45.24pt,-17.76pt) -- cycle;
    \end{scope}
    
    
    \draw[black, fill=white] (0pt,0pt) circle (1.3ex);
    \draw[black, fill=white] (50pt,0pt) circle (1.3ex);
    \draw[black, fill=white] (25pt,43.3pt) circle (1.3ex);

\end{tikzpicture}

%% file: pics/mainth.tikz
\begin{tikzpicture}
	
	\draw[fill=white] (275pt,55pt) -- (275pt,125pt) -- (345pt,125pt) -- (345pt,55pt) -- cycle;
 	\node[anchor=center] at (310pt,97pt) {\large $A'$};
    \node[anchor=center] at (310pt,83pt) {\small \textit{(all red)}};

    \draw[fill=white] (0pt,50pt) -- (0pt,130pt) -- (80pt,130pt) -- (80pt,50pt) -- cycle;
    \node[anchor=center] at (40pt,97pt) {\large $A$};
    \node[anchor=center] at (40pt,83pt) {\small \textit{(all blue)}};

    \draw[black, fill=white] (105pt,65pt) circle (1.3ex);

    \draw[ultra thick] (105pt,65pt) -- (72pt,65pt);
    
    \draw[fill=white] (135pt,50pt) rectangle (155pt,70pt);
    \draw[fill=white] (170pt,50pt) rectangle (190pt,70pt);
    \draw[fill=white] (205pt,50pt) rectangle (225pt,70pt);

    \draw[fill=white] (152pt,80pt) rectangle (172pt,100pt);
    \draw[fill=white] (208pt,80pt) rectangle (188pt,100pt);

    \draw[fill=white] (135pt,110pt) rectangle (155pt,130pt);
    \draw[fill=white] (170pt,110pt) rectangle (190pt,130pt);
    \draw[fill=white] (205pt,110pt) rectangle (225pt,130pt);
	
    
    \draw[rounded corners = 5pt, dashed, gray, ultra thick] (95pt,55pt) -- (143pt,55pt) -- (143pt,70pt) -- (160pt,80pt) -- (160pt,100pt) -- (143pt,110pt) -- (143pt,125pt) -- (125pt,125pt) -- (95pt,75pt) -- cycle;
    
    \begin{scope}[lightgray, very thick]
    
    \draw(132pt,139pt) -- (132pt,142pt) -- (228pt,142pt) -- (228pt,139pt);
    \draw(180pt,142pt) -- (180pt,146pt);

    \draw(99pt,50pt) -- (99pt,47pt) -- (111pt,47pt) -- (111pt,50pt);
    \draw[arrows=-Stealth](105pt,47pt) -- (105pt,42pt) -- (60pt,33pt) -- (60pt,20pt);
    
    \draw[arrows=-Stealth] (125pt,53pt) -- (125pt,42pt) -- (230pt,33pt) -- (230pt,20pt);
    
    \end{scope}
    
    \node[anchor=center] at (180pt,153pt) {\textit{each edge of $G$}};
    \node[anchor=center] at (180pt,163pt) {\textit{separate block for}};

    \node[anchor=center] at (60pt,13pt) {\textit{blue node (connected by}};
    \node[anchor=center] at (60pt,3pt) {\textit{many hyperedges to $A$)}};

    \node[anchor=center] at (230pt,13pt) {\textit{separate hyperedge for each node of $G$, which}};
    \node[anchor=center] at (230pt,3pt) {\textit{intersects into the blocks of incident edges}};

\end{tikzpicture}

%% file: pics/1constraint_limits.tikz
\begin{tikzpicture}

	\draw[ultra thick, gray, dashed] (-8pt,2pt) -- (-8pt,58pt) -- (61pt,58pt) -- (61pt,2pt) -- cycle;
	\draw[ultra thick, gray, dashed] (138pt,2pt) -- (138pt,58pt) -- (69pt,58pt) -- (69pt,2pt) -- cycle;
	
	\node[anchor=center] at (26.5pt,-11pt) {$G_1$};
	\node[anchor=center] at (103.5pt,-11pt) {$G_2$};
	
	
	\draw[arrows=-Stealth] (0pt,20pt) -- (15pt,11pt);
	\draw[arrows=-Stealth] (0pt,20pt) -- (15pt,28pt);
	\draw[arrows=-Stealth] (0pt,40pt) -- (15pt,32pt);
	\draw[arrows=-Stealth] (0pt,40pt) -- (15pt,49pt);
	\draw[arrows=-Stealth] (18pt,10pt) -- (33pt,19pt);
	\draw[arrows=-Stealth] (18pt,30pt) -- (33pt,38pt);
	\draw[arrows=-Stealth] (18pt,50pt) -- (33pt,42pt);
	\draw[arrows=-Stealth] (36pt,20pt) -- (51pt,28pt);
	\draw[arrows=-Stealth] (36pt,40pt) -- (51pt,32pt);
	\draw[arrows=-Stealth] (54pt,30pt) -- (72pt,30pt);
	\draw[arrows=-Stealth] (76pt,30pt) -- (90pt,30pt);
	\draw[arrows=-Stealth] (76pt,30pt) -- (91pt,11pt);
	\draw[arrows=-Stealth] (76pt,30pt) -- (91pt,49pt);
	\draw[arrows=-Stealth] (94pt,50pt) -- (108pt,50pt);
	\draw[arrows=-Stealth] (94pt,10pt) -- (108pt,10pt);
	\draw[arrows=-Stealth] (94pt,10pt) -- (109pt,28pt);
	\draw[arrows=-Stealth] (94pt,30pt) -- (108pt,30pt);
	\draw[arrows=-Stealth] (94pt,30pt) -- (109pt,12pt);
	\draw[arrows=-Stealth] (112pt,10pt) -- (127pt,28pt);
	\draw[arrows=-Stealth] (112pt,30pt) -- (126pt,30pt);
	\draw[arrows=-Stealth] (112pt,50pt) -- (127pt,32pt);
	
	
	\draw[black, fill=white] (0pt,20pt) circle (0.6ex);
	\draw[black, fill=white] (0pt,40pt) circle (0.6ex);
	\draw[black, fill=white] (18pt,10pt) circle (0.6ex);
	\draw[black, fill=white] (18pt,30pt) circle (0.6ex);
	\draw[black, fill=white] (18pt,50pt) circle (0.6ex);
	\draw[black, fill=white] (36pt,20pt) circle (0.6ex);
	\draw[black, fill=white] (36pt,40pt) circle (0.6ex);
	\draw[black, fill=white] (54pt,30pt) circle (0.6ex);
	\draw[black, fill=white] (76pt,30pt) circle (0.6ex);
	\draw[black, fill=white] (94pt,10pt) circle (0.6ex);
	\draw[black, fill=white] (94pt,30pt) circle (0.6ex);
	\draw[black, fill=white] (94pt,50pt) circle (0.6ex);
	\draw[black, fill=white] (112pt,10pt) circle (0.6ex);
	\draw[black, fill=white] (112pt,30pt) circle (0.6ex);
	\draw[black, fill=white] (112pt,50pt) circle (0.6ex);
	\draw[black, fill=white] (130pt,30pt) circle (0.6ex);

\end{tikzpicture}

%% file: pics/DAG_layers.tikz
\begin{tikzpicture}
	
	\draw[thick, arrows=-Stealth] (0pt,30pt) -- (25pt,18pt);
	\draw[thick, arrows=-Stealth] (0pt,30pt) -- (26pt,52pt);
	\draw[thick, arrows=-Stealth] (30pt,55pt) -- (55pt,41pt);
	\draw[thick, arrows=-Stealth] (30pt,55pt) -- (55pt,69pt);
	\draw[thick, arrows=-Stealth] (30pt,55pt) -- (87pt,44.5pt);
	\draw[thick, arrows=-Stealth] (30pt,15pt) -- (117.5pt,35.5pt);
	\draw[thick, arrows=-Stealth] (60pt,40pt) -- (84.5pt,40pt);
	\draw[thick, arrows=-Stealth] (60pt,40pt) -- (87pt,66pt);
	\draw[thick, arrows=-Stealth] (60pt,70pt) -- (84.5pt,70pt);
	\draw[thick, arrows=-Stealth] (90pt,40pt) -- (114.5pt,40pt);
	\draw[thick, arrows=-Stealth] (90pt,70pt) -- (117pt,44pt);
	
	
	\draw[ultra thick, gray, dashed] (0pt,5pt) -- (0pt,80pt);
	\draw[ultra thick, gray, dashed] (30pt,5pt) -- (30pt,80pt);
	\draw[ultra thick, gray, dashed] (60pt,5pt) -- (60pt,80pt);
	\draw[ultra thick, gray, dashed] (90pt,5pt) -- (90pt,80pt);
	\draw[ultra thick, gray, dashed] (120pt,5pt) -- (120pt,80pt);
	
	\node[anchor=center] at (0pt,-7pt) {\large $V_1$};
	\node[anchor=center] at (30pt,-7pt) {\large $V_2$};
	\node[anchor=center] at (60pt,-7pt) {\large $V_3$};
	\node[anchor=center] at (90pt,-7pt) {\large $V_4$};
	\node[anchor=center] at (120pt,-7pt) {\large $V_5$};
	
	
	\draw[black, fill=white] (0pt,30pt) circle (1.0ex);
	\draw[black, fill=white] (30pt,15pt) circle (1.0ex);
	\draw[black, fill=white] (30pt,55pt) circle (1.0ex);
	\draw[black, fill=white] (60pt,40pt) circle (1.0ex);
	\draw[black, fill=white] (60pt,70pt) circle (1.0ex);
	\draw[black, fill=white] (90pt,40pt) circle (1.0ex);
	\draw[black, fill=white] (90pt,70pt) circle (1.0ex);
	\draw[black, fill=white] (120pt,40pt) circle (1.0ex);

\end{tikzpicture}

%% file: pics/layers_limits.tikz
\begin{tikzpicture}
	
	\draw[thick, arrows=-Stealth] (0pt,80pt) -- (27pt,57pt);
	\draw[thick, arrows=-Stealth] (0pt,80pt) -- (26pt,88pt);
	\draw[thick, arrows=-Stealth] (0pt,80pt) -- (27pt,99pt);
	\draw[thick, arrows=-Stealth] (0pt,80pt) -- (27pt,108pt);
	\draw[thick, arrows=-Stealth] (0pt,80pt) -- (27.5pt,118.5pt);
	\draw[thick, arrows=-Stealth] (30pt,55pt) -- (56pt,41.5pt);
	\draw[thick, arrows=-Stealth] (30pt,55pt) -- (55.5pt,51.5pt);
	\draw[thick, arrows=-Stealth] (30pt,55pt) -- (55.5pt,58.5pt);
	\draw[thick, arrows=-Stealth] (30pt,55pt) -- (56pt,68.5pt);
	\draw[thick, arrows=-Stealth] (30pt,90pt) -- (58pt,101pt);
	\draw[thick, arrows=-Stealth] (30pt,100pt) -- (56.5pt,103.5pt);
	\draw[thick, arrows=-Stealth] (30pt,110pt) -- (56.5pt,106.5pt);
	\draw[thick, arrows=-Stealth] (30pt,120pt) -- (58pt,109pt);
	\draw[thick, arrows=-Stealth] (60pt,40pt) -- (88pt,51pt);
	\draw[thick, arrows=-Stealth] (60pt,50pt) -- (86.5pt,53.5pt);
	\draw[thick, arrows=-Stealth] (60pt,60pt) -- (86.5pt,56.5pt);
	\draw[thick, arrows=-Stealth] (60pt,70pt) -- (88pt,59pt);
	\draw[thick, arrows=-Stealth] (60pt,105pt) -- (85.5pt,105pt);
	\draw[thick, arrows=-Stealth] (90pt,55pt) -- (117pt,77pt);
	\draw[thick, arrows=-Stealth] (90pt,105pt) -- (117pt,83pt);
	
	
	\draw[ultra thick, gray, dashed] (0pt,33pt) -- (0pt,127pt);
	\draw[ultra thick, gray, dashed] (30pt,33pt) -- (30pt,127pt);
	\draw[ultra thick, gray, dashed] (60pt,33pt) -- (60pt,127pt);
	\draw[ultra thick, gray, dashed] (90pt,33pt) -- (90pt,127pt);
	\draw[ultra thick, gray, dashed] (120pt,33pt) -- (120pt,127pt);
	
	
	
	\draw[black, fill=white] (0pt,80pt) circle (0.7ex);
	\draw[black, fill=white] (30pt,55pt) circle (0.7ex);
	\draw[black, fill=white] (30pt,90pt) circle (0.7ex);
	\draw[black, fill=white] (30pt,100pt) circle (0.7ex);
	\draw[black, fill=white] (30pt,110pt) circle (0.7ex);
	\draw[black, fill=white] (30pt,120pt) circle (0.7ex);
	\draw[black, fill=white] (60pt,40pt) circle (0.7ex);
	\draw[black, fill=white] (60pt,50pt) circle (0.7ex);
	\draw[black, fill=white] (60pt,60pt) circle (0.7ex);
	\draw[black, fill=white] (60pt,70pt) circle (0.7ex);
	\draw[black, fill=white] (60pt,105pt) circle (0.7ex);
	\draw[black, fill=white] (90pt,55pt) circle (0.7ex);
	\draw[black, fill=white] (90pt,105pt) circle (0.7ex);
	\draw[black, fill=white] (120pt,80pt) circle (0.7ex);

\end{tikzpicture}

%% file: pics/hierarchy.tikz
\begin{tikzpicture}
	
	\draw[thick] (0pt,0pt) -- (11pt,20pt) -- (22pt,0pt);
	\draw[thick] (44pt,0pt) -- (55pt,20pt) -- (66pt,0pt);
	\draw[thick] (88pt,0pt) -- (99pt,20pt) -- (110pt,0pt);
	\draw[thick] (132pt,0pt) -- (143pt,20pt) -- (154pt,0pt);
	
	\draw[thick] (11pt,20pt) -- (33pt,40pt) -- (55pt,20pt);
	\draw[thick] (99pt,20pt) -- (121pt,40pt) -- (143pt,20pt);
	
	\draw[thick] (33pt,40pt) -- (77pt,60pt) -- (121pt,40pt);
	
	
	\draw[ultra thick, gray, arrows=-stealth] (44pt,0pt) arc (-110:-75:32pt);
	\draw[ultra thick, gray, arrows=-stealth] (44pt,0pt) arc (-40:-125:16pt);
	\draw[ultra thick, gray, arrows=-stealth] (44pt,0pt) arc (-173:-16:22pt);
	
	\node[anchor=center] at (55pt,-7.5pt) {$\boldsymbol{1}$};
	\node[anchor=center] at (34pt,-12.5pt) {$\boldsymbol{g_2}$};
	\node[anchor=center] at (66pt,-25.5pt) {$\boldsymbol{g_1}$};
	
	
	\draw[black, fill=white] (0pt,0pt) circle (0.7ex);
	\draw[black, fill=white] (22pt,0pt) circle (0.7ex);
	\draw[black, fill=white] (44pt,0pt) circle (0.7ex);
	\draw[black, fill=white] (66pt,0pt) circle (0.7ex);
	\draw[black, fill=white] (88pt,0pt) circle (0.7ex);
	\draw[black, fill=white] (110pt,0pt) circle (0.7ex);
	\draw[black, fill=white] (132pt,0pt) circle (0.7ex);
	\draw[black, fill=white] (154pt,0pt) circle (0.7ex);
	
	\draw[black, fill=white] (11pt,20pt) circle (0.4ex);
	\draw[black, fill=white] (55pt,20pt) circle (0.4ex);
	\draw[black, fill=white] (99pt,20pt) circle (0.4ex);
	\draw[black, fill=white] (143pt,20pt) circle (0.4ex);
	
	\draw[black, fill=white] (33pt,40pt) circle (0.4ex);
	\draw[black, fill=white] (121pt,40pt) circle (0.4ex);
	
	\draw[black, fill=white] (77pt,60pt) circle (0.4ex);

\end{tikzpicture}

%% file: pics/recursive.tikz
\begin{tikzpicture}
	
	\draw[thick]
	(40pt,20pt) --
	(20pt,65pt) -- (40pt,110pt);
	
	\draw[thick] (81pt,15pt) -- (96pt,33pt) --
	(105pt,53pt) --
	(105pt,77pt) --
	(96pt,97pt) --
	(81pt,115pt);
	
	\draw[fill=white] (10pt,55pt) -- (10pt,75pt) -- (30pt,75pt) -- (30pt,55pt) -- cycle;
	
	\draw[fill=white] (30pt,10pt) -- (30pt,30pt) -- (50pt,30pt) -- (50pt,10pt) -- cycle;
	
	\draw[fill=white] (30pt,100pt) -- (30pt,120pt) -- (50pt,120pt) -- (50pt,100pt) -- cycle;
	
	\draw[fill=white] (76pt,10pt) -- (76pt,20pt) -- (86pt,20pt) -- (86pt,10pt) -- cycle;
	
	\draw[fill=white] (91pt,28pt) -- (91pt,38pt) -- (101pt,38pt) -- (101pt,28pt) -- cycle;
	
	\draw[fill=white] (100pt,48pt) -- (100pt,58pt) -- (110pt,58pt) -- (110pt,48pt) -- cycle;
	
	\draw[fill=white] (100pt,72pt) -- (100pt,82pt) -- (110pt,82pt) -- (110pt,72pt) -- cycle;
	
	\draw[fill=white] (91pt,92pt) -- (91pt,102pt) -- (101pt,102pt) -- (101pt,92pt) -- cycle;
	
	\draw[fill=white] (76pt,110pt) -- (76pt,120pt) -- (86pt,120pt) -- (86pt,110pt) -- cycle;


	\draw[thick]
	(245pt,20pt) --
	(225pt,65pt) -- (245pt,110pt);
	
	\draw[thick] (286pt,15pt) -- (301pt,33pt) --
	(310pt,53pt) --
	(310pt,77pt) --
	(301pt,97pt) --
	(286pt,115pt);
	
	\draw[fill=white] (215pt,55pt) -- (215pt,75pt) -- (235pt,75pt) -- (235pt,55pt) -- cycle;
	
	\draw[fill=white] (235pt,10pt) -- (235pt,30pt) -- (255pt,30pt) -- (255pt,10pt) -- cycle;
	
	\draw[fill=white] (235pt,100pt) -- (235pt,120pt) -- (255pt,120pt) -- (255pt,100pt) -- cycle;
	
	\draw[fill=white] (281pt,10pt) -- (281pt,20pt) -- (291pt,20pt) -- (291pt,10pt) -- cycle;
	
	\draw[fill=white] (296pt,28pt) -- (296pt,38pt) -- (306pt,38pt) -- (306pt,28pt) -- cycle;
	
	\draw[fill=white] (305pt,48pt) -- (305pt,58pt) -- (315pt,58pt) -- (315pt,48pt) -- cycle;
	
	\draw[fill=white] (305pt,72pt) -- (305pt,82pt) -- (315pt,82pt) -- (315pt,72pt) -- cycle;
	
	\draw[fill=white] (296pt,92pt) -- (296pt,102pt) -- (306pt,102pt) -- (306pt,92pt) -- cycle;
	
	\draw[fill=white] (281pt,110pt) -- (281pt,120pt) -- (291pt,120pt) -- (291pt,110pt) -- cycle;
	
	
	\draw[gray, ultra thick, dashed] (63pt,10pt) -- (63pt,120pt);
	
	\draw[gray, ultra thick, dashed] (5pt,65pt) -- (50pt,65pt);
	
	\draw[gray, ultra thick, dashed] (76pt,65pt) -- (121pt,65pt);
	
	\begin{scope}[rounded corners = 5pt, dashed]
    \draw[gray, ultra thick] (230pt,125pt) -- (296pt,125pt) -- (296pt,109pt) -- (252pt,95pt) -- (230pt,95pt) -- cycle;
    
    \draw[gray, ultra thick] (230pt,5pt) -- (296pt,5pt) -- (296pt,21pt) -- (252pt,35pt) -- (230pt,35pt) -- cycle;
    
    \draw[gray, ultra thick] (210pt,80pt) -- (238pt,80pt) -- (296pt,105pt) -- (310pt,105pt) -- (310pt,89pt) -- (301pt,89pt) -- (240pt,50pt) -- (210pt,50pt) -- cycle;
    
    \draw[gray, ultra thick] (302pt,85pt) -- (319pt,85pt) -- (319pt,45pt) -- (309pt,24pt) -- (292pt,24pt) -- (292pt,38pt) -- cycle;
    \end{scope}

\end{tikzpicture}

%% file: pics/twostep.tikz
\begin{tikzpicture}

    \draw[white] (-30pt,0pt) (-30pt,5pt);

    \begin{scope}[lightgray, ultra thick]
    
    \draw (0pt,0pt) rectangle (50pt,50pt);
    \draw (0pt,75pt) rectangle (50pt,125pt);
    \draw (75pt,0pt) rectangle (125pt,50pt);
    \draw (75pt,75pt) rectangle (125pt,125pt);

    \draw (230pt,0pt) rectangle (280pt,50pt);
    \draw (230pt,75pt) rectangle (280pt,125pt);
    \draw (305pt,0pt) rectangle (355pt,50pt);
    \draw (305pt,75pt) rectangle (355pt,125pt);
    
    \end{scope}

    \begin{scope}[very thick]

    \draw (25pt,46pt) -- (25pt,79pt);
    \draw (46pt,34pt) -- (100pt,79pt);
    \draw (46pt,10pt) -- (79pt,10pt);

    \draw (276pt,35pt) -- (309pt,40pt);
    \draw (276pt,25pt) -- (309pt,25pt);
    \draw (276pt,15pt) -- (309pt,10pt);
    
    \end{scope}

    \begin{scope}[thick]

    \draw (4pt,4pt) rectangle (46pt,46pt);
    \node[anchor=center] at (25pt,25pt) {\large $A$};

    \draw (4pt,79pt) rectangle (46pt,91pt);
    \draw (79pt,4pt) rectangle (121pt,16pt);
    \draw (79pt,79pt) rectangle (121pt,91pt);
    \node[anchor=center] at (25pt,85pt) {\large $B_1$};
    \node[anchor=center] at (100pt,85pt) {\large $B_2$};
    \node[anchor=center] at (100pt,10pt) {\large $B_3$};

    \draw (234pt,4pt) rectangle (276pt,46pt);
    \node[anchor=center] at (255pt,25pt) {\large $A$};

    \draw (309pt,4pt) rectangle (351pt,16pt);
    \draw (309pt,19pt) rectangle (351pt,31pt);
    \draw (309pt,34pt) rectangle (351pt,46pt);
    \node[anchor=center] at (330pt,40pt) {\large $B_1$};
    \node[anchor=center] at (330pt,25pt) {\large $B_2$};
    \node[anchor=center] at (330pt,10pt) {\large $B_3$};

    \end{scope}

    \begin{scope}[gray, arrows=-Stealth, dashed]

    \draw (154pt,135pt) -- (25pt,109pt);
    \draw (160pt,135pt) -- (100pt,109pt);
    \draw (164pt,135pt) -- (100pt,34pt);

    \draw (368pt,135pt) -- (255pt,100pt);
    \draw (372pt,135pt) -- (330pt,100pt);
    
    \end{scope}

    \node[anchor=center] at (160pt,151pt) {\textit{remaining}};
    \node[anchor=center] at (160pt,142pt) {\textit{blocks}};

    \node[anchor=center] at (370pt,151pt) {\textit{remaining}};
    \node[anchor=center] at (370pt,142pt) {\textit{blocks}};

    \node[anchor=center] at (62.5pt,-24pt) {\textit{optimum for regular partitioning}};
    \node[anchor=center] at (62.5pt,-35pt) {\textit{(but high hierarchical cost)}};

    \node[anchor=center] at (292.5pt,-24pt) {\textit{suboptimal for regular partitioning}};
    \node[anchor=center] at (292.5pt,-35pt) {\textit{(but low hierarchical cost)}};

\end{tikzpicture}